\documentclass[11pt]{article}
\usepackage{amsmath, mathtools,amssymb,bbm}
\usepackage[margin=1in]{geometry}
\usepackage{thm-restate}
\usepackage{thmtools}
\usepackage[numbers,comma,sort&compress]{natbib}
\usepackage{subcaption}
\usepackage{authblk}

\usepackage[english]{babel}
\usepackage[utf8]{inputenc}
\usepackage{tcolorbox}
\usepackage{amsthm}
\theoremstyle{definition}

\usepackage{caption}

\usepackage{algorithmicx}
\usepackage{algorithm}
\usepackage[noend]{algpseudocode}

\usepackage{xcolor,soul}
\usepackage{qcircuit}
\usepackage[colorinlistoftodos, color=green!40, prependcaption]{todonotes}
\usepackage{blindtext} 
\usepackage[pdftex, pdftitle={Article}, pdfauthor={Author}, linktocpage]{hyperref} 
\hypersetup{
    colorlinks=true,
    citecolor=magenta,
    linkcolor=blue,
    filecolor=magenta,
    urlcolor=black,
}
\newtheorem{lemma}{Lemma}

\newtheorem{proposition}{Proposition}
\newtheorem{theorem}{Theorem}
\newtheorem{corollary}{Corollary}

\DeclareMathOperator{\erf}{erf}
\DeclareMathOperator{\erfc}{erfc}
\DeclareMathOperator{\sgn}{sgn}
\DeclareMathOperator{\polylog}{polylog}

\newcommand{\ceil}[1]{\lceil #1 \rceil}
\newcommand{\floor}[1]{\lfloor #1 \rfloor}

\newcommand{\cbox}[1]{\measure{\mbox{$#1$}}}

\definecolor{mypurple}{RGB}{164,64,214}
\definecolor{mycyan}{RGB}{0, 191, 255}

\definecolor{mypurple}{RGB}{164,64,214}

\usetikzlibrary{decorations.pathmorphing}

\newcommand{\ket}[1]{| #1 \rangle}
\newcommand{\bra}[1]{\langle #1 |}

\newcommand{\sel}[1][]{\text{\textsc{sel}$_{#1}$--}}

\begin{document}
\title{Fast digital methods for adiabatic state preparation}

\author[1,2]{Kianna Wan\thanks{\small \texttt{kianna@stanford.edu}}}
\affil[1]{\textit{Stanford Institute for Theoretical Physics, Stanford University, Stanford, CA 94305, USA}}
\author[2]{\normalsize Isaac H. Kim\thanks{\small\texttt{isaac.kim@sydney.edu.au} \\
This research was conducted during the first author's internship at PsiQuantum, which was supervised by the second author. The second author's current affiliation is the Department of Computer Science, UC Davis, Davis, California 95616, USA}}
\affil[2]{\textit{PsiQuantum, Palo Alto, CA 94304, USA}}

\date{\today} 
\maketitle
\begin{abstract}
We present a quantum algorithm for adiabatic state preparation on a gate-based quantum computer, with complexity polylogarithmic in the inverse error. 
Our algorithm digitally simulates the adiabatic evolution between two self-adjoint operators $H_0$ and $H_1$, exponentially suppressing the diabatic error by harnessing the theoretical concept of quasi-adiabatic continuation as an algorithmic tool. Given an upper bound $\alpha$ on $\|H_0\|$ and $\|H_1\|$ along with the promise that the $k$th eigenstate $\ket{\psi_k(s)}$ of $H(s) \coloneqq (1-s)H_0 + sH_1$ is separated from the rest of the spectrum by a gap of at least $\gamma > 0$ for all $s \in [0,1]$, this algorithm implements an operator $\widetilde{U}$ such that $\|\ket{\psi_k(1)} - \widetilde{U}\ket{\psi_k(s)}\| \leq \epsilon$ using $\mathcal{O}\left(({\alpha^2}/{\gamma^2})\polylog({\alpha}/{\gamma\epsilon})\right)$ queries to block-encodings of $H_0$ and $H_1$.
In addition, we develop an algorithm that is applicable only to ground states and requires multiple queries to an oracle that prepares $\ket{\psi_0(0)}$, but has slightly better scaling in all parameters. We also show that the costs of both algorithms can be further reduced under certain reasonable conditions, such as when $\|H_1 - H_0\|$ is small compared to $\alpha$, or
when more information about the gap of $H(s)$ is available. For certain problems, the scaling can even be improved to linear in $\|H_1 - H_0\|/\gamma$ up to polylogarithmic factors. 
 
\end{abstract}

\newpage
\tableofcontents
\newpage
\section{Introduction}
One of the most anticipated applications of quantum computation is the simulation of quantum systems. Lloyd was the first to formalise the concept of a quantum simulator, providing an explicit algorithm for efficiently simulating the dynamics of systems with local interactions~\cite{Lloyd1996}. Since then, significant progress has been made towards making quantum simulation more practical in different domains, resulting in promising approaches for studying quantum chemistry~\cite{Reiher2017,Babbush2018}, strongly correlated electron models~\cite{Wecker2015}, and even quantum field theory~\cite{Jordan2011}. 

Many of these advances can be attributed to innovative developments in Hamiltonian simulation~\cite{Childs2012,Berry2013,Berry2015,Low2017,Low2018,Low2019,Childs2019,Kieferova2019}. These works, however, do not directly address the important issue of initial state preparation. 
Algorithms for studying physical properties typically require as input a quantum state of physical interest, such as a low-energy eigenstate or a thermal state. Such states cannot be prepared efficiently in general~\cite{Knill1995,Kitaev2002}. Indeed, it is difficult to avoid a factor that scales exponentially with the system size in the worst case. For instance, the complexity of existing methods for ground state preparation~\cite{Ge2019,Lin2020} is inversely proportional to the overlap between the initial state and the target ground state, which would incur an exponential cost in some settings. 

Though inefficient in general, state preparation may nevertheless be feasible under certain conditions. In this paper, we consider the case where the target state is an eigenstate of a Hamiltonian $H_1$ that is adiabatically connected to another Hamiltonian $H_0$. If the corresponding eigenstate of $H_0$ can be efficiently prepared, one can then approximately obtain the eigenstate of $H_1$ by evolving the system under a time-dependent Hamiltonian that slowly interpolates between $H_0$ and $H_1$. This approach, known as adiabatic state preparation, is often the method of choice for initial state preparation in quantum simulation~\cite{Reiher2017,Wecker2015}. 

Prior studies have estimated the costs of adiabatic state preparation algorithms for various problems~\cite{Reiher2017,Wecker2015}. Using the best rigorous bounds for the Trotter-Suzuki decomposition~\cite{Suzuki1993,Wiebe2010,Wiebe2011} that were known at the time, \cite{Reiher2017} concluded that the cost of adiabatically preparing the ground state of a typical quantum chemistry Hamiltonian would scale as $\mathcal{O}({N^{12+o(1)}}/{\gamma^2})$, where $N$ is the number of spin-orbitals and $\gamma$ is a lower bound on the spectral gap of the interpolating Hamiltonian. Under more optimistic assumptions, it was argued that $\mathcal{O}({N^{5.5}}/{\gamma^3})$ scaling could be achieved. Though not explicitly analysed in~ \cite{Reiher2017}, it is obvious that the error dependence of their algorithms is polynomial in $1/\epsilon$, where $\epsilon$ is the target error, since they are based on product formulae. 


In this paper, we provide a unified treatment of adiabatic state preparation that improves upon prior works. Our main contribution is a quantum algorithm (Algorithm~\ref{alg1}) that simulates an adiabatic evolution between $H_0$ and $H_1$ on a digital quantum computer. The complexity of this algorithm is polylogarithmic in $1/\epsilon$ and nearly optimal in all of the other relevant parameters as well; see Theorem~\ref{thm:main1}. To achieve this scaling, we adapt the idea of quasi-adiabatic continuation, which was originally introduced by Hastings~\cite{Hastings2004_LSM} as a purely theoretical construct. We emphasise that our approach is conceptually very different from algorithms based on the adiabatic theorem. In the latter, one would simulate the time evolution governed by a time-dependent Hamiltonian $H(s)$ that interpolates between $H_0$ and $H_1$, and in order to suppress the diabatic error, the total evolution time must be chosen to be sufficiently large. Based on currently known error bounds for the adiabatic theorem, the complexity of such algorithms would scale at best subpolynomially (but not polylogarithmically) with $1/\epsilon$~\cite{Jansen2007,Wiebe2012}. Instead, using the quasi-adiabatic continuation framework, we (approximately) implement the time evolution generated by a self-adjoint operator of the form $\int dt\, f(t) e^{iH(s)t} H'(s) e^{-iH(s)t}$ for a constant time, suppressing the diabatic error via a judicious choice of the function $f$. This leads to a polylogarithmic dependence on $1/\epsilon$, \emph{without} using any filtering-based techniques~\cite{Poulin2009,Ge2019,Lin2020}.

In addition, we provide a different algorithm that specifically prepares ground states, by integrating the approach of Ref.~\cite{Lin2020} with digital adiabatic state preparation, making some necessary adjustments. This algorithm (see Algorithm~\ref{alg2} and Theorem~\ref{thm:adiabatic_ground}) has slightly better scaling than our main algorithm, but a key difference lies in the access model to the initial state. To prepare the $k$th eigenstate of $H_1$, Algorithm~\ref{alg1} requires only a \emph{single} \emph{copy} of the $k$th eigenstate of $H_0$. On the other hand, Algorithm~\ref{alg2} assumes access to a unitary oracle that prepares the ground state of $H_0$ from some fixed state, such as the all-zeros state, and uses multiple queries to this oracle.

The complexities of both algorithms can be substantially improved by taking into account extra information that may be available in practical scenarios. For example, the complexities scale with $\|H_1 - H_0\|$, which can be trivially upper-bounded by $\|H_0\| + \|H_1\|$, but may be known to be much smaller in some cases, e.g., if $H_1$ is a small perturbation of $H_0$. We also show that the number of queries and gates can be significantly reduced if a better lower bound on the gap of $H(s)$ is known for different values of $s$, in the same way that the runtime of adiabatic algorithms on an analog computer can be reduced by varying the rate of change according to the instantaneous gap~\cite{Roland2002}.  

The rest of this paper is structured as follows. In Section~\ref{sec:overview}, we provide a high-level overview of the main results and key techniques. We then describe and analyse our algorithm for preparing general eigenstates in Sections~\ref{sec:qac} through \ref{sec:adiabatic_general}. Specifically, in Section~\ref{sec:qac}, we construct a discrete approximation to quasi-adiabatic continuation that is amenable to simulation using digital methods. In Section~\ref{sec:oracle_synthesis}, we synthesise a circuit that encodes the generator of this discretised quasi-adiabatic continuation. By combining the results of Section~\ref{sec:oracle_synthesis} with known Hamiltonian simulation methods, we bound the complexity of our main algorithm in Section~\ref{sec:adiabatic_general}. In Section~\ref{sec:improved_gap_dependence}, we show how this algorithm can be modified to take advantage of additional gap information, if it is available. In Section~\ref{sec:adiabatic_ground}, we propose an alternative algorithm for ground state preparation. We discuss open problems and future directions in Section~\ref{sec:conclusion}.

\section{Summary \label{sec:overview}}

Our main result is a quantum algorithm that prepares an eigenstate of a Hamiltonian $H_1$ from the corresponding eigenstate of an adiabatically connected Hamiltonian $H_0$. More specifically, consider the one-parameter family of Hamiltonians $H(s) \coloneqq (1-s)H_0 + sH_1$ for $s \in [0,1]$, and let $\ket{\psi_k(s)}$ denote the $k$th eigenstate of $H(s)$.
Our algorithm implements an operator $\widetilde{U}$ such that
\begin{equation} \label{eq:main_goal}
    \left\| \ket{\psi_k(1)} - \widetilde{U}\ket{\psi_k(0)} \right\| \leq \epsilon, 
\end{equation}
assuming that $\ket{\psi_k(s)}$ is gapped from the rest of the spectrum at least $\gamma > 0$. 
We also provide an alternative algorithm for preparing ground states.

We begin by defining our input model in subsection~\ref{sec:inputmodel}. In subsection~\ref{sec:main}, we state the complexities of our algorithms in terms of $\epsilon$, $\gamma$, as well as the spectral norms of $H_0$, $H_1$, and $H_1 - H_0$. The underlying techniques are discussed in subsection~\ref{sec:techniques}. 

\subsection{Input model} \label{sec:inputmodel} We assume access to unitary oracles $O_{H_0}, O_{H_1} \in \mathbb{C}^{2^{n_a+n_s} \times 2^{n_a+n_s}}$ such that
\begin{align*}
&(\bra{0}_a\otimes I_s)O_{H_0}(\ket{0}_a\otimes I_s) = \frac{H_0}{\alpha}, \\
&(\bra{0}_a\otimes I_s)O_{H_1}(\ket{0}_a\otimes I_s) = \frac{H_1}{\alpha},
\end{align*}
with the promise that $\alpha \geq \max\{\|H_0\|,\|H_1\|\}$. We also require a unitary oracle $O_{H'} \in \mathbb{C}^{2^{n_b+n_s} \times 2^{n_b+n_s}}$ that encodes $H' \coloneqq H_1 - H_0$ as
\begin{equation*} \label{O_H'} \bra{0}_b\otimes I_s)O_{H'}(\ket{0}_b\otimes I_s) = \frac{H'}{\beta}, 
\end{equation*}
where $\beta$ is an upper bound on $\|H'\|$. $O_{H'}$ could either be provided as a black box that is available independently of $O_{H_0}$ and $O_{H_1}$, or constructed using $O_{H_0}$ and $O_{H_1}$ as shown in Appendix~\ref{sec:oracle_relation}. In the terminology of~\cite{Chakraborty2018}, $O_{H_0}$, $O_{H_1}$, and $O_{H'}$ are \textit{block-encodings} of $H_0$, $H_1$, and $H'$ respectively. The block-encoding framework, first introduced by~\cite{Low2019}, is motivated by the fact that it encompasses several widely applicable input models, including standard techniques for encoding sparse Hamiltonians and for encoding Hamiltonians that are specified as a linear combination of unitaries.

For our ground state preparation procedure (Algorithm~\ref{alg2}), we also assume access to a unitary oracle $G_0 \in \mathbb{C}^{2^{n_s} \times 2^{n_s}}$ that prepares the ground state $\ket{\psi_0(0)}$ of $H_0$ from the all-zeros state, i.e., $G_0 \ket{0} = \ket{\psi_0(0)}$. This oracle is not required by our main algorithm (Algorithm~\ref{alg1}), which is applicable to general eigenstates but has slightly worse scaling than Algorithm~\ref{alg2} in the case of ground states.

In quantifying the costs of our algorithms, we assume that the query complexity to a controlled version of a unitary oracle is the same as that to the original oracle, as long as the control register consists of a constant number of qubits. This is justified by the fact that adding a constant number of controls to a unitary does not affect the asymptotic gate complexity. In the worst case, each gate in the circuit implementation of the unitary can be replaced by its controlled version, incurring a constant multiplicative overhead. 

\subsection{Main results} \label{sec:main}

The majority of this paper is concerned with constructing an algorithm that implements $\widetilde{U}$ satisfying Eq.~\eqref{eq:main_goal} for any $k$, provided that the $k$th eigenstate of $H(s)$ is separated from the rest of the spectrum. We briefly summarise this algorithm below; the details are developed in Sections~\ref{sec:qac}--\ref{sec:adiabatic_general}.

\begin{algorithm}[H]  \caption{Approximately prepares the $k$th eigenstate of $H_1$ from the $k$th eigenstate of $H_0$ \label{alg1}}
\textbf{Inputs:} Unitary oracles $O_{H_0}$, $O_{H_1}$, and $O_{H'}$ (defined as in subsection~\ref{sec:inputmodel}), an upper bound on $\max\{\|H_0\|, \|H_1\|\}$, an upper bound $\beta$ on $\|H_1 - H_0\|$, a lower bound $\gamma$ on the gap(s) separating the $k$th eigenstate of $(1-s)H_0 + sH_1$ from the rest of the spectrum, a target precision $\epsilon > 0$.
\begin{algorithmic}[1]
\State Apply the truncated Dyson series algorithm of~\cite{Low2018} using the operators $O_{\widetilde{D}}^{(j)}$ constructed in Lemma~\ref{lem:oracle} as the oracle inputs, with the parameters $\Delta$, $T$, and $N$ chosen according to Corollary~\ref{corollary:parameters}.
\end{algorithmic}
\end{algorithm}

\noindent The query and gate complexities of Algorithm~\ref{alg1} scale with $\alpha$, $\beta$, and $\gamma$ as $\widetilde{\mathcal{O}}(\alpha\beta/\gamma^2)$ (where $\widetilde{O}$ hides logarithmic factors), and depend polylogarithmically on $\epsilon$. The precise scaling of the query and gate complexities as well as the space overhead is given in the following theorem, proved in Section~\ref{sec:adiabatic_general}. Here and throughout the paper, we refer to arbitrary one- and two-qubit gates as ``elementary gates," and we use $\mathcal{M}(b)$ to denote the complexity of $b$-bit multiplication.

\begin{restatable}[Preparation of general eigenstates]{theorem}{adiabatic}
\label{thm:main1} 
For self-adjoint operators $H_0, H_1 \in \mathbb{C}^{2^{n_s} \times 2^{n_s}}$, let $\ket{\psi_k(0)}$ and $\ket{\psi_k(1)}$ denote the $k$th eigenstates of $H_0$ and $H_1$, respectively. Given the promise that $\alpha \geq \max\{\|H_0\|,\|H_1\|\}$ and $\beta \geq \|H_1 - H_0\|$, and that for all $s \in [0,1]$, the $k$th eigenstate of $(1-s)H_0 + sH_1$ is non-degenerate and separated from the rest of the spectrum by a gap of at least $\gamma > 0$, an operator $\widetilde{U}$ can be implemented such that  
\[ \left\|\ket{\psi_k(1)} - \widetilde{U}\ket{\psi_k(0)}\right\| \leq \epsilon \]
with probability $1-\mathcal{O}(\epsilon)$ using
\begin{itemize}
\item $\mathcal{O}\left(\frac{\beta}{\gamma}\left[\frac{\alpha}{\gamma} + \log\left(\frac{1}{\epsilon}\right)\right]\frac{\log^{2.5}\left(\frac{\beta}{\gamma\epsilon}\right)}{\log\log\left(\frac{\beta}{\gamma\epsilon}\right)}\right)$ queries to $O_{H_0}$ and $O_{H_1}$, $\mathcal{O}\left(\frac{\beta}{\gamma}\frac{\log^{1.5}\left(\frac{\beta}{\gamma\epsilon}\right)}{\log\log\left(\frac{\beta}{\gamma\epsilon}\right)}\right)$ queries to $O_{H'}$,
\item $\mathcal{O}\left(\frac{\beta}{\gamma}\left[\frac{\alpha}{\gamma} + \log\left(\frac{1}{\epsilon}\right)\right]\left[n_a + n_b + \log\left(\frac{\alpha\beta}{\gamma^2\epsilon}\right)\mathcal{M}\left(\log\left(\frac{\alpha\beta}{\gamma^2\epsilon}\right)\right)\right]\frac{\log^{2.5}\left(\frac{\beta}{\gamma\epsilon}\right)}{\log\log\left(\frac{\beta}{\gamma\epsilon}\right)}\right)$ elementary gates, and 
\item $n_s + n_a + \mathcal{O}\left(n_b + \log^2\left(\frac{\alpha\beta}{\gamma^2\epsilon}\right)\right)$ qubits,
\end{itemize}
where $O_{H_0}, O_{H_1} \in \mathbb{C}^{2^{n_a + n_s} \times 2^{n_a + n_s}}$ and $O_{H'} \in \mathbb{C}^{2^{n_b + n_s} \times 2^{n_b + n_s}}$ are defined as in subsection~\ref{sec:inputmodel}. 
\end{restatable}

\noindent Although this theorem is stated specifically in terms of an operator that approximately maps the $k$th eigenstate of $H_0$ to that of $H_1$, it immediately implies a more general result. For $s \in [0,1]$, let $\ket{\psi_k(s)}$ denote the $k$th eigenstate of $H(s)$. Observing that a block-encoding $O_{H(s)}$ of $H(s)$ can be constructed for any $s$ using one query to each of $O_{H_0}$ and $O_{H_1}$ [cf.~Appendix~\ref{sec:oracle_linear_combination}], we can implement for any $0 \leq s_0 < s_1 \leq 1$ an operator $\widetilde{U}(s_0, s_1)$ such that $\|\ket{\psi_k(s_1)} - \widetilde{U}(s_0,s_1)\ket{\psi_k(s_0)}\| \leq \epsilon$ by choosing $H_0$ and $H_1$ in Theorem~\ref{thm:main1} to be $H(s_0)$ and $H(s_1)$, respectively. In this case, $\beta$ would be rescaled to $(s_1 - s_0)\beta$. 

This ability to simulate the adiabatic evolution over any segment $[s_0,s_1] \subset [0,1]$ can be leveraged to obtain better scaling in $\gamma$, if more information about the gap of $H(s)$ is available, i.e., if we are given the promise that for all $s \in [0,1]$, the gap of $H(s)$ is bounded by $\gamma(s)$ for some function $\gamma(s)$ that is not necessarily uniform. The new scaling will of course depend on the behaviour of $\gamma(s)$.
In Section~\ref{sec:improved_gap_dependence}, we make this dependence more precise by quantifying the relevant properties of $\gamma(s)$, using a result of~\cite{Jarret2018}. We then show that under a reasonable set of assumptions on $\gamma(s)$, the query and gate complexities in Theorem~\ref{thm:main1} can be improved to $\widetilde{\mathcal{O}}(\beta/\gamma)$ (where $\gamma\coloneqq \min_{s\in[0,1]}\gamma(s)$), shaving off a factor of $\alpha/\gamma$ while maintaining polylogarithmic scaling in $1/\epsilon$; see Theorem~\ref{thm:improved_gap}. This represents an exponential improvement in precision over previous work on optimising gap dependence~\cite{Jarret2018}. The method of~\cite{Jarret2018} achieves $\beta/\gamma$ scaling (under the same assumptions), but scales with the precision as $1/\epsilon$ since it is based on the adiabatic theorem.

For ground states, we can improve the scaling of Algorithm~\ref{alg1} in all of the parameters, at the expense of making multiple queries to an extra unitary oracle $G_0$ that prepares the ground state $\ket{\psi_0(0)}$ of $H_0$ from $\ket{0}$. (In contrast, if we were to use Algorithm~\ref{alg1} to prepare the ground state of $H_1$ starting from $\ket{0}$, only one call to $G_0$ would be required.) 

\begin{algorithm}[H]  \caption{Approximately prepares the ground state of $H_1$ \label{alg2}}
\textbf{Inputs:} Unitary oracles $O_{H_0}$, $O_{H_1}$, $O_{H'}$, and $G_0$ (defined as in subsection~\ref{sec:inputmodel}), an upper bound $\alpha$ on $\max\{\|H_0\|, \|H_1\|\}$, an upper bound $\beta$ on $\|H_1 - H_0\|$, a lower bound $\gamma$ on the spectral gap of $(1-s)H_0 + sH_1$, a lower bound $\gamma_1$ on the spectral gap of $H_1$, a target precision $\epsilon > 0$.
\begin{algorithmic}[1]
\State Estimate the ground state energy of $H_1$ (Section~\ref{sec:adiabatic_ground}).
\State Prepare a state that has constant overlap with the ground state of $H_1$, using Lemma~\ref{lemma:adiabatic_constant_error}.
\State Apply eigenstate filtering~\cite{Lin2019,Lin2020}.
\end{algorithmic}
\end{algorithm}

\noindent The complexity of Algorithm~\ref{alg2} is established by the following theorem, proved in Section~\ref{sec:adiabatic_ground}.
\begin{restatable}[Ground state preparation]{theorem}{adiabaticg}
\label{thm:adiabatic_ground}
For self-adjoint operators $H_0, H_1 \in \mathbb{C}^{2^{n_s} \times 2^{n_s}}$, let it be promised that $\alpha \geq \max\{\|H_0\|, \|H_1\| \}$ and $\beta \geq \|H_1-H_0\|$, and that for all $s\in[0,1]$, the ground state of $(1-s)H_0 + sH_1$ is non-degenerate and separated from the rest of the spectrum by a gap of at least $\gamma > 0$. Let $\ket{\psi_1(0)}$ denote the ground state of $H_1$, and suppose that $\gamma_1 (\geq \gamma)$ is a lower bound on the spectral gap of $H_1$. Then, a state $\ket{\widetilde{\psi}_0(1)}$ such that
\begin{equation*}
    \left\| \ket{\psi_0(1)} - \ket{\widetilde{\psi}_0(1)}\right\|\leq \epsilon
\end{equation*}
can be prepared
with probability $1- \mathcal{O}(\epsilon)$ using
\begin{itemize}
    \item $\mathcal{O}\left(\frac{\alpha\beta}{\gamma^2}\log\left(\frac{\alpha}{\gamma_1}\right) \log\left(\frac{\log(\alpha/\gamma_1)}{\epsilon}\right)\frac{\log(\alpha\beta/\gamma^2)}{\log\log(\alpha\beta/\gamma^2)}\right)$ queries to $O_{H_0}, O_{H_1}$, $\mathcal{O}\left(\log\left(\frac{\alpha}{\gamma_1}\right) \log\left(\frac{\log(\alpha/\gamma_1)}{\epsilon}\right)\right)$ queries to $G_0$ 
    \item $\mathcal{O}\left(\frac{\alpha\beta}{\gamma^2}\left[n_a + \log\left(\frac{\alpha\beta}{\gamma^2}\right)\mathcal{M}\left(\log\left(\frac{\alpha\beta}{\gamma^2}\right)\right) \right]\log\left(\frac{\alpha}{\gamma_1}\right) \log\left(\frac{\log(\alpha/\gamma_1)}{\epsilon}\right)\frac{\log({\alpha\beta}/{\gamma^2})}{\log\log({\alpha\beta}/{\gamma^2})}\right)$ elementary gates, and
    \item $n_s + \mathcal{O}\left(n_a + \log^2\left(\frac{\alpha\beta}{\gamma^2}\right)\right)$ qubits,
\end{itemize}
where $O_{H_0},O_{H_1} \in \mathbb{C}^{2^{n_a+n_s} \times 2^{n_a+n_s}}$, and $G_0 \in \mathbb{C}^{n_s \times n_s}$ are defined as in subsection~\ref{sec:inputmodel}.
\end{restatable}
\noindent
Given more information about the gap of $H(s)$, the scaling in Theorem~\ref{thm:adiabatic_ground} can also be improved to $\widetilde{\mathcal{O}}(\beta/\gamma)$ in some cases, using a similar idea to that behind Theorem~\ref{thm:improved_gap}.

\subsection{Techniques for Algorithm~\ref{alg1}}
\label{sec:techniques}
In this subsection, we summarise the  techniques that we use to construct our main algorithm.

\subsubsection{Quasi-adiabatic continuation}
\label{sec:technique_qac}
Algorithm~\ref{alg1} prepares eigenstates of a Hamiltonian $H_1$ by simulating the adiabatic evolution corresponding to $H(s) \coloneqq (1-s)H_0 + sH_1$. To obtain polylogarithmic scaling in the target precision, we exploit the machinery of quasi-adiabatic continuation~\cite{Hastings2004_LSM,Osborne2007}. The main idea is to approximately implement the unitary $U(s)$ generated by a quasi-adiabatic continuation operator
\[ D(s) \coloneqq \int_{-\infty}^\infty dt\, W(t)e^{iH(s)t}H' e^{-iH(s)t}, \]
where $H' \coloneqq dH(s)/ds = H_0 - H_1$ in this case, 
and $W$ is an odd function with $W(t)\geq 0$ for $t\geq 0$. There are two factors to consider when choosing the specific form of $W$. On the one hand, certain choices of $W$ result in better bounds on the error $\|\ket{\psi_k(s)} - U(s)\ket{\psi_k(0)}\|$. In fact, there exist functions $W$ for which $U(s)$ describes the adiabatic evolution \emph{exactly}, i.e., $U(s)\ket{\psi_k(0)} = \ket{\psi_k(s)}$~\cite{Osborne2007,Bachmann2012}. On the other hand, these functions are difficult to approximate (as far as we know), and the complexity of implementing $U(s)$ depends in part on the efficiency of computing integrals of $W$. In this paper, we use a function both leads to good error bounds (as shown in subsection~\ref{sec:3.1}) and is easy to integrate (as shown in Appendix~\ref{sec:prepareW}). 

\subsubsection{Oracle construction} \label{sec:technique_oracle_synthesis}

Most of the technical work involved in designing and analysing Algorithm~\ref{alg1} concerns the circuit implementation of an oracle $O_D$ that approximately block-encodes $D(s)$, as
\begin{equation} \label{O_D} (\bra{0}\otimes I)O_D(\ket{0}\otimes I) \approx \sum_s \ket{s}\bra{s}\otimes D(s). \end{equation}
Here, the sum is over a finite set of values of $s \in [0,1]$ (depending on the input parameters $\alpha$, $\beta$, $\gamma$, and $\epsilon$). The purpose of Section~\ref{sec:oracle_synthesis} and Appendices~\ref{specialfunctions} and~\ref{sec:prepareW} is to build a circuit for $O_D$ from our ``lowest-level" oracles $O_{H_0}$, $O_{H_1}$, and $O_{H'}$, in such a way that the number of queries, additional elementary gates, and ancilla qubits required is polylogarithmic in the inverse error. Once $O_D$ is constructed, we can use it as the oracle input to the time-dependent Hamiltonian simulation algorithm of~\cite{Low2018} to simulate $U(s)$ (where the ``Hamiltonian" in this case is the $s$-dependent self-adjoint operator $D(s)$). 

We provide a very rough sketch of how the oracle $O_D$ is synthesised, leaving the details to later sections. The starting point is the observation that for a fixed value of $s$,
\[ D(s) \approx (\bra{W} \otimes I)\left(\sum_t \ket{t}\bra{t}\otimes e^{iH(s)t} H'e^{-iH(s)t} \right)(\ket{W'} \otimes I) \]
up to normalisation, 
where $\ket{W} \propto \sum_t \sqrt{\int_{t}^{t+\delta}dt\,|W(t)|}\ket{t}$ approximately encodes $\int dt\, W(t)$ (with finite limits of integration), and $\ket{W'}$ differs from $\ket{W}$ only in that the amplitudes corresponding to $t<0$ are negative, to account for the fact that $W(t)$ is an odd function. The error in the approximation can be made arbitrarily small by decreasing the step size $\delta$ and increasing the range of the sum [cf.~subsection~\ref{sec:3.2}]. 
Thus, a block-encoding $O_{D(s)}$ of $D(s)$ for any $s$ can be constructed by combining three different components:
\begin{enumerate}
    \item A unitary that prepares $\ket{W}$ from the all-zeros states. While the cost of preparing states with arbitrary coefficients scales in general with the number of basis states, our choice of the function $W$ ensures that $\ket{W}$ can be approximated efficiently (with complexity polylogarithmic in the number of basis states) via the standard approach of Grover and Rudolph~\cite{Grover2002} [cf.~Appendix~\ref{sec:prepareW}]. By using a suitable encoding for the basis states, $\ket{W'}$ is trivially obtained from $\ket{W}$.
    \item A block-encoding $O_{H'}$ of $H'$, which is one of the inputs to our algorithm [cf.~subsection~\ref{sec:inputmodel}]
    \item The multiply-controlled operation $\sum_t \ket{t}\bra{t} \otimes e^{-iH(s)t}$ (and its inverse), which applies, to the second register, the time-evolution operator $e^{-iH(s)t}$ corresponding to $H(s)$ for a time $t$ conditioned on the first register being in the state $\ket{t}$. We implement the control logic by using the same gate sequence as that in quantum phase estimation -- the states $\ket{t}$ are encoded in binary, and controlled-$e^{-iH(s)2^i\delta}$ is controlled on qubit $i$ for each $i$. Each of these controlled time-evolution operators can be implemented using the qubitisation algorithm of Refs.~\cite{Low2017,Low2019}. 
\end{enumerate}

Next, we need to further condition the block-encodings $O_{D(s)}$ on an additional register that encodes discrete values of $s$, resulting in an operator of the form $O_D = \sum_{s} \ket{s}\bra{s} \otimes O_{D(s)}$. Since each $O_{D(s)}$ block-encodes $D(s)$ (for a particular $s$), $O_D$ satisfies Eq.~\eqref{O_D}. Note that of the three components of $O_{D(s)}$ discussed above, only the third depends on the parameter $s$. Hence, it suffices to control the operator $\sum_t \ket{t}\bra{t} \otimes e^{-iH(s)t}$ on the $s$ register, i.e., construct an operator of the form $\sum_s \sum_t \ket{s}\bra{s}\otimes \ket{t}\bra{t} \otimes e^{-iH(s)t}$. To do so, we note that the qubitisation algorithm simulates $e^{-iH(s)t}$ (for some $t$) by making queries to a block-encoding $O_{H(s)}$ of $H(s)$. It follows that applying the qubitisation algorithm to the multiply-controlled block-encoding $\sum_s \ket{s}\bra{s} \otimes O_{H(s)}$ simulates $\sum_s \ket{s}\bra{s}\otimes e^{-iH(s)t}$. This can then be controlled on the $t$ register as explained above. We construct efficient circuits that approximately implement $\sum_s \ket{s}\bra{s} \otimes O_{H(s)}$ and $\sum_s \sum_t \ket{s}\bra{s}\otimes \ket{t}\bra{t}\otimes e^{-iH(s)t}$ (using queries to $O_{H_0}$ and $O_{H_1}$) in subsections~\ref{sec:selUH} and~\ref{sec:selselV}. 

The approximation error in our implementation of the oracle $O_D$ is predominantly due to finite-precision computation of rotation angles (on a quantum computer). We explicitly analyse methods for computing the functions we need in Appendix~\ref{specialfunctions} to demonstrate that the gate and space complexities are at most polylogarithmic in the inverse error.

\subsection{Techniques for Algorithm~\ref{alg2}}
The basic idea of Algorithm~\ref{alg2} is to adiabatically prepare a constant-error approximation to the ground state, then apply eigenstate filtering. While such an approach has been already advocated by Wecker \textit{et al.}~\cite{Wecker2015}, their filtering method is based on quantum phase estimation, which would yield $\Omega(1/\epsilon)$ error scaling. Instead, we combine the algorithm of Lin and Tong~\cite{Lin2020} with time-dependent Hamiltonian simulation~\cite{Low2018}. Our main technical contribution is the synthesis of an oracle that when used as the input to the algorithm of~\cite{Low2018}, efficiently approximates the adiabatic evolution. 
It should be noted that we cannot use the result of~\cite{Lin2020} directly, because the operator that we use to prepare the initial state is not unitary. However, this problem is easily fixed by modifying one of the proofs in~\cite{Lin2020}; see Section~\ref{sec:adiabatic_ground} for details.

\section{Discretised quasi-adiabatic continuation\label{sec:qac}}

In this section, we lay the groundwork for our main algorithm by introducing a discretised quasi-adiabatic continuation~\cite{Hastings2004_LSM} approach to approximating adiabatic evolution. Although in later sections, we specifically consider linear interpolations of the form $H(s) = (1-s)H_0 + sH_1$, the results of this section hold for arbitrary one-parameter families of self-adjoint operators $H(s)$. They can also be straightforwardly generalised to degenerate eigenspaces that are gapped away from the rest of the spectrum. 

\subsection{Quasi-adiabatic continuation} \label{sec:3.1}
We define the \emph{quasi-adiabatic continuation operator}~\cite{Hastings2004_LSM,Hastings2005,Osborne2007,Bachmann2012}
\begin{align} D_{\Delta}(s) 
&\coloneqq \int_{-\infty}^{\infty}dt\,W_\Delta(t)e^{iH(s)t}H'(s)e^{-iH(s)t}, \label{eq:D_Delta2}
\end{align}
where $H'(s) \coloneqq dH(s)/ds$ and $W_{\Delta}(t)$ is a particular ``filter function." How well the unitary generated by $D_{\Delta}(s)$ approximates the adiabatic evolution governed by $H(s)$ depends on the choice of $W_{\Delta}(t)$. While there are several viable options, we make the following choice:\footnote{This choice for $W_{\Delta}(t)$ has several properties that make our analysis straightforward. First, $w_{\Delta}$ and its Fourier transform $\widetilde{w}_{\Delta}$ both decay rapidly. We use these facts in Lemma~\ref{Dbound1} and~\ref{Dbound2}. Second, $W_{\Delta}(t)$ restricted to $t\geq 0$ can be viewed as an unnormalised probability density function, and this function can be efficiently integrated. It is well-known that quantum states corresponding to efficiently integrable probability density functions can be efficiently prepared~\cite{Grover2002}. This will be useful for constructing oracles that encode approximations of $D_{\Delta}(s)$ [cf.~Appendix~\ref{sec:prepareW}]. We remark that there may well exist other function that not only have all of these desirable properties, but yield a better bound than that in Lemma~\ref{Dbound1}; we leave this as an open problem.}
\begin{equation} \label{eq:WDelta}
 W_\Delta(t) \coloneqq \begin{dcases} \int_t^{\infty} dt'\, w_\Delta(t') \quad &t \geq 0 \\
-\int_{-\infty}^tdt'\, w_\Delta(t') \quad &t< 0, \end{dcases}
\end{equation}
where
\begin{equation} \label{eq:wDelta}
w_{\Delta}(t) \coloneqq \frac{\Delta}{\sqrt{2\pi}}\exp\left(-\frac{\Delta^2t^2}{2}\right).
\end{equation}

\begin{lemma} \label{Dbound1} Let $H(s)$ be a one-parameter family of self-adjoint operators with bounded derivative $H'(s)$, and let $D_{\Delta}(s)$ be defined by Eqs.~\eqref{eq:D_Delta2}--\eqref{eq:wDelta} for some $\Delta \in \mathbb{R}$. For any non-degenerate eigenstate $\ket{\psi_k(s)}$ of $H(s)$ that is separated from the rest of the spectrum by a gap of at least $\gamma(s) > 0$,
\[ \left\|\frac{d}{ds}\ket{\psi_k(s)} - iD_\Delta(s))\ket{\psi_k(s)}\right\| \leq \frac{1}{\gamma(s)}\exp\left(-\frac{\gamma(s)^2}{2\Delta^2}\right)\|H'(s)\|. \]
\end{lemma}
\begin{proof}
Noting that $dW_\Delta(t)/dt = -w_\Delta(t)$ for all $t\in \mathbb{R}\setminus \{0\}$ (and can be extended by continuity at $t = 0$), we can rewrite Eq.~\eqref{eq:D_Delta2} using integration by parts as
\[ D_\Delta(s) = \int_{-\infty}^{\infty}dt\,w_\Delta(t)\int_0^t du\,e^{iH(s)u}H'(s)e^{-iH(s)u}, \]
where the boundary terms vanish since $\lim_{t \to \pm\infty}W_\Delta(t) = 0$ and $H'(s)$ is bounded.

Let $\ket{\psi_j(s)}$ denote the eigenstates of $H(s)$ with corresponding eigenvalues $E_j(s)$. Inserting the identity $I = \sum_{j}\ket{\psi_{j}(s)}\bra{\psi_{j}(s)}$, we have, for any eigenstate $\ket{\psi_k(s)}$,
\begin{align*}
iD_\Delta(s)\ket{\psi_k(s)}
&= i\sum_j\ket{\psi_j(s)}\bra{\psi_j(s)}\int_{-\infty}^{\infty}dt\,w_\Delta(t)\int_k^t du\, e^{iH(s)u}H'(s)e^{-iH(s)u}\ket{\psi_k(s)} \\
&= i\sum_{j \neq k}\ket{\psi_j(s)}\bra{\psi_j(s)}H'(s)\ket{\psi_k(s)}\int_{-\infty}^{\infty}dt\,w_\Delta(t)\int_{0}^tdu\,e^{-i(E_k(s)-E_j(s))u} \\
&= \sum_{j\neq k}\ket{\psi_j(s)}\frac{\bra{\psi_j(s)}H'(s)\ket{\psi_k(s)}}{E_k(s)-E_j(s)}\int_{-\infty}^\infty dt\,w_\Delta(t)\left[1- e^{-i(E_k(s)-E_j(s))t}\right] \\
&= \sum_{j\neq k}\ket{\psi_j(s)}\frac{\bra{\psi_j(s)}H'(s)\ket{\psi_k(s)}}{E_k(s)-E_j(s)}\left\{ 1 - \exp\left[-\frac{(E_k(s)-E_j(s))^2}{2\Delta^2}\right]\right\} \\
&= \frac{d}{ds}\ket{\psi_k(s)} - \sum_{j\neq k}\ket{\psi_j(s)}\frac{\bra{\psi_j(s)}H'(s)\ket{\psi_k(s)}}{E_k(s)-E_j(s)}\exp\left[-\frac{(E_k(s)-E_j(s))^2}{2\Delta^2}\right].
\end{align*}
In the second line, the $j =k$ term in the sum vanishes because $w_\Delta$ is an even function. The fourth line follows from the fact that the Fourier transform of $w_\Delta$ is $\widetilde{w}_\Delta(x) = \exp[-x^2/(2\Delta^2)]$, and the fifth from perturbation theory. 
Hence, by the assumption that $|E_k(s) - E_j(s)| \geq \gamma(s)$ for all $j \neq k$,
\begin{align*}
\left\|\frac{d}{ds}\ket{\psi_k(s)} - iD_\Delta(s))\ket{\psi_k(s)}\right\| &= \left\|\sum_{j\neq k}\ket{\psi_j(s)}\frac{\bra{\psi_j(s)}H'(s)\ket{\psi_k(s)}}{E_k(s)-E_j(s)}\exp\left[-\frac{(E_k(s)-E_j(s))^2}{2\Delta^2}\right]\right\| \\
&= \left\{\sum_{j\neq k}\left|\exp\left[-\frac{(E_k(s) - E_j(s))^2}{2\Delta^2}\right]\frac{\bra{\psi_j(s)}H'(s)\ket{\psi_0(s)}}{E_k(s) - E_j(s)}\right|^2\right\}^{1/2} \\
&\leq \frac{1}{\gamma(s)}\exp\left(-\frac{\gamma(s)^2}{2\Delta^2}\right)\left(\sum_{j}\left|\bra{\psi_j(s)}H'(s)\ket{\psi_k(s)}\right|^2\right)^{1/2} \\
&\leq \frac{1}{\gamma(s)}\exp\left(-\frac{\gamma(s)^2}{2\Delta^2}\right)\|H'(s)\|.
\end{align*}
\end{proof}

Let $U_{\Delta}(s) \coloneqq \mathcal{S}'[e^{i\int_0^s ds'\, D_\Delta(s')}]$ denote the ordered exponential of $D_{\Delta}(s)$. The following fact allows us to use Lemma~\ref{Dbound1} to bound the deviation of the state $U_{\Delta}(s)\ket{\psi_k(0)}$, obtained via quasi-adiabatic continuation, from the actual eigenstate $\ket{\psi_k(s)}$ of $H(s)$. 

\begin{proposition} \label{DtoU}
For any one-parameter family of self-adjoint operators $D(s)$, let $U(s)$ be the family of unitaries such that 
\begin{align*}
\frac{d}{ds}U(s) = iD(s)U(s), \quad U(0) = I.
\end{align*}
Then, for any family of states $\ket{\psi(s)}$ and any $\tau \in \mathbb{R}$,
\[ \left\|\ket{\psi(\tau)} - U(\tau)\ket{\psi(0)}\right\| \leq \int_0^\tau ds \left\|\frac{d}{ds}\ket{\psi(s)} - i D(s)\ket{\psi(s)}\right\|. \]
\end{proposition}
\begin{proof}
\begin{align*}
\left\|\ket{\psi(\tau)} - U(\tau)\ket{\psi(0)}\right\| &= \left\|U(\tau)^\dagger\ket{\psi(\tau)} - \ket{\psi(0)} \right\| \\
&= \left\|\int_0^\tau ds\, \frac{d}{ds}\left(U(s)^\dagger\ket{\psi(s)}\right)\right\| \\
&= \left\|\int_0^\tau ds\, \left(U(s)^\dagger \frac{d}{ds}\ket{\psi(s)} -iU(s)^\dagger D(s)\ket{\psi(s)}\right) \right\| \\
&\leq  \int_0^\tau ds\, \left\|\frac{d}{ds}\ket{\psi(s)} - i D(s)\ket{\psi(s)}\right\|. 
\end{align*}
\end{proof}

The following result is immediate from Lemma~\ref{Dbound1} and Proposition~\ref{DtoU}.
\begin{theorem}
\label{thm:qac_bound}
For a one-parameter family of self-adjoint operators $H(s)$, let $D_{\Delta}(s)$ be defined by Eqs.~\eqref{eq:D_Delta2}--\eqref{eq:wDelta} for some $\Delta \in \mathbb{R}$, and let $U_{\Delta}(\tau) \coloneqq \mathcal{S}[e^{i\int_0^\tau ds\, D_\Delta(s)}]$. Suppose that for all $s \in [0,\tau]$, $\|H'(s)\|$ is bounded and the $k$th eigenstate $\ket{\psi_k(s)}$ of $H(s)$ is non-degenerate, separated from the rest of the spectrum by a gap of at least $\gamma(s) > 0$. Then,
\[ \left\|\ket{\psi_k(\tau)} - U_\Delta(\tau)\ket{\psi_k(0)}\right\| \leq \int_0^\tau ds\,\frac{1}{\gamma(s)}\exp\left(-\frac{\gamma(s)^2}{2\Delta^2}\right)  \|H'(s)\|.  \]
\end{theorem}

\subsection{Discretisation} \label{sec:3.2}
Next, we introduce a discretised approximation of $D_{\Delta}(s)$, which we define as
\begin{equation} \label{D_DTN}  
\begin{aligned} D_{\Delta,T,N}(s) &\coloneqq \sum_{n=-N}^{-1}\left(\int_{nT/N}^{(n+1)T/N}dt\,W_{\Delta}(t)\right)e^{iH(s)nT/N}H'(s)e^{-iH(s)nT/N} \\
&\quad +\sum_{n=1}^{N}\left(\int_{(n-1)T/N}^{nT/N}dt\,W_\Delta(t)\right)e^{iH(s)nT/N}H'(s)e^{-iH(s)nT/N}
\end{aligned} 
\end{equation}
for $T \in \mathbb{R}_{\geq 0}$ and $N \in \mathbb{N}$. Correspondingly, we use $U_{\Delta, T, N}(s)$ to denote the familiy of unitaries generated by $D_{\Delta, T, N}(s)$: 
\begin{equation}
    \label{U_DTN}
    U_{\Delta, T, N}(s)\coloneqq \mathcal{S}' \left[\int_0^s ds'\, D_{\Delta, T, N}(s') \right].
\end{equation}
In this subsection, we show that for any eigenstate $\ket{\psi_k(s)}$ of $H(s)$ that is separated from the rest of the spectrum for all $s \in [0,\tau]$, $\|\ket{\psi_k(\tau)} - U_{\Delta,T,N}(\tau)\ket{\psi_k(0)}\|$ can be arbitrarily suppressed by making suitable choices for the parameters $\Delta$, $T$, and $N$. We begin by bounding the difference between $D_{\Delta}(s)$ and $D_{\Delta, T, N}(s)$.

\begin{lemma} \label{Dbound2} For any family of self-adjoint operators $H(s)$,
\[ \left\|D_\Delta(s) - D_{\Delta,T,N}(s)\right\| \leq 2\sqrt{2\pi}\|H'(s)\|\frac{e^{-\Delta^2T^2/2}}{\Delta} + 2\sqrt{\frac{2}{\pi}}\|H(s)\|\|H'(s)\|\frac{T}{\Delta N},\]
where $D_{\Delta}(s)$ is defined by Eqs.~\eqref{eq:D_Delta2}--\eqref{eq:wDelta} and $D_{\Delta,T,N}(s)$ is defined by Eq.~\eqref{D_DTN}.
\end{lemma}
\begin{proof}
Let
\[ D_{\Delta,T}(s) \coloneqq \int_{-T}^{T}dt\,W_\Delta(t)e^{iH(s)t} H'(s)e^{-iH(s)t}. \]
$D_{\Delta,T}(s)$ differs from $D_{\Delta}(s)$ only in that the former is defined with finite limits of integration $\pm T$. Using the tail bound $|W_\Delta(t)| \leq e^{-\Delta^2t^2/2}$, 
\begin{align*}
\left\|D_{\Delta}(s) - D_{\Delta,T}(s)\right\| &= \left\|\int_{-\infty}^{-T}dt\,W_\Delta(t)e^{iH(s)t}H'(s)e^{-iH(s)t} + \int_{T}^{\infty}dt\, W_\Delta(t)e^{iH(s)t}H'(s)e^{-iH(s)t} \right\| \\
&\leq \|H'(s)\|\left(\int_{-\infty}^{-T}dt\,|W_\Delta(t)| + \int_T^\infty dt\, |W_\Delta(t)|\right) \\
&\leq 2\|H'(s)\|\int_{T}^{\infty}e^{-\Delta^2t^2/2} \\
&= 2\|H'(s)\| \frac{\sqrt{2\pi}}{\Delta}W_\Delta(T) \\
&\leq 2\sqrt{2\pi}\|H'(s)\|\frac{e^{-\Delta^2 T^2/2}}{\Delta}. 
\end{align*}

Since 
\begingroup \allowdisplaybreaks
\begin{align*} 
D_{\Delta,T}(s) &= \sum_{n=-N}^{-1}\int_{nT/N}^{(n+1)T/N}dt\,W_{\Delta}(t)e^{iH(s)t}H'(s)e^{-iH(s)t} \\
&\quad + \sum_{n=1}^{N}\int_{(n-1)T/N}^{nT/N}dt\,W_{\Delta}(t)e^{iH(s)t}H'(s)e^{-iH(s)t},
\end{align*}
we also have
\begin{align*}
&\left\|D_{\Delta,T}(s) - D_{\Delta,T,N}(s)\right\| \\
&\quad \leq \sum_{n=-N}^{-1}\int_{nT/N}^{(n+1)T/N}dt\,\left|W_\Delta(t)\right| \left\|e^{iH(s)t}H'(s)e^{-iH(s)t} - e^{iH(s)nT/N}H'(s)e^{-iH(s)nT/N} \right\| \\
&\quad \quad + \sum_{n=1}^{N}\int_{(n-1)T/N}^{nT/N}dt\,\left|W_\Delta(t)\right| \left\|e^{iH(s)t}H'(s)e^{-iH(s)t} - e^{iH(s)nT/N}H'(s)e^{-iH(s)nT/N} \right\| \\
&\quad\leq \sum_{n=-N}^{-1}\int_{nT/N}^{(n+1)T/N}dt\,\left|W_\Delta(t)\right|\|H'(s)\|\left(\left\|e^{iH(s)t}-e^{iH(s)nT/N}\right\|+\left\|e^{-iH(s)t}-e^{-iH(s)nT/N}\right\|\right) \\
&\quad \quad + \sum_{n=1}^N\int_{(n-1)T/N}^{nT/N}dt\,\left|W_\Delta(t)\right|\|H'(s)\|\left(\left\|e^{iH(s)t}-e^{iH(s)nT/N}\right\|+\left\|e^{-iH(s)t}-e^{-iH(s)nT/N}\right\|\right) \\
&\quad \leq 2\|H'(s)\| \left(\sum_{n=-N}^{-1}\int_{nT/N}^{(n+1)T/N}dt\,\left|W_\Delta(t)\right|\left\|H(s)\right\|\left|t - nT/N\right| \right. \\
&\quad \quad \left. + \sum_{n=1}^{N}\int_{(n-1)T/N}^{nT/N}dt\,\left|W_\Delta(t)\right|\left\|H(s)\right\|\left|t - nT/N\right| \right) \\
&\quad \leq 2\|H'(s)\|\left\|H(s)\right\|\frac{T}{N}\left(\sum_{n=-N}^{-1}\int_{nT/N}^{(n+1)T/N}dt\,\left|W_\Delta(t)\right| + \sum_{n=1}^N\int_{(n-1)T/N}^{nT/N}dt\,\left|W_\Delta(t)\right|\right) \\
&\quad\leq 2\|H'(s)\|\|H(s)\|\frac{T}{N}\int_{-\infty}^{\infty}dt\,\left|W_\Delta(t)\right| \\
&\quad= 2\|H'(s)\|\left\|H(s)\right\|\frac{T}{N}\sqrt{\frac{2}{\pi}}\frac{1}{\Delta},
\end{align*}
where the third inequality is obtained by using the fact that for any self-adjoint $H$ and $t_1, t_2 \in \mathbb{R}$, 
\begin{align*}
\left\|e^{iHt_1} - e^{iHt_2}\right\| &= \left\|e^{iH(t_1 - t_2)} - I \right\|
= \left\|\int_0^{t_1 - t_2} dt\, \frac{d}{dt}\left(e^{iHt}\right)\right\| = \left\|\int_{0}^{t_1 - t_2}dt\, iH e^{iHt}\right\| \leq \|H\| |t_1 - t_2|.
\end{align*}
\endgroup

The result then follows from
\[ \left\|D_{\Delta}(s) - D_{\Delta,T,N}(s)\right\| \leq \left\|D_{\Delta}(s) - D_{\Delta,T}(s)\right\| + \left\|D_{\Delta,T}(s) - D_{\Delta,T,N}(s)\right\|. \]

\end{proof}

We can now combine Lemmas~\ref{Dbound1} and~\ref{Dbound2} with Proposition~\ref{DtoU} to prove the following. 

\begin{theorem} \label{Ubound}
For a one-parameter family of self-adjoint operators $H(s)$, let $U_{\Delta,T,N}(s)$ be the family of unitaries defined by Eqs.~\eqref{D_DTN} and~\eqref{U_DTN} for some $\Delta,T \in \mathbb{R}_{\geq 0}$ and $N \in \mathbb{N}$. Suppose that for all $s \in [0,\tau]$, $\|H'(s)\|$ is bounded and the $k$th eigenstate $\ket{\psi_k(s)}$ of $H(s)$ is non-degenerate, separated from the rest of the spectrum by a gap of at least $\gamma(s) > 0$. Then,
\begin{align*}
&\left\|\ket{\psi_k(\tau)} - U_{\Delta,T,N}(\tau)\ket{\psi_k(0)}\right\| \\
&\quad \leq \int_0^\tau ds\, \|H'(s)\|\left[\frac{1}{\gamma(s)}\exp\left(-\frac{\gamma(s)^2}{2\Delta^2}\right) + 2\sqrt{2\pi}\frac{e^{-\Delta^2T^2/2}}{\Delta} + 2\sqrt{\frac{2}{\pi}}\frac{\|H(s)\|T}{\Delta N} \right].
\end{align*}
\end{theorem}
\begin{proof}
By Proposition~\ref{DtoU}, 
\begin{align*}
\left\|\ket{\psi_k(\tau)} - U_{\Delta,T,N}(\tau)\ket{\psi_k(0)}\right\|
&\leq \int_0^\tau\,ds\left\|\frac{d}{ds}\ket{\psi_k(s)} - iD_{\Delta,T,N}(s) \ket{\psi_k(0)}\right\| \\
&\leq \int_0^s\,ds\left(\left\|\frac{d}{ds}\ket{\psi_k(s)} - iD_\Delta(s)\ket{\psi_k(s)}\right\| + \left\|D_{\Delta}(s) - D_{\Delta,T,N}(s)\right\|\right),
\end{align*}
and the result follows upon substituting the bounds from Lemma~\ref{Dbound1} and~\ref{Dbound2}. 

\end{proof}

In particular, to approximate $\ket{\psi_k(1)}$ to within any error $\epsilon$ by applying $U_{\Delta,T,N}(1)$ to $\ket{\psi_k(0)}$, we can use Theorem~\ref{Ubound} to choose the values of the parameters $\Delta$, $T$, and $N$, given lower bounds on the norms of $H(s)$ and $H'(s)$ and on the difference between the the energy of $\ket{\psi_k(s)}$ and that of the other eigenstates.

\begin{corollary}
Under the assumptions of Theorem~\ref{Ubound}, 
\[ \left\|\ket{\psi_k(1)} - U_{\Delta,T,N}(1)\ket{\psi_k(0)}\right\| \leq \epsilon \]
for some
\begin{equation*}
\Delta \in \Theta\left(\gamma \log^{-1/2}\left(\frac{\langle\|H'\|\rangle}{\gamma\epsilon}\right) \right),
\end{equation*}
\begin{equation*}
T \in \Theta\left(\frac{1}{\gamma}\log\left(\frac{\langle\|H'\|\rangle}{\gamma\epsilon}\right)\right),
\end{equation*}
\begin{equation*}
N \in \Theta\left(\frac{\alpha \langle \|H'\|\rangle}{\gamma^2\epsilon}\log^{3/2}\left(\frac{\langle\|H'\|\rangle}{\gamma\epsilon}\right)\right),
\end{equation*}
where $\gamma \leq \gamma(s)$ and $\alpha \geq \|H(s)\|$ for all $s \in [0,1]$, and $\langle\|H'\|\rangle \geq \int_0^1ds\, \|H'(s)\|$. \label{corollary:parameters}
\end{corollary}
\begin{proof}
By Theorem~\ref{Ubound},
\[ \left\|\ket{\psi_0(1)} - U_{\Delta,T,N}(1)\ket{\psi_0(0)}\right\| \leq \langle\|H'\|\rangle \left[\frac{1}{\gamma}\exp\left(-\frac{\gamma^2}{2\Delta^2}\right) + 2\sqrt{2\pi}\frac{e^{-\Delta^2T^2/2}}{\Delta} + 2\sqrt{\frac{2}{\pi}}\frac{\alpha T}{\Delta N}\right].\]
The first term on the right-hand side is upper-bounded by $\epsilon/3$ for 
\[ \Delta \leq \frac{\gamma}{\sqrt{2\ln\left(\dfrac{3\langle\|H'\|\rangle}{\gamma\epsilon}\right)}} \in {\Theta}\left(\gamma \log^{-1/2}\left(\frac{\langle\|H'\|\rangle}{\gamma\epsilon}\right) \right). \]
Then, the second term can be upper-bounded by $\epsilon/3$ by choosing
\[ T \geq \frac{1}{\Delta}\sqrt{\ln\left(\frac{6\sqrt{2}\pi}{\Delta\epsilon}\langle\|H'\|\rangle\right)} \in \Theta\left(\frac{1}{\gamma}\sqrt{\log\left(\frac{\langle\|H'\|\rangle}{\gamma\epsilon} \right)\log\left[\frac{\langle\|H'\|\rangle}{\gamma\epsilon}\sqrt{\log\left(\frac{\langle\|H'\|\rangle}{\gamma\epsilon}\right)}\right]}\right). \]
Since $\log(x\sqrt{\log x}) \in \Theta(\log x)$, the inequality is satisfied for some
\[ T \in \Theta\left(\frac{1}{\gamma}\log\left(\frac{\langle\|H'\|\rangle}{\gamma \epsilon}\right)\right). \]
With these choices for $\Delta$ and $T$, the third term is upper-bounded by $\epsilon/3$ for
\[ N \geq 6\sqrt{\frac{2}{\pi}} \frac{\alpha\langle\|H'\|\rangle T}{\Delta \epsilon} \in \Theta\left(\frac{\alpha \langle \|H'\|\rangle}{\gamma^2\epsilon}\log^{3/2}\left(\frac{\langle\|H'\|\rangle}{\gamma\epsilon}\right)\right).\]
\end{proof}

\section{Oracle synthesis} \label{sec:oracle_synthesis}
In the previous section, we showed that the adiabatic evolution corresponding to a given family of Hamiltonians $H(s)$ can be approximated using the unitary $U_{\Delta,T,N}(s)$ generated by a certain discretised quasi-adiabatic continuation operator $D_{\Delta,T,N}(s)$ [cf.~Eqs.~\eqref{D_DTN} and~\eqref{U_DTN}]. 
By constructing oracles that encode $D_{\Delta,T,N}(s)$ for different values of $s$, $U_{\Delta, T, N}(s)$ can be simulated using the truncated Dyson series algorithm of~\cite{Low2018}. More specifically, the algorithm of~\cite{Low2018} can implement an approximation to $U_{\Delta,T,N}(s)$ if given access to oracles $O_{D}^{(j)}$ such that
\begin{equation} \label{O_D^j}
(\bra{0}\otimes I)O_{D}^{(j)}(\ket{0}\otimes I) \coloneqq \sum_{m=0}^{M-1}\ket{m}\bra{m} \otimes \frac{D_{\Delta,T,N}^{(j)}(m\tau/M)}{\|D\|_{\max}},
\end{equation}
for every $j \in \{1,\dots, \lceil 2\|D\|_{\max} \rceil\}$. Here, $\|D\|_{\max}$ denotes an upper bound on $\max_{s\in[0,1]}\|D_{\Delta,T,N}(s)\|$, and $D_{\Delta,T,N}^{(j)}(s) \coloneqq D_{\Delta,T,N}((j-1)\tau + s)$ for $s \in [0, \tau]$ with $\tau \coloneqq 1/\lceil 2\|D\|_{\max} \rceil$. Thus, each $j$ corresponds to a different {segment} of the evolution. $M$ is an integer chosen according to Theorem~3 of~\cite{Low2018}, and corresponds to the number of points within each segment. In this section, we demonstrate that these oracles can be efficiently approximated for linear interpolations
\[ H(s) = (1-s)H_0 + sH_1\] 
for $s \in [0,1]$, for any self-adjoint operators $H_0$ and $H_1$. Our construction readily generalises to interpolations of the form $H(s) = (1-f(s))H_0 + f(s)H_1$ for efficiently computable functions $f$.

\subsection{Implementation outline}
Our implementation of the oracles $O_D^{(j)}$ is based on the following key identity:
\begin{equation}
\frac{D_{\Delta,T,N}(s)}{2\mathcal{N}_{\Delta,T}} = \left((\bra{+}\bra{W_{\Delta,T,N}})_d\otimes I_s \right)\sum_{\substack{n=-N\\n\neq 0}}^N \ket{n}\bra{n}_d \otimes e^{iH(s)nT/N}H' e^{-iH(s)nT/N}\left((\ket{-}\ket{W_{\Delta,T,N}})_d\otimes I_s\right)
\label{Ds}
\end{equation}
for any $s\in [0, 1]$, where $\mathcal{N}_{\Delta,T} \coloneqq  \int_0^T dt\, W_{\Delta}(t)$ and $\ket{W_{\Delta,T,N}} \coloneqq \sum_{n=1}^N \sqrt{W_{n}}\ket{n}$
with 
\begin{equation} \label{W_Delta,n} W_{n}\coloneqq \frac{1}{\mathcal{N}_{\Delta,T}}\int_{(n-1)T/N}^{nT/N}dt\, W_\Delta(t) \end{equation}
for $n \in \{1,\dots, N\}$. The register labelled $d$ encodes the states $\ket{n}$ by storing the sign of $n$ in the first qubit, so that the state of the first qubit is $\ket{0}$ if $n > 0$ and $\ket{1}$ if $n < 1$. 

Eq.~\eqref{Ds} suggests a block-encoding of $D_{\Delta,T,N}(s)$ (for a fixed $s$) that involves three main components: $H'$ (which is straightforwardly implemented using $O_{H'}$), a state preparation unitary ${W}_{\Delta,T,N}$ such that ${W}_{\Delta,T,N}\ket{0} = \ket{W_{\Delta,T,N}}$, and the multiply-controlled operation $\sum_n \ket{n}\bra{n}_d\otimes e^{-iH(s)nT/N}$ which applies $e^{iH(s)nT/N}$ conditioned on the $d$ register being in the state $\ket{n}$. Note that the only operations that depend on $s$ are the unitaries $e^{-iH(s)nT/N}$. These can be simulated using the ``qubitisation" technique of~\cite{Low2017,Low2019}, which makes queries to an oracle that block-encodes $H(s)$. By synthesising a block-encoding of $\sum_{m=0}^{M-1}\ket{m}\bra{m}\otimes H((j-1 + m/M)\tau)$ and using it as the oracle input to the qubitisation algorithm,
we can further condition the $\sum_n\ket{n}\bra{n}_d\otimes e^{-iH(s)nT/N}$ operations on the ancilla register storing $\ket{m}$, thereby constructing the block-encoding $O_{D}^{(j)}$ of $\sum_{m=1}^{M}\ket{m}\bra{m}\otimes D_{\Delta,T,N}^{(j)}(m\tau/M)$.

To make this implementation scheme more explicit, we first establish some notational conventions. For clarity of presentation, we will henceforth consider a fixed index $j \in \{1, \dots, 1/\tau\}$, and define 
\begin{equation} \label{H[m]} H[m] \coloneqq H((j-1 + m/M)\tau), \end{equation}
for $m \in \{0,\dots, M-1\}$, forgoing the label $j$ on the left-hand side. The parameters $M$ and $\tau$ are also assumed to have fixed values, to be chosen later. In this notation, replacing $H(s)$ on the right-hand side of Eq.~\eqref{Ds} with $H[m]$ yields an expression for $D^{(j)}_{\Delta,T,N}(m\tau/M)/(2\mathcal{N}_{\Delta,T})$.

We define the unitary operator
\begin{equation} \label{eq:selselV}\sel[c]\sel[d] V_{H} \coloneqq \sum_{m=0}^{M-1}\sum_{\substack{n=-N\\n\neq 0}}^N\ket{m}\bra{m}_c \otimes \ket{n}\bra{n}_d \otimes e^{-iH[m]nT/N},\end{equation}
where the dependence on $j$ is again implicit. This operation ``selects" which time evolution operator to apply based on the states of the two registers $c$ and $d$, which store $m \in \{0,\dots, M-1\}$ and $n \in \{\pm 1, \dots, \pm N\}$, respectively. It follows from Eq.~\eqref{Ds} that
\begin{equation} \label{mDs} \begin{aligned}
\sum_{m=0}^{M-1}\ket{m}\bra{m}_c \otimes \frac{D_{\Delta,T,N}^{(j)}(m\tau/M)}{2\mathcal{N}_{\Delta,T} \beta} &= (\bra{0}_b(\bra{+}\bra{W_{\Delta,T,N}})_d\otimes I_{cs})(I_b\otimes \sel[c]\sel[d] V_H)^\dagger (O_{H'}\otimes I_{cd}) \\
&\quad \times (I_b \otimes \sel[c]\sel[d] V_H)(\ket{0}_b(\ket{-}\ket{W_{\Delta,T,N}})_d\otimes I_{cs}),
\end{aligned}
\end{equation}
which implies that the circuit of Fig.~\ref{fig:outline} implements a block-encoding $O_D^{(j)}$ that satisfies Eq.~\eqref{O_D^j}. 
\begin{center}
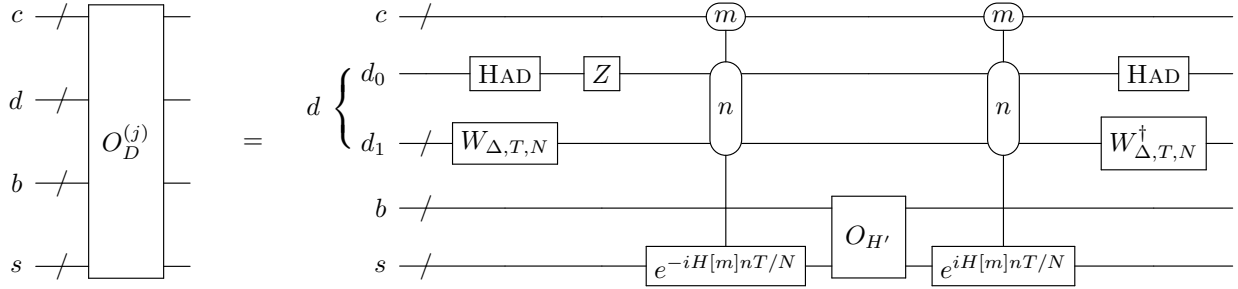

\captionsetup{type=figure}
\[{\small \Qcircuit @C=1em @R=2.25em {
&\lstick{c} &/\qw &\multigate{3}{O_D^{(j)}} &\qw \\
&\lstick{d} &/\qw &\ghost{O_D^{(j)}} &\qw \\
&\lstick{b} &/\qw &\ghost{O_D^{(j)}} &\qw \\
&\lstick{s} &/\qw &\ghost{O_D^{(j)}} &\qw
} \qquad \raisebox{-4.7em}{=} \qquad\quad
\Qcircuit @C=1em @R=1em {
&&\lstick{c} &/\qw &\qw &\qw &\cbox{m} \qwx[1] &\qw &\cbox{m} \qwx[1] &\qw &\qw \\
\lstick{\raisebox{-3.1em}{$d$}\enspace} &&\lstick{d_0} &\qw &\gate{\textsc{Had}} &\gate{Z} &\multimeasure{1}{\hspace{-0.2em} n\hspace{-0.2em}}  &\qw &\multimeasure{1}{\hspace{-0.2em} n\hspace{-0.2em}} &\gate{\textsc{Had}} &\qw \\
&&\lstick{d_1} &/\qw &\gate{{W}_{\Delta,T,N}} &\qw &\ghost{\hspace{-0.2em} n\hspace{-0.2em}} \qwx[2] &\qw &\ghost{\hspace{-0.2em} n\hspace{-0.2em}} \qwx[2] &\gate{{W}_{\Delta,T,N}^\dagger} &\qw\\
&&\lstick{b} &/\qw &\qw &\qw &\qw &\multigate{1}{O_{H'}} &\qw &\qw &\qw \\
&&\lstick{s} &/\qw &\qw &\qw &\gate{e^{-iH[m]nT/N}} &\ghost{O_{H'}} &\gate{e^{iH[m]nT/N}} &\qw &\qw 
\gategroup{2}{1}{3}{1}{0.3em}{\{}}
}\]
\captionof{figure}{An outline of the circuit that implements the oracle $O^{(j)}_D$, which block-encodes the quasi-adiabatic continuation operator $D_{\Delta,T,N}(s)$ at discrete values of $s$ [cf.~Eq.~\eqref{O_D^j}]. An important subroutine is the multiply-controlled time-evolution operator $\sel[c]\sel[d] V_H$ of Eq.~\eqref{eq:selselV}, with control registers $c$ and $d$. The state of register $c$ specifies the Hamiltonian $H[m]$ ($m \in \{0,\dots, M-1\}$) that generates the time evolution, and the state of register $d$ determines the amount of time $nT/N$ ($n \in \{\pm 1, \dots, \pm N\}$). In this diagram, the first qubit $d_0$ of register $d$ encodes the sign of $n$, while the remaining qubits $d_1$ encode its the absolute value. The unitary ${W}_{\Delta, T, N}$ maps $\ket{0}$ to $\ket{W_{\Delta,T,N}}$ [cf.~Eq.~\eqref{W_Delta,n}], and $O_{H'}$ is a block-encoding of $H' = H_1-H_0$ [cf.~Eq.~\eqref{O_H'}].}  \label{fig:outline}
\end{center}

Since the function $W_\Delta(t)$ [cf.~Eq.~\eqref{eq:WDelta}] is efficiently integrable, the unitary $W_{\Delta,T,N}$ that prepares $\ket{W_{\Delta, T, N}}$ 
can be efficiently implemented using the approach of~\cite{Grover2002}; see Appendix~\ref{sec:prepareW} for details. The oracle $O_{H'}$ is discussed in Section~\ref{sec:inputmodel}. It remains to construct $\sel[c]\sel[d] V_{H}$. In subsection~\ref{sec:selUH}, we describe how to implement an operation
$
\sel U_{\widetilde{H}} \coloneqq \sum_{m=0}^{M-1}\ket{m}\bra{m} \otimes U_{\widetilde{H}[m]}
$
using one query to each of $O_{H_0}$ and $O_{H_1}$, where for each $m$, $U_{\widetilde{H}[m]}$ is a unitary that block-encodes an approximation $\widetilde{H}[m]$ to $H[m]$.
Then, by applying the qubitisation algorithm of Refs.~\cite{Low2017,Low2019} to $\sel U_{\widetilde{H}}$, we can construct an approximation $\sel[c]\sel[d] \widetilde{V}_{\widetilde{H}}$ to $\sel[c]\sel[d] V_H$, as shown in subsection~\ref{sec:selselV}. 
Finally, in subsection~\ref{sec:full_oracle}, we put all of the components together, and bound the cost of implementing approximations to $O_{D}^{(j)}$ in terms of the target precision. 

Note that the operators $\sel U_{{H}}$ and $\sel\sel \widetilde{V}_{{H}}$ and their approximations are implicitly associated with an index $j$, in the same spirit as that of Eq.~\eqref{H[m]}. On the other hand, the state $\ket{W_{\Delta,T,N}}$ does not depend on the parameter $s$ [cf.~Eq.~\eqref{W_Delta,n}] and is therefore the same for all $j$. It will be clear that the methods for implementing $\sel U_{\widetilde{H}}$ and $\sel\sel \widetilde{V}_{\widetilde{H}}$ described in the following subsections are applicable to any $j$. 

\subsection{Circuit for \textnormal{$\sel U_{\widetilde{H}}$}} \label{sec:selUH}

In this subsection, we construct a unitary
\begin{equation}\label{selUHtilde} \sel U_{\widetilde{H}} \coloneqq \sum_{m=0}^{M-1}\ket{m}\bra{m} \otimes U_{\widetilde{H}[m]}, \end{equation}
where for each $m \in \{0,\dots, M-1\}$, $U_{\widetilde{H}[m]}$ is a block-encoding of a self-adjoint operator $\widetilde{H}[m]$ that approximates $H[m]$. 

The intuition behind our construction is easily understood by first considering how $H[m]$ could be block-encoded using infinite-precision operations. Indeed, for any $s \in [0,1]$, it is straightforward to construct a block-encoding of $H(s) = (1-s)H_0 + sH_1$ using the oracles $O_{H_0}$ and $O_{H_1}$. Let $\theta(s) \coloneqq \arcsin(\sqrt{s})/(2\pi)$, and for any $\theta \in \mathbb{R}$, define $R(\theta) \coloneqq e^{-i2\pi\theta Y}$ so that $R(\theta)\ket{0} = \cos(2\pi\theta)\ket{0} + \sin(2\pi\theta)\ket{1}$. Note that
\begin{equation} \label{beH(s)}
(\bra{0}_{a'}\bra{0}_a\otimes I_s)(R(\theta(s))_{a'}^\dagger \otimes I_{as})\left(\sel[a'] O_{H}\right)(R(\theta(s))_{a'}\otimes I_{as})(\ket{0}_{a'}\ket{0}_a\otimes I_s) = \frac{H(s)}{\alpha},
\end{equation}
where $\sel O_H$ is defined as in Eq.~\eqref{selOH}. 
In particular, by choosing the rotation angle to be
\begin{equation*} \label{eq:theta_m} \theta_m \coloneqq \frac{1}{2\pi}\arcsin\left(\sqrt{(j-1 + m/M)\tau}\right), \end{equation*}
it follows from Eq.~\eqref{beH(s)} that 
\begin{equation} \label{eq:U_H[m]} U_{H[m]} \coloneqq (R(\theta_m)_{a'}^\dagger \otimes I_{as})(\sel[a'] O_{H})(R(\theta_m)_{a'}\otimes I_{as}) \end{equation} block-encodes $H[m]$. Then, if we introduce another register $c$, to represent $m \in \{0,\dots, M-1\}$, and replace the $R(\theta_m)$ operators in Eq.~\eqref{eq:U_H[m]} by the multiply-controlled rotations $\sum_{m=0}^{M-1}\ket{m}\bra{m}_c \otimes R(\theta_m)$,
we would obtain a multiply-controlled unitary $\sel[c] U_{H} \coloneqq \sum_{m=0}^{M-1}\ket{m}\bra{m}_c\otimes U_{H[m]}$ that applies the block-encoding $U_{H[m]}$ of $H[m]$ conditioned on the state of register $c$ being set to $\ket{m}$. The circuit for $\sel[c] U_{H}$ is depicted in Fig.~\ref{fig:selUH}.

\begin{center}
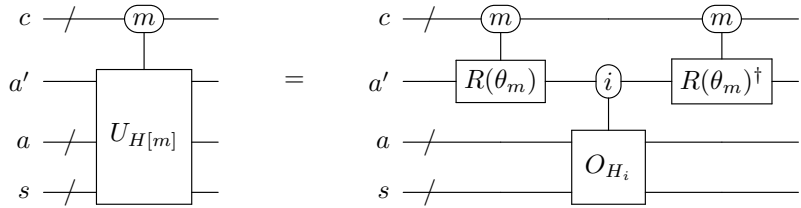
 
\captionsetup{type=figure}
\[ {\small\Qcircuit @C=1em @R=1em
{
\lstick{c} &/\qw &\cbox{m} \qwx[1] &\qw &&&&&&& \lstick{c} &/\qw & \measure{\mbox{$m$}} \qwx[1] &\qw &\measure{\mbox{$m$}} \qwx[1] &\qw\\
\lstick{a'} &\qw &\multigate{2}{U_{H[m]}} &\qw & &\rstick{=} &&&&&\lstick{a'} &\qw &\gate{R(\theta_m)} &\measure{\mbox{$i$}}\qwx[1] &\gate{R(\theta_m)^\dagger} &\qw \\
\lstick{a} &/\qw &\ghost{U_{H[m]}} &\qw &&&&&&&\lstick{a} &/ \qw &\qw &\multigate{1}{O_{H_i}} &\qw &\qw\\
\lstick{s} &/\qw &\ghost{U_{H[m]}} &\qw &&&&&&&\lstick{s} &/ \qw &\qw &\ghost{O_{H_i}} &\qw &\qw
}}
\] 
\captionof{figure}{Circuit representation of $\sel[c] U_{H} \coloneqq \sum_{m=0}^{M-1}\ket{m}\bra{m}_c\otimes U_{H[m]}$. The multiply-controlled rotation operators, which apply $R(\theta_m)$ conditioned on register $c$ being in the state $\ket{m}$, are idealised in the sense that the angles $\theta_m$ are assumed to be computed with infinite precision. This assumption is relaxed in Lemma~\ref{lemma:sel_UtildeH}.} \label{fig:selUH}
\end{center}

In practice, $\sum_{m=0}^{M-1}\ket{m}\bra{m}_c\otimes R(\theta_m)$ cannot be implemented perfectly. The angles $\theta_m$ can only be coherently computed to some finite precision, leading to an approximate version $\sel U_{\widetilde{H}}$ of $\sel U_H$. Here, we construct multiply-controlled rotation operators as products of two-qubit controlled rotations.\footnote{Alternatively, the multiply-controlled rotations can be synthesised using the phase gradient technique~\cite{Kitaev2002}. This would lead to the same asymptotic scaling, but requires introducing an additional ancilla register, which would make the analysis marginally more complicated.}

\begin{proposition} \label{prop:selR}
Let $R(\theta) \coloneqq e^{-i 2\pi \theta Y}$. For $b \in \mathbb{N}$, the $(b+1)$-qubit unitary
\begin{equation} \label{eq:selRtilde} \sel R \coloneqq \sum_{x=0}^{2^b-1}\ket{x}\bra{x}\otimes {R}(x/2^b) \end{equation} 
can be implemented using $b$ elementary gates (and zero ancilla qubits). 
\end{proposition}
\begin{proof}
Let $x = x_{b-1} \dots x_1x_0$ be the binary representation of $x$, i.e., $x = \sum_{k=0}^{b-1} x_k 2^k$. $\sel R$ can be implemented by controlling $R(2^{k-b})$ on the qubit storing $\ket{x_k}$, for each $k \in \{0, \dots, b -1\}$.  
\end{proof}

As will become clear in the next subsection, the total error in simulating the adiabatic evolution corresponding to $H(s)$ depends in part on $\|H[m] - \widetilde{H}[m]\|$. This can be controlled by computing each $\theta_m$ to higher precision, and thus determines the complexity of implementing $\sel U_{\widetilde{H}}$.

\begin{lemma} Let $H[m]$ be defined as in Eq.~\eqref{H[m]}.
For any $\epsilon_0 > 0$, a $(n_s + n_a + \ceil{\log M} +1)$-qubit unitary $\sel U_{\widetilde{H}}$ of the form of Eq.~\eqref{selUHtilde} 
such that
\[ \left\|H[m] - \widetilde{H}[m]\right\| \leq \epsilon_0 \]
for all $m \in \{0,\dots, M-1\}$ can be implemented using 
\begin{itemize}
\item one query to each of $O_{H_0}$ and $O_{H_1}$,
\item $\mathcal{O}(\log({\alpha}/{\epsilon_0})\mathcal{M}(\log({\alpha}/{\epsilon_0})))$ elementary gates, and
\item $\mathcal{O}(\log^2({\alpha}/{\epsilon_0}))$ ancilla qubits (initialised in and returned to $\ket{0}$).
\end{itemize}
\label{lemma:sel_UtildeH}
\end{lemma}
\begin{proof}
We implement $\sel U_{\widetilde{H}}$ using the circuit of Fig.~\ref{fig:selUHtilde}, which has the same structure as the circuit for $\sel U_H$ of Fig.~\ref{fig:selUH}. The difference is that the infinite-precision rotations in Fig.~\ref{fig:selUH} are replaced by finite-precision rotations. Conditioned on register $c$ being in the state $\ket{m}$ for $m \in \{0,\dots, M-1\}$, a $b$-bit integer $\widetilde{\theta}_m$ approximating $2^b \theta_m$ is computed and written to the $b$-qubit ancilla register $e$, which is initialised in the state $\ket{0}$. The multiply-controlled rotation $\sel {R}$ of Proposition~\ref{prop:selR} is then applied to register $e$ and the single-qubit register $a'$, with $e$ as the control register and $a'$ as the target. Register $e$ is subsequently uncomputed. By Eq.~\eqref{eq:selRtilde}, the combined effect on registers $c$ and $a'$ is the operation $\sum_{m=0}^{M-1}\ket{m}\bra{m}_c \otimes {R}(\widetilde{\theta}_m)_{a'}$.

\begin{figure}
\captionsetup{type=figure}
\[ {\small\Qcircuit @C=1em @R=1em {
&\lstick{c} &/\qw &\cbox{m} \qwx[2] &\qw 
&&&&&&& \lstick{c} &/\qw &\cbox{m} \qwx[1] &\qw &\qw &\qw &\cbox{m} \qwx[1] &\qw \\
&\lstick{\ket{0}_e} &/\qw &\qw &\qw
&&&&&&& \lstick{\ket{0}_e} &/\qw &\gate{\widetilde{\theta}_m} &\cbox{\widetilde{\theta}_m} \qwx[1] &\qw &\cbox{\widetilde{\theta}_m} \qwx[1] &\gate{\text{$\widetilde{\theta}_m$}^\dagger} &\qw \\
&\lstick{a'} &\qw &\multigate{2}{U_{\widetilde{H}[m]}} &\qw &&\rstick{=} 
&&&&& \lstick{a'} &\qw &\qw &\gate{{R}(\widetilde{\theta}_m/2^b)} &\cbox{i} \qwx[1] &\gate{{R}(\widetilde{\theta}_m/2^b)^\dagger} &\qw &\qw \\
&\lstick{a} &/\qw &\ghost{U_{\widetilde{H}[m]}} &\qw 
&&&&&&& \lstick{a} &/\qw &\qw &\qw &\multigate{1}{O_{H_i}} &\qw &\qw &\qw \\
&\lstick{s} &/\qw &\ghost{U_{\widetilde{H}[m]}} &\qw 
&&&&&&& \lstick{s} &/\qw &\qw &\qw &\ghost{O_{H_i}} &\qw &\qw &\qw
}}\]
\caption{The circuit implementation of $\sel U_{\widetilde{H}}$. The first gate on the right-hand side maps $\ket{m}_c\ket{0}_e$ to $\ket{m}_c \ket{\widetilde{\theta}_m}_e$ for all $m \in \{0,\dots, M-1\}$), where $\widetilde{\theta}_m$ is an integer approximation to $2^b\theta_m$. Then, the second gate effects a rotation by $\widetilde{\theta}_m/2^b$ on register $a'$, using the approach of Proposition~\ref{prop:selR}. The operator $\sel[a'] O_H$ applies the block-encoding $O_{H_i}$ of $H_i$ conditioned on register $a'$ being in the state $i$, for $i \in \{0,1\}$ [cf.~Eq.~\eqref{selOH}]. Note that register $e$ is ultimately reset to its initial state, unentangled from the other registers.} \label{fig:selUHtilde}
\end{figure}
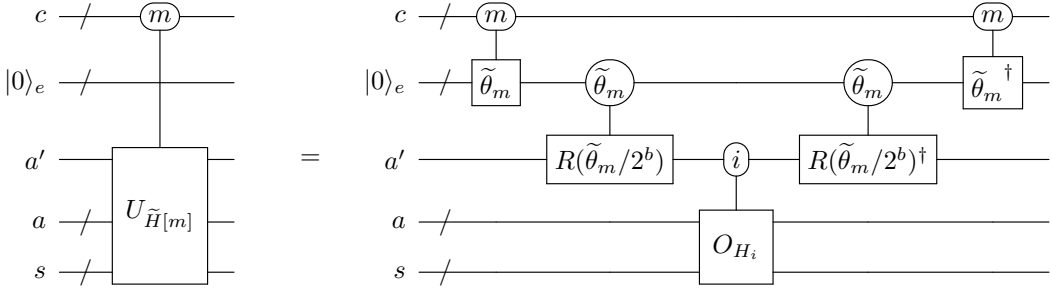

Thus,
\begin{align*} \sel U_{\widetilde{H}}
&= \sum_{m=0}^{M-1}\ket{m}\bra{m}_c \otimes \left({R}(\widetilde{\theta}_m/2^b)_{a'}^\dagger\otimes I_{as} \right)\left(\sel[a'] O_H\right)\left({R}(\widetilde{\theta}_m/2^b)_{a'}\otimes I_{as} \right),
\end{align*}
which has the form of Eq.~\eqref{selUHtilde} with \begin{equation} \label{eq:U_H[m]tilde} U_{\widetilde{H}[m]} = ({R}(\widetilde{\theta}_m/2^b)_{a'}^\dagger\otimes I_{as})(\sel[a'] O_H)({R}(\widetilde{\theta}_m/2^b)_{a'}\otimes I_{as}) \end{equation} for each $m \in \{0,\dots, M-1\}$.  Each $U_{\widetilde{H}[m]}$ block-encodes
${\widetilde{H}[m]} \coloneqq \alpha (\bra{0}_{a'a} \otimes I_s)U_{\widetilde{H}[m]}(\ket{0}_{a'a} \otimes I_s)$.
Then, since $H[m]/\alpha = (\bra{0}_{a'a}\otimes I_s)U_{H[m]}(\ket{0}_{a'a}\otimes I_s)$, we can bound $\|H[m] - \widetilde{H}[m]\|$ as
\begin{align*}
    \left\|H[m] - \widetilde{H}[m]\right\| &\leq \alpha\left\|U_{H[m]} - U_{\widetilde{H}[m]}\right\| \leq 2\alpha\left\|R(\theta_m) - {R}(\widetilde{\theta}_m/2^b)\right\| \leq 4\pi\alpha\left|\theta_m - \widetilde{\theta}_m/2^b\right|,
\end{align*}
where the second inequality follows from Eqs.~\eqref{eq:U_H[m]} and~\eqref{eq:U_H[m]tilde}, and the third from the fact that for any $\theta_1,\theta_2 \in \mathbb{R}$, $\|R(\theta_1) - R(\theta_2)\| = |e^{i2\pi\theta_1} - e^{i2\pi\theta_2}| \leq 2\pi|\theta_1 -\theta_2|$. For this error to be at most $\epsilon_0$ for all $m \in \{0,\dots, M-1\}$, each $\theta_m$ must be approximated to within absolute error $\epsilon_0/(4\pi\alpha)$.

Recall that $\theta_m \coloneqq \arcsin(\sqrt{(j-1+m/M)\tau})/(2\pi)$, and note that since $1/\tau$ is an integer by definition, a $b$-bit approximation $\widetilde{y}_m$ of the argument $y_m \coloneqq (j-1 + m/M)\tau$ can be computed such that $|y_m - \widetilde{y}_m| = \mathcal{O}(2^{-b})$. 
Combining Proposition~\ref{prop:arcsin} with the fact that $|\mathrm{arcsin}(\sqrt{y}) - \arcsin(\sqrt{\widetilde{y}})| = \mathcal{O}(\sqrt{|y - \widetilde{y}|})$ for any $y, \widetilde{y} \geq 0$,
we can construct a quantum circuit that maps $\ket{m}_c\ket{0}_e$ to $\ket{m}\ket{\widetilde{\theta}_m}_e$ such that
\[ |\theta_m - \widetilde{\theta}_m/2^b| \leq \frac{\epsilon_0}{4\pi\alpha} \]
for all $m \in \{0,\dots, M-1\}$ using $\mathcal{O}(\log(\alpha/\epsilon_0)\mathcal{M}(\log(\alpha/\epsilon)))$ elementary gates and $\mathcal{O}(\log^2(\alpha/\epsilon))$ ancillae. By Proposition~\ref{prop:arcsin}, each $\widetilde{\theta}_m$ is a $\mathcal{O}(\log(\alpha/\epsilon_0))$-bit number, so we allocate $b = \mathcal{O}(\log(\alpha/\epsilon_0))$ qubits to register~$e$. It then follows from Proposition~\ref{prop:selR} that $\sel {R}$ consists of $\mathcal{O}(\log(\alpha/\epsilon_0))$ elementary gates. 

Therefore, the total gate complexity is $\mathcal{O}(\log(\alpha/\epsilon_0)\mathcal{M}(\log(\alpha/\epsilon_0)))$, and the number of ancillae required is $\mathcal{O}(\log^2(\alpha/\epsilon_0))$.  In addition, $\sel U_{\widetilde{H}}$ uses one application of $\sel O_H$, which makes one query to each of $O_{H_0}$ and $O_{H_1}$. 
\end{proof}

\subsection{Circuit for \textnormal{$\sel[c]\sel[d] \widetilde{V}_{\widetilde{H}}$}} \label{sec:selselV}

For any self-adjoint operator $H$ with $\|H\| \leq 1$ and $t \in \mathbb{R}$, the qubitisation algorithm of Refs.~\cite{Low2017,Low2019} can implement, with probability $1 - \mathcal{O}(\epsilon')$, an operator $\widetilde{V}_H(t)$ such that $\|e^{-iHt} - \widetilde{V}_H(t)\| \leq \epsilon'$ and $\|\widetilde{V}_H(t)\| \leq 1$ by making $\mathcal{O}(t + \log(1/\epsilon'))$ queries to an oracle $U_H$ that block-encodes $H$. Recall from Eq.~\eqref{selUHtilde} that the unitary $\sel[c] U_{\widetilde{H}}$ constructed in Lemma~\ref{lemma:sel_UtildeH} applies a block-encoding of $\widetilde{H}[m]/\alpha$ conditioned on the control register $c$ being in the state $\ket{m}$. Therefore, as summarised in Fig.~\ref{fig:controlled-qubitisation}, by adding a $\lceil \log M\rceil$-qubit ancilla register $c$ and replacing each of the queries in the qubitisation algorithm with an application of $\sel[c] U_{\widetilde{H}}$, we can implement the multiply-controlled operation
\begin{equation} \label{selV} \sel[c] \widetilde{V}_{\widetilde{H}}(t) \coloneqq \sum_{m=0}^{M-1}\ket{m}\bra{m}_c \otimes \widetilde{V}_{\widetilde{H}[m]}(t), \end{equation}
where each $\widetilde{V}_{\widetilde{H}[m]}(t)$ approximates $e^{-i\widetilde{H}[m]t}$. Since $e^{-i\widetilde{H}[m]t} = e^{-i(\widetilde{H}[m]/\alpha)(\alpha t)}$ and $\|\widetilde{H}[m]/\alpha\| \leq 1$ for every $m \in \{0, \dots, M-1\}$, $\mathcal{O}(\alpha t + \log(1/\epsilon'))$  queries to $\sel[c] U_{\widetilde{H}}$ are sufficient to implement $\sel[c] \widetilde{V}_{\widetilde{H}}(t)$ such that $\|e^{-i\widetilde{H}[m]t} - \widetilde{V}_{\widetilde{H}[m]}(t)\| \leq \epsilon'$ with success probability $1 - \mathcal{O}(\epsilon')$. 

\begin{center}
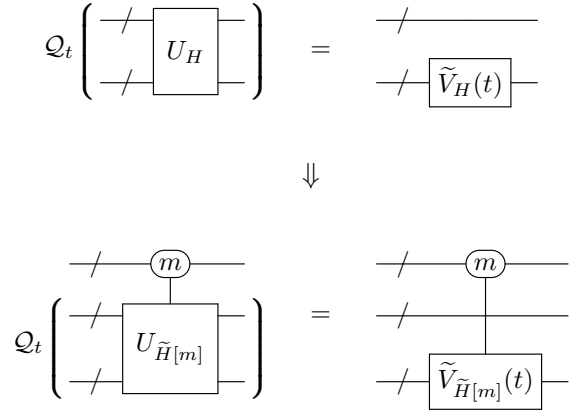

\captionsetup{type=figure}
\begin{gather*}
{\small\Qcircuit @C=1em @R=1em{
\lstick{\raisebox{-2.7em}{$\mathcal{Q}_t \enspace$}} &/ \qw &\multigate{1}{U_H} &\qw &&\rstick{\raisebox{-2.7em}=} &&&&/ \qw &\qw &\qw \\
&/ \qw &\ghost{U_H} &\qw &&&&& &/\qw &\gate{\widetilde{V}_{H}(t)} &\qw
\gategroup{1}{1}{2}{4}{1.2em}{(}
\gategroup{1}{1}{2}{4}{1.2em}{)}
}}\\\\
\Downarrow \\\\
{\small\Qcircuit @C=1em @R=1em {
&/\qw &\cbox{m}\qwx[1] &\qw &&&&& &/\qw &\cbox{m} \qwx[2] &\qw \\
\lstick{\raisebox{-2.7em}{$\mathcal{Q}_t \enspace$}} &/\qw &\multigate{1}{U_{\widetilde{H}[m]}} &\qw &&\rstick{=} &&&&/\qw &\qw &\qw  \\
&/\qw &\ghost{U_{\widetilde{H}[m]}} &\qw &&&&& &/\qw &\gate{\widetilde{V}_{\widetilde{H}[m]}(t)} &\qw
\gategroup{2}{1}{3}{4}{1.2em}{(}
\gategroup{2}{1}{3}{4}{1.2em}{)}
}} \end{gather*}
\captionof{figure}{If each oracle query in the qubitisation algorithm~\cite{Low2017,Low2019} is conditioned on a control register, such that a block-encoding of $\widetilde{H}[m]$ is applied when the control register is in the state $\ket{m}$, the resulting operation simulates the evolution generated by $\widetilde{H}[m]$ conditioned on the control register being in the state $\ket{m}$. In this figure, $\mathcal{Q}_t(\cdot)$ represents using the qubitisation algorithm to simulate the evolution for time $t$, with the block-encoding in the brackets as the oracle input. \label{fig:controlled-qubitisation}}
\end{center}

We then use $\sel[c] \widetilde{V}_{\widetilde{H}}(t)$ to construct the operation
\begin{equation} \label{eq:selselVtilde}
\sel[c]\sel[d] \widetilde{V}_{\widetilde{H}} \coloneqq \sum_{m=0}^{M-1}\sum_{\substack{n=-N\\n\neq 0}}^N\ket{m}\bra{m}_c \otimes \ket{n}\bra{n}_d \otimes \widetilde{V}_{\widetilde{H}[m]}(nT/N),
\end{equation} 
This is an approximate version of the unitary $\sel[c] \sel[d] V_H$ [cf.~Eq.~\eqref{eq:selselV}] required by Eq.~\eqref{mDs}. Conditioned on the $c$ register being in the state $\ket{m}$ and the $d$ register being in the state $\ket{n}$, $\sel[c]\sel[d] \widetilde{V}_{\widetilde{H}}$ applies an approximation $\widetilde{V}_{\widetilde{H}[m]}(nT/N)$ of $e^{-iH[m]nT/N}$. The error in this approximation is determined by 1) the difference between ${H}[m]$ and $\widetilde{H}[m]$, which can be suppressed by computing the rotation angles in the construction of $\sel U_{\widetilde{H}}$ to higher precision [cf.~Lemma~\ref{lemma:sel_UtildeH}], and 2) the error in using qubitisation to simulate the time evolution corresponding to $\widetilde{H}[m]$, which can be reduced at the cost of making more ``queries" to the subroutine $\sel U_{\widetilde{H}}$. The following lemma provides an explicit implementation of $\sel[c]\sel[d] \widetilde{V}_{\widetilde{H}}$ and bounds its complexity in terms of the desired precision. 
\begin{lemma} \label{lem:selselV}
Let $H[m]$ be defined as in Eq.~\eqref{H[m]}. For any $\epsilon_1 > 0$, a $(n_s + \ceil{\log M} + \ceil{\log N})$-qubit operator $\sel[c]\sel[d] \widetilde{V}_{\widetilde{H}}$ of the form of Eq.~\eqref{eq:selselVtilde} such that 
\[ \left\|e^{-iH[m]nT/N} - \widetilde{V}_{\widetilde{H}[m]}(nT/N)\right\| \leq \epsilon_1 \] 
for all $m \in \{0,\dots, M-1\}$ and $n \in \{\pm 1, \dots, \pm N\}$ can be implemented with probability $1 - \mathcal{O}(\epsilon_1)$ using
\begin{itemize}
    \item $\mathcal{O}\left(\alpha T + \log\left(\frac{\log N}{\epsilon_1}\right)\log N\right)$ queries to $O_{H_0}$ and $O_{H_1}$,
    \item $\mathcal{O}\left(\left[n_a + \log\left(\frac{\alpha T}{\epsilon_1}\right)\mathcal{M}\left(\log\left(\frac{\alpha T}{\epsilon_1}\right)\right)\right]\left[\alpha T + \log\left(\frac{\log N}{\epsilon_1}\right)\log N\right]\right)$ elementary gates, and
    \item $n_a + \mathcal{O}(\log^2(\alpha T/\epsilon_1))$ ancilla qubits (initialised in and reset to $\ket{0}$).
\end{itemize}
\end{lemma}
\begin{proof}
We use a specific encoding of the states $\ket{n}$ in the $(\lceil \log N\rceil +1)$-qubit register labelled $d$. The first qubit $d_0$ of this register stores the sign of $n$ as $\ket{(1-\sgn(n))/2}$, and the remaining $\lceil \log N\rceil$ qubits, which we will collectively denote by $d_1$, encode $|n| - 1$ in binary. Thus, \[ \ket{n}_d \equiv \ket{(1-\sgn(n))/2}_{d_0}\ket{|n|-1}_{d_1} \] for all $n \in \{\pm 1, \dots, \pm N\}$. Note that this encoding is consistent with the one assumed in Eq.~\eqref{Ds}. Let $n_{\lceil\log N\rceil -1}\dots n_1 n_0$ be the binary representation of $|n| - 1$, i.e., $|n| - 1 = \sum_{i=0}^{\lceil \log N\rceil-1}n_i 2^{i}$. Consider the circuit in Fig.~\ref{fig:selselH_segmentize}. For each $i \in \{0,\dots, \lceil \log N \rceil - 1\}$, $\sel[c] \widetilde{V}_{\widetilde{H}}(2^{i}T/N)$ is controlled on the qubit storing $\ket{n_i}$, where  $\sel[c] \widetilde{V}_{\widetilde{H}}(t)$ denotes the operation resulting from running the qubitisation algorithm for time $t$ with $\sel[c] U_{\widetilde{H}}$ as the oracle input [cf.~Fig.~\ref{fig:controlled-qubitisation}]. In addition, $\sel[c] \widetilde{V}_{\widetilde{H}}(T/N)$ is applied independently of the $d$ register.\footnote{\label{footnote1} Technically speaking, the diagrams that involve approximate time-evolution operators do not depict quantum circuits in full detail. The operators $\widetilde{V}_{\widetilde{H}[m]}(t)$ are not necessarily unitary; they are implemented by applying a unitary to a larger system, then measuring the ancilla register (which is why each application of $\widetilde{V}_{H[m]}(t)$ fails with some probability) [cf.~\cite{Low2018}].} When the $c$ register is in the state $\ket{m}$, this circuit effects the operation $\sum_{|n|=1}^N\ket{|n|-1}\bra{|n|-1}_{d_1}\otimes \widetilde{V}_{\widetilde{H}[m]}(|n|T/N)$ on registers $d_1$ and $s$, where
\begin{equation} \label{V(|n|)}
    \widetilde{V}_{\widetilde{H}[m]}(|n|T/N) \coloneqq \widetilde{V}_{\widetilde{H}[m]}(T/N)\prod_{i=0}^{\lceil\log N\rceil - 1}\widetilde{V}_{\widetilde{H}[m]}(n_i2^iT/N)
\end{equation}
for $|n| \in \{1,\dots,N\}$, with $\widetilde{V}_{\widetilde{H}[m]}(0) \equiv I$. 
Since $\ket{0}_{d_0}\ket{|n|-1}_{d_1} = \ket{|n|}_d$ and $\ket{1}_{d_0}\ket{|n|-1}_{d_1} = \ket{-|n|}_d$, the circuit of Fig.~\ref{fig:selselV_pm} implements 
\begin{align*}
    &\sum_{m=0}^{M-1}\sum_{|n|=1}^N \ket{m}\bra{m}_c\otimes \ket{|n|-1}\bra{|n|-1}_{d_1} \otimes \left( \ket{0}\bra{0}_{d_0}\otimes  \widetilde{V}_{\widetilde{H}[m]}(|n|T/N) + \ket{1}\bra{1}_{d_0}\otimes \widetilde{V}_{\widetilde{H}[m]}(|n|T/N)^\dagger \right) \\
    &\quad = \sum_{m=0}^{M-1}\ket{m}\bra{m}_c \otimes \left(\sum_{n=1}^N\ket{n}\bra{n}_d\otimes \widetilde{V}_{\widetilde{H}[m]}(|n|T/N) + \sum_{n=-N}^{-1}\ket{n}\bra{n}_d\otimes \widetilde{V}_{\widetilde{H}[m]}(|n|T/N)^\dagger \right),
\end{align*}
which is of the form of Eq.~\eqref{eq:selselVtilde} with $\widetilde{V}_{\widetilde{H}[m]}(nT/N) = \widetilde{V}_{\widetilde{H}[m]}(|n|T/N)$ for $n > 0$ and $\widetilde{V}_{\widetilde{H}[m]}(nT/N) = \widetilde{V}_{\widetilde{H}[m]}(|n|T/N)^\dagger$ for $n < 0$. 

\begin{figure}[h!]
\begin{subfigure}{\textwidth}
\centering
\[{\small\Qcircuit @C=1em @R=1em {
\lstick{c} &/\qw &\cbox{m} \qwx[1] &\qw \\
\lstick{d_1} &/\qw &\cbox{|n|} \qwx[3] &\qw \\
 \\
 \\
\lstick{s} &/\qw &\gate{\widetilde{V}_{\widetilde{H}[m]}\left(|n|\frac{T}{N}\right)} &\qw 
}}\]
\[ \rotatebox{90}{$=$} \]
\[
{\small\Qcircuit @C=1em @R=1em {
&\lstick{c} &/\qw &\cbox{m} \qwx[4] &\cbox{m} \qwx[3] &\qw &\dots & &\cbox{m} \qwx[1] &\cbox{m} \qwx[5] &\qw \\
&\lstick{\raisebox{-4em}{$d_1$\enspace\enspace }} &\qw &\qw &\qw &\qw &\qw &\qw &\ctrl{4} &\qw &\qw \\
&&\vdots  \\
& &\qw &\qw &\ctrl{2} &\qw &\qw &\qw &\qw  &\qw &\qw \\
& &\qw &\ctrl{1} &\qw &\qw &\qw &\qw &\qw &\qw &\qw\\
&\lstick{s} &/\qw &\gate{\widetilde{V}_{\widetilde{H}[m]}\left(2^0\frac{T}{N}\right)} &\gate{\widetilde{V}_{\widetilde{H}[m]}\left(2^1\frac{T}{N}\right)} &\qw &\dots &&\gate{\widetilde{V}_{\widetilde{H}[m]}\left(2^{\lceil \log N\rceil -1}\frac{T}{N}\right)} &\gate{\widetilde{V}_{\widetilde{H}[m]}\left(\frac{T}{N}\right)}  &\qw
\gategroup{2}{2}{5}{2}{0.7em}{\{}
}}
\]
\caption{A circuit that implements $\widetilde{V}_{\widetilde{H}[m]}(|n|T/N)$ on register $s$ conditioned on register $c$ and sub-register $d_1$ being in the states $\ket{m}$ and $\ket{|n|-1}$, respectively. The top wire of $d_1$ stores the most significant bit of $|n| -1$.}
\label{fig:selselH_segmentize}
\end{subfigure}

\begin{subfigure}{\textwidth}
\centering
\[ {\small\Qcircuit @C=1em @R=1em {
\lstick{c} &/\qw &\cbox{m} \qwx[1] &\qw 
&&&&&&&&& \lstick{c} &/\qw &\cbox{m} \qwx[1] &\cbox{m} \qwx[1] &\qw \\ 
\lstick{d} &/\qw &\cbox{n} \qwx[2] &\qw &&\rstick{=} 
&&&&&\lstick{\raisebox{-3em}{$d$}\enspace}&& \lstick{d_0} &\qw &\ctrlo{1} &\ctrl{1} &\qw \\
&&&&&&&&&&&& \lstick{d_1} &/\qw &\cbox{|n|} \qwx[1] &\cbox{|n|} \qwx[1] &\qw\\
\lstick{s} &/\qw &\gate{\widetilde{V}_{\widetilde{H}[m]}\left(n\frac{T}{N}\right)} &\qw 
&&&&&&&&& \lstick{s} &/\qw &\gate{\widetilde{V}_{\widetilde{H}[m]}\left(|n|\frac{T}{N}\right)} &\gate{\widetilde{V}_{\widetilde{H}[m]}\left(|n|\frac{T}{N}\right)^\dagger} &\qw
\gategroup{2}{11}{3}{11}{0.3em}{\{}}} \]
\caption{$\sel[c]\sel[d] \widetilde{V}_{\widetilde{H}}$ can be constructed using the multiply-controlled operation in Fig.~\ref{fig:selselH_segmentize}.}.
\label{fig:selselV_pm}
\end{subfigure}
\caption{The implementation of $\sel[c]\sel[c]\widetilde{V}_{\widetilde{H}}$ considered in Lemma~\ref{lem:selselV}.   \label{fig:selselV}}
\end{figure}

Using Lemma~\ref{lemma:sel_UtildeH}, we construct $\sel[c] U_{\widetilde{H}}$ such that $\|H[m] - \widetilde{H}[m]\| \leq \epsilon_0$ for all $m \in \{0,\dots, M-~1\}$. Suppose that each of the $\sel[c] \widetilde{V}_{\widetilde{H}}(2^iT/N)$ operations ($i \in \{0, \dots, \lceil \log N \rceil - 1\}$) in Fig.~\ref{fig:selselH_segmentize}  is implemented such that 
\begin{equation} \label{qubitisationerror} \left\|e^{-i\widetilde{H}[m]2^iT/N} - \widetilde{V}_{\widetilde{H}[m]}(2^iT/N)\right\| \leq \epsilon'. \end{equation}
Then, we have
\begin{align*}
    \left\|e^{-iH[m]nT/N} - \widetilde{V}_{\widetilde{H}[m]}(nT/N)\right\| & \leq \left\|e^{-iH[m]nT/N} - e^{-i\widetilde{H}[m]nT/N}\right\| + \left\|e^{-i\widetilde{H}[m]|n|T/N} - \widetilde{V}_{\widetilde{H}}(|n|T/N) \right\| \\
    &\leq \left\|H[m] - \widetilde{H}[m]\right\|T + \left\|e^{-i\widetilde{H}[m]T/N} - \widetilde{V}_{\widetilde{H}[m]}(T/N)\right\| \\
    &\quad+ \sum_{i=0}^{\lceil \log N\rceil - 1}\left\|e^{-i\widetilde{H}[m]n_i2^iT/N} - \widetilde{V}_{\widetilde{H}[m]}(n_i2^iT/N)\right\|\\
    &\leq \epsilon_0 T + (\lceil \log N\rceil +1)\epsilon'
\end{align*}
for all $n \in\{\pm 1, \dots, \pm N\}$, where the second inequality follows from Eq.~\eqref{V(|n|)} and the fact that $\|V_{\widetilde{H}[m]}(t)\| \leq 1$ for any $t \in \mathbb{R}$. Thus, choosing $\epsilon_0 = \epsilon_1/(2T)$ and $\epsilon' = \epsilon_1/[2(\lceil \log N\rceil +1)]$ gives the desired bound.

Implementing $\sel[c] U_{\widetilde{H}}$ with $\epsilon_0 = \epsilon_1/(2T)$ in Lemma~\ref{lemma:sel_UtildeH} requires $\mathcal{O}(\log(\alpha T/\epsilon_1)\mathcal{M}(\log(\alpha T/\epsilon_1)))$ elementary gates and one query to each of $O_{H_0}$ and $O_{H_1}$. Note that adding a constant number of controls to $\sel[c] U_{\widetilde{H}}$ does not change the asymptotic complexity (in fact, it can be observed from Fig.~\ref{fig:selUHtilde} that controls only need to be added to the $\sel O_H$ operation).

By applying the qubitisation algorithm of Refs.~\cite{Low2017,Low2019} to $\sel[c] U_{\widetilde{H}}$, the error bound in Eq.~\eqref{qubitisationerror} with $\epsilon' = \epsilon_1/[2(\lceil\log N\rceil + 1)]$ can be achieved for each $i \in \{0,\dots, \lceil \log N\rceil -1\}$ with failure probability $\mathcal{O}(\epsilon_1/\log N)$ using $\mathcal{O}(\alpha 2^i T/N + \log[(\log N)/\epsilon_1])$ calls to (controlled-)$\sel[c] U_{\widetilde{H}}$ along with $\mathcal{O}(n_{\mathrm{anc}}\{\alpha 2^i T/N + \log[(\log N)/\epsilon_1]\})$ additional elementary gates. 
Here, $n_\mathrm{anc}$ denotes the size (in terms of the number of qubits) of the ancillary space in the block-encodings $U_{\widetilde{H}[m]}$. This ancillary space corresponds to the registers labelled $a'$ and $a$ in the circuit diagram for $\sel[c] U_{\widetilde{H}}$ of Fig.~\ref{fig:selUHtilde}. Since $a'$ is a single-qubit register, $n_{\mathrm{anc}} = n_a + 1$. Summing over $i$, the total number of calls made to $\sel[c] U_{\widetilde{H}}$ is $q_U = \mathcal{O}(\alpha T + \log[(\log N)/\epsilon_1]\log N)$ and the number of additional gates required by the qubitisation algorithm is $g_{\mathcal{Q}} = \mathcal{O}(n_a\{\alpha T + \log[(\log N)/\epsilon_1]\log N\})$. Since each application of $\sel[c] U_{\widetilde{H}}$ makes one query to each of $O_{H_0}$ and $O_{H_1}$ and the gate complexity is $g_U = \mathcal{O}(\log(\alpha T/\epsilon_1)\mathcal{M}(\log(\alpha T/\epsilon_1)))$, it follows that implementing $\sel[c] \sel[d] \widetilde{V}_{\widetilde{H}}$ as in Fig.~\ref{fig:selselV_pm} uses $\mathcal{O}(\alpha T + \log[(\log N)/\epsilon_1]\log N)$ queries to $O_{H_0}$ and $O_{H_1}$, and \[ q_U g_U + g_\mathcal{Q} = \mathcal{O}\left(\left[n_a + \log(\alpha T/\epsilon_1)\mathcal{M}(\log(\alpha T/\epsilon_1))\right]\{\alpha T + \log[(\log N)/\epsilon_1]\log N\}\right) \] elementary gates.

By Lemma~\ref{lemma:sel_UtildeH}, 
$\sel U_{\widetilde{H}}$ is supported on $n_s + n_a + \ceil{\log M} + 1$ qubits, and $\mathcal{O}(\log^2(\log(\alpha T/\epsilon_1)))$ ancillae are required for its implementation. These ancillae are initialised in and returned to $\ket{0}$, so the same ancillae can be used in all of the applications of $\sel U_{\widetilde{H}}$. The qubitisation algorithm uses $\mathcal{O}(1)$ extra ancillae. The total number of qubits used to implement $\sel[c]\sel[d] \widetilde{V}_{\widetilde{H}}$ is therefore $n_s + n_a + \ceil{\log M} + \ceil{\log N} +  \mathcal{O}(\log^2(\alpha T/\epsilon_1))$ ($\sel[c]\sel[d] \widetilde{V}_{\widetilde{H}}$ acts on $n_s + \ceil{\log N} + \ceil{\log N}$ of these qubits, and the rest are ancillae that are reset to their initial states).

Since each $\sel[c] \widetilde{V}_{\widetilde{H}}(2^iT/N)$ in Fig.~\ref{fig:selselH_segmentize} fails independently with probability $\mathcal{O}(\epsilon_1/\log N)$ and there are $2(\lceil \log N\rceil + 1)$ such operations in the circuit of Fig.~\ref{fig:selselV_pm} for $\sel[c]\sel[d] \widetilde{V}_{\widetilde{H}}$, the overall failure probability of $\sel[c]\sel[d] \widetilde{V}_{\widetilde{H}}$ is $\mathcal{O}(\epsilon_1)$. 
\end{proof}

\subsection{Full implementation} \label{sec:full_oracle}

We can now combine the subroutines described in the previous subsection and in Appendix~\ref{sec:prepareW} to assemble the circuits for finite-precision approximations to the requisite oracles $O_{D}^{(j)}$ [cf.~Eq.~\eqref{O_D^j}]. Specifically, we use Lemmas~\ref{lem:selselV} and~\ref{lem:W} to construct block-encodings $O_{\widetilde{D}}^{(j)}$ of $\widetilde{D}_{\Delta,T,N}(s)$ that satisfy the form of the input oracle required by the truncated Dyson series algorithm of~\cite{Low2018} and ensure that $\|D_{\Delta,T,N}(s) - \widetilde{D}_{\Delta,T,N}(s)\|$ is upper-bounded by $\delta$ for all relevant values of $s$, for any given $\delta > 0$. 

\begin{lemma} \label{lem:oracle} Let
$D_{\Delta,T,N}^{(j)}(s) \coloneqq D_{\Delta,T,N}((j-1)\tau + s)$ for $s \in [0,\tau]$, 
where $D_{\Delta,T,N}(s)$ is defined as in Eq.~\eqref{D_DTN} with $H(s) = (1-s)H_0 + sH_1$. For any $\tau \in [0,1]$, $j\in [1,1/\tau]$, and $\delta > 0$, a $n_s + n_b + \ceil{\log M} + \ceil{\log N}$-qubit operator $O_{\widetilde{D}}^{(j)}$ such that
\[ (\bra{0} \otimes I)O_{\widetilde{D}}^{(j)}(\ket{0}\otimes I) = \sum_{m=0}^{M-1}\ket{m}\bra{m} \otimes \frac{\widetilde{D}^{(j)}_{\Delta,T,N}(m\tau/M)}{\|D\|_{\max}},  \]
where $\|D\|_{\max} \geq \max_{s\in[0,1]}\|D_{\Delta,T,N}(s)\|$ and 
\[ \left\|D^{(j)}_{\Delta,T,N}(m\tau/M) - \widetilde{D}^{(j)}_{\Delta,T,N}(m\tau/M)\right\| \leq \delta \]
for all $m \in \{0,\dots, M-1\}$, 
can be implemented with probability $1- \mathcal{O}(\Delta\delta/\beta)$ using
\begin{itemize}
\item  $\mathcal{O}\left(\alpha T + \log\left(\frac{\beta\log N}{\Delta\delta}\right)\log N\right)$ queries to each of $O_{H_0}$ and $O_{H_1}$, one query to $O_{H'}$,
\item $g_V + g_W$ elementary gates, and
\item $n_a + \mathcal{O}\left(\log^2\left(\frac{\alpha\beta T}{\Delta\delta}\right) + \left[\Delta^2 T^2 + \log\left(\frac{\beta N}{\Delta\delta}\right)\right]^2\right)$ ancilla qubits (initialised in and reset to $\ket{0}$),
\end{itemize} 
where 
\begin{equation} \label{g_V} g_V \in \mathcal{O}\left(\left[n_a + \log\left(\frac{\alpha \beta T}{\Delta \delta}\right)\mathcal{M}\left(\log\left(\frac{\alpha\beta T}{\Delta\delta}\right) \right)\right]\left[\alpha T + \log\left(\frac{\beta\log N}{\Delta \delta}\right)\log N\right]\right) \end{equation}
and 
\begin{equation} \label{g_W} g_W \in \mathcal{O}\left(\left[\Delta^2 T^2 + \log\left(\frac{\beta N}{\Delta\delta}\right) \right]\mathcal{M}\left(\Delta^2 T^2 + \log\left(\frac{\beta N}{\Delta\delta}\right)\right)\log N\right). \end{equation}
\end{lemma}
\begin{proof} We implement $O_{\widetilde{D}}^{(j)}$ using the circuit in Fig.~\ref{fig:outline_approx}, which is an approximate version of the idealised circuit for $O_D^{(j)}$ in Fig~\ref{fig:outline}. The main components are the oracle $O_{H'}$ [cf.~subsection~\ref{sec:inputmodel}], the state preparation unitary $\widetilde{W}_{\Delta,T,N}$ synthesised in Lemma~\ref{lem:W} (in Appendix~\ref{sec:prepareW}, and the multiply-controlled operation $\sel[c]\sel[d] \widetilde{V}_{\widetilde{H}}$ constructed in Lemma~\ref{lem:selselV}. Like in Lemma~\ref{lem:selselV}, the $(\ceil{\log N}+1)$-qubit $d$ register encodes the states $\ket{n}$ for $n \in \{\pm 1, \dots, \pm N\}$ by storing the sign of $n$ in the first qubit (labelled $d_0$ in Fig.~\ref{fig:outline_approx}) and the binary representation of $|n|-1$ in the remaining $\ceil{\log N}$ qubits (labelled $d_1$). 

\begin{center}
\captionsetup{type=figure}
\[ {\small \Qcircuit @C=1em @R=2.25em {
&\lstick{c} &/\qw &\multigate{3}{O_{\widetilde{D}}^{(j)}} &\qw \\
&\lstick{d} &/\qw &\ghost{O_{\widetilde{D}}^{(j)}} &\qw \\
&\lstick{b} &/\qw &\ghost{O_{\widetilde{D}}^{(j)}} &\qw \\
&\lstick{s} &/\qw &\ghost{O_{\widetilde{D}}^{(j)}} &\qw
} \qquad \raisebox{-4.7em}{=}  \qquad\quad \Qcircuit @C=1em @R=1em {
&&\lstick{c} &/\qw &\qw &\qw &\cbox{m} \qwx[1] &\qw &\cbox{m} \qwx[1] &\qw &\qw \\
\lstick{\raisebox{-3.1em}{$d$}\enspace} &&\lstick{d_0} &\qw &\gate{\textsc{Had}} &\gate{Z} &\multimeasure{1}{\hspace{-0.2em} n\hspace{-0.2em}}  &\qw &\multimeasure{1}{\hspace{-0.2em} n\hspace{-0.2em}} &\gate{\textsc{Had}} &\qw \\
&&\lstick{d_1} &/\qw &\gate{\widetilde{W}_{\Delta,T,N}} &\qw &\ghost{\hspace{-0.2em} n\hspace{-0.2em}} \qwx[2] &\qw &\ghost{\hspace{-0.2em} n\hspace{-0.2em}} \qwx[2] &\gate{\widetilde{W}_{\Delta,T,N}^\dagger} &\qw\\
&&\lstick{b} &/\qw &\qw &\qw &\qw &\multigate{1}{O_{H'}} &\qw &\qw &\qw \\
&&\lstick{s} &/\qw &\qw &\qw &\gate{\widetilde{V}_{\widetilde{H}[m]}\left(n\frac{T}{N}\right)} &\ghost{O_{H'}} &\gate{\widetilde{V}_{\widetilde{H}[m]}\left(n\frac{T}{N}\right)^\dagger} &\qw &\qw 
\gategroup{2}{1}{3}{1}{0.3em}{\{}}}
\]
\captionof{figure}{Implementation outline for the approximation $O_{\widetilde{D}}^{(j)}$ to $O_{D}^{(j)}$ (compare to Fig.~\ref{fig:outline}).} \label{fig:outline_approx}
\end{center}

Recall from Eq.~\eqref{O_H'} that the oracle $O_{H'}$ block-encodes $H'$: $(\bra{0}_b\otimes I_s)O_H'(\ket{0}_b\otimes I_s) = H'/\beta$, where $\beta$ is an upper bound on $\|H'\|$. Thus, $O_{\widetilde{D}}^{(j)}$ is a block-encoding of
\begin{equation} \label{ODjtilde_BE} (\bra{0}_{bd} \otimes I_{cs})O_{\widetilde{D}}^{(j)}(\ket{0}_{bd}\otimes I_{cs}) = \sum_{m=0}^{M-1}\ket{m}\bra{m}_c \otimes \frac{\widetilde{D}^{(j)}_{\Delta,T,N}(m\tau/M)}{2\mathcal{N}_{\Delta,T} \beta}, \end{equation}
where we define
\begin{equation} \label{Dj_DTNtilde}
\begin{aligned} 
&\frac{\widetilde{D}^{(j)}_{\Delta,T,N}(m\tau/M)}{2\mathcal{N}_{\Delta,T}} \\&\quad\coloneqq (\bra{+}_{d_0}\bra{\widetilde{W}_{\Delta,T,N}}_{d_1} \otimes I_s)\sum_{\substack{n=-N\\n\neq 0}}^{N}\ket{n}\bra{n}_d \otimes \widetilde{V}_{\widetilde{H}[m]}(nT/N)^\dagger H' \widetilde{V}_{\widetilde{H}[m]}(nT/N)(\ket{-}_{d_0}\ket{\widetilde{W}_{\Delta,T,N}}_{d_1} \otimes I_s)
\end{aligned}
\end{equation}
for each $m \in \{0,\dots, M-1\}$. Here, $\ket{\widetilde{W}_{\Delta,T,N}} \coloneqq \widetilde{W}_{\Delta,T,N}\ket{0}$, and 
\begin{equation} \label{mathcalN} N_{\Delta,T} \coloneqq \int_0^T dt\, W_{\Delta}(t) \leq \frac{1}{\sqrt{2\pi}\Delta}. \end{equation}
The upper bound on $N_{\Delta,T}$ follows from Eqs.~\eqref{eq:WDelta} and~\eqref{eq:wDelta}. It is clear from the definition of $D_{\Delta,T,N}(s)$ in Eq.~\eqref{D_DTN} that $2\mathcal{N}_{\Delta,T}\beta \geq \|D_{\Delta,T,N}(s)\|$ for all $s \in [0,1]$. 

To bound $\|D_{\Delta,T,N}^{(j)}(m\tau/M) - \widetilde{D}_{\Delta,T,N}^{(j)}(m\tau/M)\|$, observe from Eq.~\eqref{Ds} that
\begin{equation} \label{Dj_DTN}
\begin{aligned} 
&\frac{{D}^{(j)}_{\Delta,T,N}(m\tau/M)}{2\mathcal{N}_{\Delta,T}} \\&\quad\coloneqq(\bra{+}_{d_0}\bra{{W}_{\Delta,T,N}}_{d_1} \otimes I_s)\sum_{\substack{n=-N\\n\neq 0}}^{N}\ket{n}\bra{n}_d \otimes e^{iH[m]nT/N} H' e^{-iH[m]nT/N}(\ket{-}_{d_0}\ket{{W}_{\Delta,T,N}}_{d_1} \otimes I_s),
\end{aligned}
\end{equation}
recalling the notational convention [cf.~Eq.~\eqref{H[m]}] that $H[m] \coloneqq H((j-1+m/M)\tau)$. Using the fact that 
$\|(\bra{\psi}\otimes I)A(\ket{\varphi}\otimes I) - (\bra{\widetilde{\psi}}\otimes I)\widetilde{A}(\ket{\widetilde{\varphi}}\otimes I)\| \leq \|\ket{\psi} - \ket{\widetilde{\psi}}\|\|\widetilde{A}\| + \|A - \widetilde{A}\| +\|A\|\|\ket{\varphi} - \ket{\widetilde{\varphi}}\|$
for any operators $A,\widetilde{A}$ and states $\ket{\psi},\ket{\varphi},\ket{\widetilde{\psi}}, \ket{\widetilde{\varphi}}$, it follows from Eqs.~\eqref{Dj_DTNtilde} and~\eqref{Dj_DTN} that
\begin{align*}
&\left\|D_{\Delta,T,N}^{(j)}(m\tau/M) - \widetilde{D}_{\Delta,T,N}^{(j)}(m\tau/M)\right\| \\
&\quad \leq 4\mathcal{N}_{\Delta,T}\|H'\|\left(\max_n \left\|e^{-iH[m]nT/N} - \widetilde{V}_{\widetilde{H}[m]}(nT/N) \right\| + \left\|\ket{W_{\Delta,T,N}} - \ket{\widetilde{W}_{\Delta,T,N}} \right\| \right).
\end{align*}
Since $\|H'\| \leq \beta$ and $\mathcal{N}_{\Delta,T} \leq 1/(2\Delta)$, this error is at most $\delta$ if each of the terms in the parentheses is upper-bounded by $\Delta \delta/(4\beta)$. 
We now use Lemmas~\ref{lem:selselV} and~\ref{lem:W} to determine the cost of constructing $\sel[c]\sel[d] \widetilde{V}_{\widetilde{H}}$ and $\widetilde{W}_{\Delta,T,N}$ to satisfy these bounds. Setting $\epsilon_1 = \Delta\delta/(4\beta)$ in Lemma~\ref{lem:selselV}, $\sel[c]\sel[d]\widetilde{V}_{\widetilde{H}}$ can be implemented such that for all $m \in \{0,\dots, M-1\}$ and $n \in \{\pm 1,\dots, \pm N\}$, 
\[ \left\|e^{-iH[m]nT/N} - \widetilde{V}_{\widetilde{H}[m]}(nTN)\right\| \leq \frac{\Delta \delta}{4\beta} \]
with probability $1 - \mathcal{O}(\Delta\delta/\beta)$, using $\mathcal{O}\left(\alpha T + \log\left(\frac{\beta \log N}{\Delta \delta}\right)\log N\right)$ queries to each of $O_{H_0}$ and $O_{H_1}$ along with $\mathcal{O}\left(\left[n_a + \log\left(\frac{\alpha \beta T}{\Delta \delta}\right)\mathcal{M}\left(\log\left(\frac{\alpha\beta T}{\Delta\delta}\right) \right)\right]\left[\alpha T + \log\left(\frac{\beta\log N}{\Delta \delta}\right)\log N\right]\right)$ elementary gates and $n_a + \mathcal{O}\left(\log^2\left(\frac{\alpha \beta T}{\Delta \delta}\right)\right)$ ancilla qubits. Taking $\epsilon_2 = \Delta \delta/(4\beta)$ in Lemma~\ref{lem:W}, $\widetilde{W}_{\Delta,T,N}$ can be implemented such that
\[ \left\|\ket{W_{\Delta,T,N}} - \ket{\widetilde{W}_{\Delta,T,N}} \right\| \leq \frac{\Delta \delta}{4\beta}\]
using $\mathcal{O}\left(\left[\Delta^2 T^2 + \log\left(\frac{\beta N}{\Delta\delta}\right) \right]\mathcal{M}\left(\Delta^2 T^2 + \log\left(\frac{\beta N}{\Delta\delta}\right)\right)\log N\right)$ gates and $\mathcal{O}([\Delta^2 T^2 + \log(N/\delta)]^2)$ ancilla qubits. 
In addition, the circuit for $O_{\widetilde{D}}^{(j)}$ makes one query to $O_{H'}$ and applies a constant number of single-qubit gates. Since $\sel[c] \sel[d] \widetilde{V}_{\widetilde{H}}$ are and its inverse are applied once each, the total failure probability is at most $\mathcal{O}(\Delta\delta/\beta)$. 
\end{proof}

\section{Adiabatic preparation of general eigenstates}
\label{sec:adiabatic_general}

Having explicitly constructed approximations to block-encodings of the quasi-adiabatic continuation operator $D_{\Delta,T,N}(s)$, we can now tally the costs of Algorithm~\ref{alg1} and prove Theorem~\ref{thm:main1}. Before delving into the technical details, we recapitulate the general idea. As shown in Lemma~\ref{lem:oracle} in the previous section, we can implement operators $O_{\widetilde{D}}^{(j)}$ that block-encode an approximation $\widetilde{D}_{\Delta,T,N}(s)$ to $D_{\Delta,T,N}(s)$ at certain discrete values of $s$. The error in this approximation can be arbitrarily suppressed, and the complexity of implementing the block-encodings $O_{\widetilde{D}}^{(j)}$ is polylogarithmic in the inverse of the target error. Using $O_{\widetilde{D}}^{(j)}$ as the oracle inputs to the truncated Dyson series algorithm of~\cite{Low2018}, we can approximately simulate the $s$-ordered exponential of $D_{\Delta,T,N}(s)$, which in turn approximates the \emph{exact} adiabatic evolution (of the ground state) as is made precise by Theorem~\ref{Ubound}. 

A na\"ive approach would be to simply plug in the complexity of $O_{\widetilde{D}}^{(j)}$ into Theorem~3 (or Corollary~4) of~\cite{Low2018}. However, the actual analysis is slightly more involved, for the following reason. Given access to an appropriate block-encoding of a parameter-dependent self-adjoint operator $H(s)$, the algorithm of~\cite{Low2018} can implement (with high probability) the truncated and discretised Dyson series\footnote{\label{footnote2}~\cite{Low2018} actually implements a Dyson series of the same form as that of Eq.~\eqref{eq:Dyson} except with \textit{strict} time-ordering, i.e., with $m_1 < \dots < m_k$ in the inner sum. We note that the algorithm can be easily changed to implement Eq.~\eqref{eq:Dyson}, via a trivial modification to the integer comparator gadget used in the proof of Theorem~3 in~\cite{Low2018}. This would remove the $\max_s\|H(s)\|^2$ term in the scaling reported in Lemma~5 and Theorem~3/Corollary~4 of~\cite{Low2018}.}
\begin{equation} \label{eq:Dyson} \textsc{Dys}_{K,M}(H; \tau) \coloneqq \sum_{k=0}^K \left(\frac{i\tau}{M}\right)^k \sum_{0 \leq m_1\leq \dots \leq m_k < M}H(m_k \tau/M)\dots H(m_1 \tau/M)  \end{equation}
for some $\tau \in \mathbb{R}$, thereby approximating the unitary $\mathcal{S}[e^{i\int_0^\tau ds\,H(s)}]$ generated by $H(s)$. The error in the simulation is due entirely to the truncation and discretisation of the Dyson series, and can therefore be bounded by appropriately choosing the truncation order $K$ and discretisation parameter $M$. The values of $K$ and $M$ then determine the number of qubits, queries, and additional gates required. Thus, Theorem~3 of~\cite{Low2018} gives the complexity of the algorithm in terms of the desired bound on the deviation from the ideal evolution operator $\mathcal{S}[e^{i\int_0^\tau ds\,H(s)}]$.  Applying this algorithm (with $H(s) \to \widetilde{D}_{\Delta,T,N}(s)$) to the block-encoding $O_{\widetilde{D}}^{(j)}$ implements the truncated Dyson series $\textsc{Dys}_{K,M}(\widetilde{D}_{\Delta,T,N};\tau)$ of $\widetilde{D}_{\Delta,T,N}$. However, $\widetilde{D}_{\Delta,T,N}(s)$ is defined, by the unitary $O_{\widetilde{D}}^{(j)}$, only at a finite number of points (namely, at $s = m\tau/M$ for $m \in \{0,\dots, M-1\}$). Consequently, unlike in the setting of Theorem~3 of~\cite{Low2018}, it does not make sense to consider the error between the $\textsc{Dys}(\widetilde{D}_{\Delta,T,N};\tau)$ the $s$-ordered exponential of $\widetilde{D}_{\Delta,T,N}(s)$, as the latter does not exist (unless we somehow analytically extend $\widetilde{D}_{\Delta,T,N}(s)$ to the entire domain $s \in [0,\tau]$). 

To circumvent this issue, we directly bound the difference between $\textsc{Dys}_{K,M}(\widetilde{D}_{\Delta,T,N};\tau)$ and the $s$-ordered exponential of $D_{\Delta,T,N}(s)$ (which \emph{is} well-defined [cf.~Eq.~\eqref{D_DTN}]), using a straightforward triangle inequality:
\begin{equation} \label{eq:Dys_triangle}
\begin{aligned} &\left\|\mathcal{S}\left[e^{i\int_0^\tau ds\, D_{\Delta,T,N}(s)}\right] - \textsc{Dys}_{K,M}(\widetilde{D}_{\Delta,T,N};\tau)\right\| \\
&\quad \leq \left\|\mathcal{S}\left[e^{i\int_0^\tau ds\, D_{\Delta,T,N}(s)}\right] - \textsc{Dys}_{K,M}({D}_{\Delta,T,N};\tau)\right\| + \left\|\textsc{Dys}_{K,M}({D}_{\Delta,T,N};\tau) - \textsc{Dys}_{K,M}(\widetilde{D}_{\Delta,T,N};\tau) \right\|. 
\end{aligned}
\end{equation}
Given any $\epsilon >0$,~\cite{Low2018} shows how to set $K$ and $M$ so that the first term is bounded by $\epsilon$; see Theorem~\ref{thm:TDS} below. We use the following proposition to bound the second term in terms of the difference between $D_{\Delta,T,N}(s)$ and $\widetilde{D}_{\Delta,T,N}(s)$. (This will then allow us to apply Lemma~\ref{lem:oracle}, which gives the complexity of implementing block-encodings of $\widetilde{D}_{\Delta,T,N}(s)$ for any desired upper bound on $\|D_{\Delta,T,N}(s) - \widetilde{D}_{\Delta,T,N}(s)\|$.)

\begin{proposition} \label{prop:Dyson}
For any $K,M \in \mathbb{N}$, $\tau \in \mathbb{R}$, and one-parameter families of bounded operators $D_1(s)$ and $D_2(s)$, let $\textsc{Dys}_{K,M}(D_1;\tau)$ and $\textsc{Dys}_{K,M}(D_2;\tau)$ be defined by Eq.~\eqref{eq:Dyson}. Suppose that $K \leq M$ and that $\|D_1(m\tau/M)\|, \|D_2(m\tau/M)\| \leq \|D\|_{\max}$ and $\|D_1(m\tau/M) - D_2(m\tau/M)\| \leq \delta$ for all $m \in \{0,\dots, M-1\}$. Then,
\[ \left\|\textsc{Dys}_{K,M}(D_1;\tau) - \textsc{Dys}_{K,M}(D_2;\tau)\right\| \leq 2|\tau| e^{2\|D\|_{\max}|\tau|}\delta. \]
\end{proposition}
\begin{proof}
By repeatedly applying the triangle inequality and sub-multiplicativity of the operator norm, we have
\begin{align*}
&\left\|D_1(m_k\tau/M)\dots D_1(m_1\tau/M) - D_2(m_k\tau/M)\dots D_2(m_1\tau/M)\right\| \\
&\quad \leq \sum_{j=1}^k\left(\prod_{i=1}^{j-1}\|D_1(m_i\tau/M)\|\right)\|D_1(m_j\tau/M) - D_2(m_j\tau/M)\| \left(\prod_{i=j+1}^k\|D_2(m_i\tau/M)\|\right) \\
&\quad \leq k\|D\|_{\max}^{k-1}\delta.
\end{align*}
Hence, noting that the $k =0$ term in Eq.~\eqref{eq:Dyson} is the same for all operators,
\begin{align*}
&\left\|\textsc{Dys}_{K,M}(D_1;\tau) - \textsc{Dys}_{K,M}(D_2;\tau)\right\| \\
&\quad \leq \sum_{k=1}^{K}\left(\frac{|\tau|}{M}\right)^k\sum_{0 \leq m_1 \leq \dots \leq m_k < M}\left\|D_1(m_k\tau/M)\dots D_1(m_1\tau/M) - D_2(m_k\tau/M)\dots D_2(m_1\tau/M)\right\| \\
&\quad \leq \sum_{k=1}^{K}\left(\frac{|\tau|}{M}\right)^k \sum_{0 \leq m_1 \leq \dots \leq m_k < M} k \|D\|_{\max}^{k-1}\delta \\
&\quad = \delta \sum_{k=1}^K\left(\frac{|\tau|}{M}\right)^k {{M + k - 1}\choose k}k\|D\|_{\max}^{k-1} \\
&\quad \leq 2|\tau| \delta \sum_{k=1}^K\frac{(2\|D\|_{\max}|\tau|)^{k-1}}{(k-1)!} \\
&\quad\leq 2|\tau| \delta e^{2\|D\|_{\max}|\tau|},
\end{align*}
using the fact that ${{M + k - 1}\choose k} \leq (2M)^k/k!$ for $k \leq M$ to obtain the third inequality. 
\end{proof}

We also rephrase the result behind Theorem~3 of~\cite{Low2018} so that we can directly apply it in our proof of Theorem~1. The essential difference between the following statement and Theorem~3 of~\cite{Low2018} is that $D_1(s)$ and $D_2(s)$ need not be the same family of operators. 

\begin{theorem}[\cite{Low2018}] \label{thm:TDS}
For any one-parameter families of bounded operators $D_1(s)$ and $D_2(s)$, let $\textsc{Dys}_{K,M}(D_1;\tau)$ and $\textsc{Dys}_{K,M}(D_2;\tau)$ be defined by Eq.~\eqref{eq:Dyson}. Suppose that $D_1(s)$ is differentiable and that $\|D\|_{\max}$ is an upper bound on $\|D_1(s)\|$ and $\|D_2(s)\|$ for all $s$. Then, for any $\tau \leq 1/(2\|D\|_{\max})$, there exist $K,M \in \mathbb{N}$ such that
\begin{equation} \label{D1bound} \left\| \mathcal{S}\left[e^{i\int_0^\tau ds\,D_1(s)}\right] - \textsc{Dys}_{K,M}(D_1;\tau) \right\| \leq \epsilon \end{equation}
and $\textsc{Dys}_{K,M}(D_2;\tau)$ can be implemented with probability $1-\mathcal{O}(\epsilon)$ using
\begin{itemize}
\item $\mathcal{O}\left(\frac{\log(1/\epsilon)}{\log\log(1/\epsilon)}\right)$ queries to $O_{D_2}$,
\item $\mathcal{O}\left(\left[n_{\mathrm{anc}} + \log\left(\frac{\tau^2}{\epsilon}\langle\|D_1'\|\rangle\right)\right]\frac{\log(1/\epsilon)}{\log\log(1/\epsilon)}\right)$ elementary gates, and
\item $n_s + \mathcal{O}\left(n_{\mathrm{anc}} + \log\left(\frac{\tau^2}{\epsilon}\langle\|D_1'\|\rangle\right)\right)$ qubits,
\end{itemize}
where $\langle \|D'_1\|\rangle \coloneqq \frac{1}{\tau}\int_0^\tau ds\,\|D_1'(s)\|$ and $O_{D_2}$ is a $(n_s + n_{\mathrm{anc}} + \ceil{\log M})$-qubit operator for which
\begin{equation*} (\bra{0}_{\mathrm{anc}}\otimes I_s)O_{D_2}(\ket{0}_{\mathrm{anc}}\otimes I_s) = \sum_{m=0}^{M-1}\ket{m}\bra{m} \otimes \frac{D_2(m\tau/M)}{\|D\|_{\max}}. \end{equation*}
\end{theorem}
\begin{proof} Lemma~5 of~\cite{Low2018} shows that for any differentiable $D_1(s)$ and $\tau \leq 1/(2\max_s\|D_1(s)\|)$, the error $\|\mathcal{S}[e^{i\int_0^\tau ds\, D_1(s)}] - \textsc{Dys}_{K,M}(D_1;\tau)\|$ in truncating and discretising the Dyson series of $D_1(s)$ is bounded by $\epsilon$ for some $K \in \mathcal{O}\left(\frac{\log(1/\epsilon)}{\log\log(1/\epsilon)}\right)$ and $M\in \mathcal{O}\left(\frac{\tau^2}{\epsilon}\langle\|D_1'\|\rangle\right)$ (see footnote~\ref{footnote2}). The proof of Theorem~3 of~\cite{Low2018} constructs an algorithm that implements $\textsc{Dys}_{K,M}(D_2;\tau)$, for any $D_2(s)$ and $K,M \in \mathbb{N}$, using $\mathcal{O}(K)$ queries to $O_{D_2}$, $\mathcal{O}((n_\mathrm{anc} + \log K + \log M)K)$ elementary gates, and $n_s + \mathcal{O}(n_{\mathrm{anc}} + \log K + \log M)$ qubits. The claim follows from substituting in the particular values of $K$ and $M$ that ensure that Eq.~\eqref{D1bound} is satisfied (and using oblivious amplitude amplification~\cite{Berry2013} to boost the success probability to $1 - \mathcal{O}(\epsilon)$). 

\end{proof}

Finally, we integrate Lemma~\ref{lem:oracle} and Theorem~\ref{thm:TDS} with Corollary~\ref{corollary:parameters}, which specifies how to choose $\Delta$, $T$, and $N$ in terms of the target precision $\epsilon$, to prove Theorem~\ref{thm:main1}, restated below for convenience.

\adiabatic*

\begin{proof}
The proof proceeds as follows. We apply the truncated Dyson series algorithm of~\cite{Low2018} to the block-encodings $O_{\widetilde{D}}^{(j)}$ of $\widetilde{D}_{\Delta,T,N}(s)$ constructed in Lemma~\ref{lem:oracle}. Using Eq.~\eqref{eq:Dys_triangle} and Proposition~\ref{prop:Dyson}, we can bound the difference between the operator $\widetilde{U}$ that is implemented by this algorithm and the $s$-ordered exponential $U_{\Delta,T,N}(s)\coloneqq \mathcal{S}[e^{i\int_0^\tau ds\,D_{\Delta,T,N}(s)}]$ of $D_{\Delta,T,N}(s)$. By choosing $\Delta$, $T$, and $N$ according to Corollary~\ref{corollary:parameters}, $U_{\Delta,T,N}(1)\ket{\psi_k(0)}$ can be made arbitrarily close to the actual eigenstate $\ket{\psi_k(1)}$ of $H_1$. Substituting the values of $\Delta$, $T$, and $N$ prescribed by Corollary~\ref{corollary:parameters} into Lemma~\ref{lem:oracle} gives the complexity of each query to $O_{\widetilde{D}}^{(j)}$, and combining this with the number of queries (as well as gates and ancillae) required by the truncated Dyson series algorithm yields the overall cost of implementing $\widetilde{U}$.

Let $\tau = 1/\ceil{\beta/\Delta}$. We divide the interval $[0,1]$ into $1/\tau$ segments of equal length, and label the intervals by an index $j \in \{1, \dots, 1/\tau\}$.  For each interval $j$, an operator $O_{\widetilde{D}}^{(j)}$ that block-encodes $\sum_{m=0}^{M-1}\ket{m}\bra{m}\otimes \widetilde{D}^{(j)}_{\Delta,T,N}(m\tau/M)/(2\mathcal{N}_{\Delta,T}\beta)$ can be implemented via Lemma~\ref{lem:oracle} such that
\begin{equation}\label{Djtilde_bound} \left\|D_{\Delta,T,N}^{(j)}(m\tau/M) - \widetilde{D}_{\Delta,T,N}^{(j)}(m\tau/M)\right\| \leq \delta \end{equation}
for all $m \in \{0,\dots, M-1\}$, where $D^{(j)}_{\Delta,T,N}(s) \coloneqq D_{\Delta,T,N}((j-1)\tau + s)$ is the restriction of $D_{\Delta,T,N}(s)$ to interval $j$. The operators $O_{\widetilde{D}}^{(j)}$ can be used as the oracle inputs to the algorithm of~\cite{Low2018} to implement a sequence of truncated Dyson series, of the form
\[ \widetilde{U} \coloneqq \prod_{j=1}^{1/\tau}\textsc{Dys}_{K,M}(\widetilde{D}^{(j)}_{\Delta,T,N}; \tau), \]
with $\textsc{Dys}_{K,M}(\cdot; \cdot)$ defined by Eq.~\eqref{eq:Dyson}. Each $\textsc{Dys}_{K,M}(\widetilde{D}_{\Delta,T,N}^{(j)};\tau)$ has operator norm at most $1$, by construction. Suppose that the integers $K$ and $M$ are chosen so that 
\begin{equation} \label{Dj_bound} \left\|\mathcal{S}\left[e^{i\int_0^\tau ds\, D^{(j)}_{\Delta,T,N}(s)}\right] - \textsc{Dys}_{K,M}(D^{(j)}_{\Delta,T,N};\tau)\right\| \leq \tau\delta \end{equation}
for each $j \in \{1,\dots, 1/\tau\}$. 
This involves taking $K^2 \leq M$, among other constraints [cf.~Lemma~5 of~\cite{Low2018}]. Thus, the assumption that $K \leq M$ of Proposition~\ref{prop:Dyson} is satisfied, and we can apply Proposition~\ref{prop:Dyson} and Eqs.~\eqref{Djtilde_bound} and~\eqref{Dj_bound} to obtain
\begin{align*}
\left\|U_{\Delta,T,N}(1) - \widetilde{U}\right\| &= \left\|\prod_{j=1}^{1/\tau} \mathcal{S}\left[e^{i\int_0^\tau ds\, D^{(j)}_{\Delta,T,N}(s)}\right] - \prod_{j=1}^{1/\tau}\textsc{Dys}_{K,M}(\widetilde{D}^{(j)}_{\Delta,T,N};\tau) \right\| \\
& \leq \sum_{j=1}^{1/\tau}\left(\left\|\mathcal{S}\left[e^{i\int_0^\tau ds\, D^{(j)}_{\Delta,T,N}(s)}\right] - \textsc{Dys}_{K,M}(D^{(j)}_{\Delta,T,N};\tau)\right\| \right. \\
&\quad + \left. \left\|\textsc{Dys}_{K,M}({D}^{(j)}_{\Delta,T,N};\tau) - \textsc{Dys}_{K,M}(\widetilde{D}^{(j)}_{\Delta,T,N};\tau) \right\|\right) \\
&\leq \frac{1}{\tau}\left(\tau\delta + 2\tau \delta e^{\beta \tau/\Delta}\right) \\
&\leq (1 + 2e)\delta.
\end{align*}
The second inequality uses the fact that for all $j$ and $m$, $\beta/(2\Delta)$ is an upper bound on both $\|D^{(j)}_{\Delta,T,N}(m\tau/M)\|$ and $\|\widetilde{D}^{(j)}_{\Delta,T,N}(m\tau/M)\|$ [cf.~Eqs.~\eqref{mathcalN} and~\eqref{Dj_DTN}], and the third inequality follows from $\tau \leq \Delta/\beta$. Hence, 
\begin{equation} \label{thm1bound1}
\left\|U_{\Delta,T,N}(1) - \widetilde{U}\right\| \leq \frac{\epsilon}{2}
\end{equation}
for some $\delta \in\Theta(\epsilon)$.

Before calculating the cost of implementing the operators $O_{\widetilde{D}}^{(j)}$ and the truncated Dyson series algorithm~\cite{Low2018} such that Eqs.~\eqref{Djtilde_bound} and~\eqref{Dj_bound} are satisfied with $\delta \in \Theta(\epsilon)$, we fix the values of the parameters $\Delta$, $T$, and $N$. As shown in Corollary~\ref{corollary:parameters}, we can ensure that
\begin{equation} \label{thm1bound2} \left\|\ket{\psi_k(1)} - U_{\Delta,T,N}(1)\ket{\psi_k(0)}\right\| \leq \frac{\epsilon}{2} \end{equation}
via suitable choices of $\Delta \in \Theta\left(\gamma\log^{-1/2}\left(\frac{\beta}{\gamma\epsilon}\right)\right)$, $T \in \Theta\left(\frac{1}{\gamma}\log\left(\frac{\beta}{\gamma\epsilon}\right)\right)$, and $N \in \Theta\left(\frac{\alpha\beta}{\gamma^2\epsilon}\log^{3/2}\left(\frac{\beta}{\gamma\epsilon}\right)\right)$. Then, Eqs.~\eqref{thm1bound1} and~\eqref{thm1bound2} would imply that
\begin{align*}
\left\|\ket{\psi_k(1)} - \widetilde{U}\ket{\psi_k(0)}\right\| &\leq \left\|\ket{\psi_k(1)} - U_{\Delta,T,N}(1)\ket{\psi_k(0)} \right\| + \left\|U_{\Delta,T,N}(1) - \widetilde{U}\right\| \leq \epsilon,
\end{align*}
as required.

We now use Lemma~\ref{lem:oracle} to determine the query, gate, and space complexity of each $O_{\widetilde{D}}^{(j)}$ with these choices for the parameters $\Delta$, $T$, and $N$. By Lemma~\ref{lem:oracle}, each $O_{\widetilde{D}}^{(j)}$ can be constructed such that Eq.~\eqref{Djtilde_bound} holds with $\delta \in \Theta(\epsilon)$ for all $m \in \{0,\dots, M-1\}$ with probability $1 - \mathcal{O}({\tau\epsilon}/{\log[1/(\tau\epsilon)]})$ using $\mathcal{O}\left(\left[\frac{\alpha}{\gamma} + \log\left(\frac{1}{\epsilon}\right)\right]\log\left(\frac{\beta}{\gamma\epsilon}\right)\right)$ queries to $O_{H_0}$ and $O_{H_1}$, and a single query to $O_{H'}$. For the gate complexity, substituting for $\Delta$, $T$, and $N$ in Eqs.~\eqref{g_V} and~\eqref{g_W} gives 
\[ g_V \in \mathcal{O}\left(\left[n_a + \log\left(\frac{\alpha\beta}{\gamma^2\epsilon}\right)\mathcal{M}\left(\log\left(\frac{\alpha\beta}{\gamma^2\epsilon}\right)\right)\right]\left[\frac{\alpha}{\gamma} + \log\left(\frac{1}{\epsilon}\right)\right]\log\left(\frac{\beta}{\gamma\epsilon}\right)\right)\]
and
\[ g_W \in \mathcal{O}\left(\log^2\left(\frac{\alpha\beta}{\gamma^2\epsilon}\right)\mathcal{M}\left(\log\left(\frac{\alpha\beta}{\gamma^2\epsilon}\right)\right)\right), \]
so the total number  $g_V + g_W$ of elementary gates required is dominated by $g_V$. Each $O_{\widetilde{D}}^{(j)}$ is supported on $n_s + n_b + \ceil{\log M} + \mathcal{O}\left(\log\left(\frac{\alpha\beta}{\gamma^2\epsilon}\right)\right)$ qubits, and its implementation uses $n_a + \mathcal{O}\left(\log^2\left(\frac{\alpha\beta}{\gamma^2\epsilon}\right)\right)$ ancillae (all of which are returned to their initial states). 

It remains to take into account the number of calls to each $O_{\widetilde{D}}^{(j)}$ and the number of additional gates and ancillae required by the truncated Dyson series algorithm~\cite{Low2018} so as to satisfy Eq.~\eqref{Dj_bound} with $\delta \in \Theta(\epsilon)$. Note that the length $\tau = 1/\ceil{\beta/\Delta}$ of each segment is no greater than $1/(2\mathcal{N}_{\Delta,T}\beta)$, since $\mathcal{N}_{\Delta,T} \leq 1/(2\Delta)$ [cf.~Eq.~\eqref{mathcalN}]. We can therefore apply Theorem~\ref{thm:TDS} with $D_1(s) = D^{(j)}_{\Delta,T,N}(s)$, $D_2(s) = \widetilde{D}^{(j)}_{\Delta,T,N}(s)$, and $\|D\|_{\max} = 2\mathcal{N}_{\Delta,T}\beta$ to see that there exist values for the parameters $K$ and $M$ for which 1) Eq.~\eqref{Dj_bound} holds with $\delta \in \mathcal{O}(\epsilon)$ and 2)~$\textsc{Dys}_{K,M}(\widetilde{D}_{\Delta,T,N}^{(j)};\tau)$ can be implemented with $\mathcal{O}\left(\frac{\log[1/(\tau\epsilon)]}{\log\log[1/(\tau\epsilon)]}\right)$ calls to $O_{\widetilde{D}}^{(j)}$. $\widetilde{U}$ is the product of $\textsc{Dys}_{K,M}(\widetilde{D}_{\Delta,T,N}^{(j)};\tau)$ over $j \in \{1,\dots, 1/\tau\}$, so $\mathcal{O}\left(\frac{1}{\tau}\frac{\log[1/(\tau\epsilon)]}{\log\log[1/(\tau\epsilon)]}\right)$ calls to some $O_{\widetilde{D}}^{(j)}$ are made in total. Multiplying the number of calls by the complexity of each call and using $1/\tau \in \Theta(\beta/\Delta) = \Theta\left(\frac{\beta}{\gamma}\log^{1/2}\left(\frac{\beta}{\gamma\epsilon}\right)\right)$, the number of queries to $O_{H_0}$ and $O_{H_1}$ is $\mathcal{O}\left(\frac{\beta}{\gamma}\left[\frac{\alpha}{\gamma} + \log\left(\frac{1}{\epsilon}\right)\right]\frac{\log^{2.5}\left(\frac{\beta}{\gamma\epsilon}\right)}{\log\log\left(\frac{\beta}{\gamma\epsilon}\right)}\right)$ and the number of queries to $O_{H'}$ is $\mathcal{O}\left(\frac{\beta}{\gamma}\frac{\log^{1.5}\left(\frac{\beta}{\gamma\epsilon}\right)}{\log\log\left(\frac{\beta}{\gamma\epsilon}\right)}\right)$, as claimed. The gate complexity due to the calls to $O_{\widetilde{D}}^{(j)}$ is \[ g_1 \in \mathcal{O}\left(\frac{\beta}{\gamma}\left[n_a + \log\left(\frac{\alpha\beta}{\gamma^2\epsilon}\right)\mathcal{M}\left(\log\left(\frac{\alpha\beta}{\gamma^2\epsilon}\right)\right)\right]\left[\frac{\alpha}{\gamma} + \log\left(\frac{1}{\epsilon}\right)\right]\frac{\log^{2.5}\left(\frac{\beta}{\gamma\epsilon}\right)}{\log\log\left(\frac{\beta}{\gamma\epsilon}\right)}\right). \] In addition, according to Theorem~\ref{thm:TDS}, using the truncated Dyson series algorithm to implement all of the $\textsc{Dys}_{K,M}(\widetilde{D}^{(j)}_{\Delta,T,N};\tau)$ requires another $g_2 \in \mathcal{O}\left(\frac{1}{\tau}\left[n_{\mathrm{anc}} + \log\left(\frac{\tau}{\epsilon}\langle\|D'_{\Delta,T,N}\|\rangle\right)\right]\frac{\log[1/(\tau\epsilon)]}{\log\log[1/(\tau\epsilon)]}\right)$ elementary gates. Here, $n_{\mathrm{anc}}$ denotes the size (in terms of the number of qubits) of the ancillary space of the block-encodings $O_{\widetilde{D}}^{(j)}$; by Eq.~\eqref{ODjtilde_BE}, $n_{\mathrm{anc}} = n_b + \ceil{\log N} = n_b + \Theta\left(\log\left(\frac{\alpha\beta}{\gamma^2\epsilon}\right)\right)$. It can be shown from the definition of $D_{\Delta,T,N}(s)$ in Eq.~\eqref{D_DTN} that $\|D_{\Delta,T,N}'(s)\| \leq 2\beta^2T/\Delta$ for all $s$, so we can replace $\langle\|D'_{\Delta,T,N}\|\rangle$ with $2\beta^2T/\Delta \in \Theta\left(\frac{\beta^2}{\gamma^2}\log^{3/2}\left(\frac{\beta}{\gamma\epsilon}\right)\right)$. Thus, 
\[ g_2 \in \mathcal{O}\left(\frac{\beta}{\gamma}\left[n_b + \log\left(\frac{\alpha\beta}{\gamma^2\epsilon}\right)\right]\frac{\log^{1.5}\left(\frac{\beta}{\gamma\epsilon}\right)}{\log\log\left(\frac{\beta}{\gamma\epsilon}\right)}\right), \]
and the total gate complexity is $g_1 + g_2$. 

It also follows from Theorem~\ref{thm:TDS} that the truncated Dyson series algorithm uses a total of $n_s + \mathcal{O}\left(n_b + \log\left(\frac{\alpha\beta}{\gamma^2\epsilon}\right)\right)$ qubits. Adding this to the $n_a + \mathcal{O}\left(\log^2\left(\frac{\alpha\beta}{\gamma^2\epsilon}\right)\right)$ ancillae required for implementing the $O_{\widetilde{D}}^{(j)}$, the total number of qubits is $n_s + n_a + \mathcal{O}\left(n_b + \log^2\left(\frac{\alpha\beta}{\gamma^2\epsilon}\right)\right)$. Each $O_{\widetilde{D}}^{(j)}$ fails with probability at most $\mathcal{O}(\tau\epsilon/\log[(1/\tau\epsilon)])$, so the probability that at least one of the $\mathcal{O}\left(\frac{1}{\tau}\frac{\log[1/(\tau\epsilon)]}{\log\log[1/(\tau\epsilon)]}\right)$ calls to $O_{\widetilde{D}}^{(j)}$ fails is $p \in \mathcal{O}(\epsilon)$. Independently, the truncated Dyson series algorithm fails with probability at most $q\in \mathcal{O}(\epsilon)$. Therefore, the overall failure probability is $q + (1-q)p \in \mathcal{O}(\epsilon)$. 

\end{proof}

\section{Improving the gap dependence}
\label{sec:improved_gap_dependence}

The algorithm of Theorem~\ref{thm:main1} requires only a lower bound $\gamma$ on the size of the gap between the eigenstate of interest and the rest of the spectrum, that holds for all $s \in [0,1]$. In general, however, the size of the gap may vary substantially as a function of $s$, and $\gamma$ may be a very loose bound for most values of $s$. If given more refined knowledge of the gap along the adiabatic path, our algorithm can be applied in a way that takes advantage of this additional information. 

To demonstrate the idea, we consider a simplified setting in which a lower bound $\gamma(s)$ on the gap is known \textit{a priori} for all $s \in [0,1]$.\footnote{As shown in~\cite{Jarret2018}, a tight \textit{a priori} bound for all $s$ is usually far from necessary. We remark that the algorithm of Theorem~\ref{thm:main1} can be straightforwardly adapted to the framework of~\cite{Jarret2018}; this would exponentially improve the error dependence of the algorithms of~\cite{Jarret2018}, but is beyond the scope of this paper.} More precisely, if we are interested in preparing the $k$th eigenstate $\ket{\psi_k(1)}$ of $H_1$ and the eigenvalues of $H(s) \coloneqq (1-s)H_0 + sH_1$ are denoted by $E_0(s), E_1(s), \dots$, the promise is that $0 < \gamma(s) \leq |E_k(s) - E_{k\pm 1}(s)|$ for all $s \in [0,1]$. 
We divide the interval $[0,1]$ into $q$ segments of variable lengths, and we label the endpoints of these segments by $s_0, \dots, s_q$ where $0 = s_0 < \dots < s_q = 1$. The length of each segment (and hence the total number of segments $q$) is determined as follows. For each $i \in \{0, \dots, q\}$, define 
\[ \gamma_i \coloneqq \min_{s \in [s_i, s_{i+1}]}\gamma(s), \]
and note that $\|H(s_i)\| \leq \alpha$ and $\|H(s_{i+1}) - H(s_i)\| = (s_{i+1} - s_i)\|H_1 - H_0\| \leq (s_{i+1} - s_i)\beta$, where $\alpha$ and $\beta$ are upper bounds on $\max\{\|H_0\|, \|H_1\|\}$ and $\|H_1 - H_0\|$, respectively. Moreover, as shown in Appendix~\ref{sec:oracle_linear_combination}, block-encodings $O_{H(s_i)}$ and $O_{H(s_{i+1})}$ of $H(s_i)$ and $H(s_{i+1})$ can be constructed using one query to each of $O_{H_0}$ and $O_{H_1}$. Thus, we can apply Theorem~\ref{thm:main1} with $H_0 \to H(s_i)$, $H_1 \to H(s_{i+1})$, $\beta \to (s_{i+1} - s_i)\beta$, and $\gamma \to \gamma_i$ to implement an operator $\widetilde{U}_i$ such that
\begin{equation} \label{U_i} \left\|\ket{\psi_k(s_{i+1})} - \widetilde{U}_i \ket{\psi_k(s_i)}\right\| \leq \frac{\epsilon}{q} \end{equation}
with success probability $1 - \mathcal{O}(\epsilon/q)$ using $\mathcal{O}\left((s_{i+1} - s_i)\frac{\alpha\beta}{\gamma_i^2}\polylog\left[(s_{i+1} - s_i)\frac{q\beta}{\gamma_i \epsilon}\right]\right)$ queries to $O_{H_0}$ and $O_{H_1}$. If $s_{i+1} - s_i = c\gamma_i/\alpha$ for some constant $c$, the query complexity becomes $\mathcal{O}\left(\frac{\beta}{\gamma_i}\polylog\left(\frac{q}{\epsilon}\right)\right)$. Analogous results hold for the number of queries to $O_{H'}$, gates, and ancilla qubits. 

Therefore, we choose the segments such that $s_{i+1} - s_i = c\gamma_i/\alpha$ for each $i \in \{0,\dots, q-1\}$, and use the algorithm of Theorem~\ref{thm:main1} to implement $\widetilde{U} \coloneqq \prod_{i=0}^{q-1} \widetilde{U}_i$ where each $\widetilde{U}_i$ satisfies Eq.~\eqref{U_i}. Then, by the triangle inequality and the fact that $\|\widetilde{U}_i\| \leq 1$,
\[ \left\|\ket{\psi_k(1)} - \widetilde{U}\ket{\psi_k(1)} \right\| \leq \sum_{i = 0}^{q - 1} \left\| \ket{\psi_k(s_{i+1})} - \widetilde{U}_i \ket{\psi_k(s_i)} \right\| \leq \epsilon \]
with probability $1 - \mathcal{O}(\epsilon)$. The total complexity of $\widetilde{U}$ depends on the number of segments $q$ and the sum of $1/\gamma_i$ over all of the segments $i \in \{0,\dots, q-1\}$, which in turn depend on the behaviour of $\gamma(s)$. This is quantified by the following theorem, adapted from Theorems~4 and~5 of~\cite{Jarret2018}. 

\begin{theorem}[Theorems~4 and~5, \cite{Jarret2018}] \label{thm:gap} For self-adjoint operators $H_0$ and $H_1$ on an $N$-dimensional Hilbert space, let $E_0(s) \leq \dots \leq E_{N-1}(s)$ denote the eigenvalues of $(1-s)H_0 + sH_1$ for $s \in [0,1]$. Suppose that $\gamma(s) \coloneqq \min\{E_k(s) - E_{k-1}(s), E_{k+1}(s) - E_k(s)\} > 0$ for all $s \in [0,1]$, and that $\alpha \geq \max\{\|H_0\|, \|H_1\|\}$. Define the sets
\[ J_\ell \coloneqq \left\{s \in [0,1]: \frac{\alpha}{2^{\ell+1}} < \gamma(s) \leq \frac{\alpha}{2^{\ell}}\right\}, \]
and assume that for all $\ell \in \mathbb{Z}$, 
\begin{enumerate}
\item $\mu(J_{\ell}) \leq L/2^\ell$ (where $\mu(J_\ell)$ is the measure of $J_\ell$), and 
\item $J_\ell$ is the union of at most $R$ intervals.
\end{enumerate} 
Then, for any sequence of points $0 = s_0 < \dots < s_q = 1$ such that 
\begin{equation} \label{segment_length} s_{i+1} - s_i = \frac{c}{\alpha} \min_{s \in [s_i, s_{i+1}]}\gamma(s) \end{equation}
for all $i \in \{0, \dots, q-1\}$ and some constant $c \in (0, 1/4)$, 
\begin{equation} \label{gap1}
q \in \mathcal{O}\left((L+R)\log\left(\frac{\alpha}{\gamma}\right)\right) \end{equation} and
\begin{equation} \label{gap2} \sum_{i=0}^{q-1} \left(\min_{s \in [s_i, s_{i+1}]}\gamma(s)\right)^{-1} \in \mathcal{O}\left(\frac{L+R}{\gamma}\right), \end{equation}
where $\gamma \coloneqq \min_{s \in [0,1]}\gamma(s)$. 
\end{theorem}
\begin{proof} \leavevmode
Let $S_\ell \coloneqq \{s_i : s_i \in J_\ell, i \in \{0, \dots, q\}\}$, and write $J_\ell = \bigcup_m I_{\ell,m}$, where the $I_{\ell,m}$ are disjoint intervals. For each $m$, 
\begin{align*}
\mu(I_{\ell,m}) &\geq \sum_{[s_i, s_{i+1}] \subseteq I_{\ell,m}} (s_{i+1} - s_i) \\
&= \sum_{[s_i, s_{i+1}] \subseteq I_{\ell,m}} \frac{c}{\alpha} \min_{s \in [s_i, s_{i+1}]}\gamma(s) \\
&> \frac{c}{2^{\ell+1}} \sum_{[s_i, s_{i+1}] \subseteq I_{\ell,m}}1 \\
&= \frac{c}{2^{\ell+1}}\left(|S_\ell \cap I_{\ell,m}| - 1\right).
\end{align*}
The second line follows from Eq.~\eqref{segment_length}, and the third line from the fact that $\gamma(s) > \alpha/2^{\ell+1}$ for all $s \in I_{\ell,m}$, by definition. To obtain the fourth line, we observe that the number of segments $[s_i, s_{i+1}]$ that are contained in an interval is equal to the number of points $s_i$ in the interval minus $1$. Hence, 
\begin{align} \label{S_ell}
|S_\ell| &= \sum_{m}|S_\ell \cap I_{\ell,m}| \leq \sum_m\left(\frac{2^{\ell+1}}{c} \mu(I_{\ell,m}) + 1\right) = \frac{2^{\ell+1}}{c}\mu(J_\ell) + \sum_m 1 \leq \frac{2L}{c} + R,
\end{align}
where the last inequality uses Assumptions~1 and~2. Since $\gamma \leq \gamma(s) \leq 2\alpha$ for all $s \in [0,1]$,  $S_\ell = \varnothing$ for $\ell > \log(\alpha/\gamma)$ and for $\ell < -1$. Therefore, 
\begin{align*}
q = \sum_{\ell = -1}^{\floor{\log(\alpha/\gamma)}}|S_\ell| \leq \left(\left\lfloor\log\left(\frac{\alpha}{\gamma}\right)\right\rfloor+2\right)\left(\frac{2L}{c}+ R\right),
\end{align*}
which proves Eq.~\eqref{gap1}. 

To prove Eq.~\eqref{gap2}, note that by Weyl's inequality, 
$\gamma(s + \delta s) \geq \gamma(s) - 2\delta s \|H_1 - H_0\| \geq \gamma(s) - 4\alpha\delta s$ for any $\delta s \geq 0$. It follows that for all $i \in \{0, \dots, q-1\}$, 
\begin{equation}  \label{Weyl} \min_{s \in [s_i, s_{i+1}]}\gamma(s) \geq \gamma(s_i) - 4(s_{i+1} - s_i)\alpha \geq (1-4c)\gamma(s_i) \end{equation}
by Eq.~\eqref{segment_length}.
Thus, 
\begin{align*}
\sum_{i=0}^{q-1}\left(\min_{s \in [s_i, s_{i+1}]}\gamma(s)\right)^{-1} &= \sum_{\ell=-1}^{\floor{\log(\alpha/\gamma)}} \sum_{s_i \in S_\ell}\left(\min_{s \in [s_i, s_{i+1}]}\gamma(s)\right)^{-1} \\
&\leq \sum_{\ell=-1}^{\floor{\log(\alpha/\gamma)}} \sum_{s_i \in S_\ell} \frac{1}{(1-4c)\gamma(s_i)} \\
&\leq \frac{1}{1-4c}\sum_{\ell=-1}^{\floor{\log(\alpha/\gamma)}}|S_\ell|\frac{2^{\ell+1}}{\alpha} \\
&\leq \frac{4}{1-4c}\left(\frac{2L}{c} + R\right)\frac{1}{\gamma}, 
\end{align*}
using Eq.~\eqref{S_ell} in the last line.
\end{proof} 

\noindent A few comments about Theorem~\ref{thm:gap} are in order:
\begin{itemize}
\item Even though Theorem~\ref{thm:gap} is stated for the exact gap of $H(s)$, it equivalently applies to any lower bound $\gamma(s)$ on the gap (in which case Eqs.~\eqref{gap1} and~\eqref{gap2} would hold for with $\gamma$ being the minimum of this lower bound), provided that this lower bound satisfies Weyl's inequality [cf.~Eq.~\eqref{Weyl}]. Indeed, given the promise that some $\gamma(s)$ lower-bounds the gap for all $s \in [0,1]$ (along with an upper bound $\alpha$ on the norms of $H_0$ and $H_1$), it is easy to modify $\gamma(s)$ if necessary so that $\gamma(s + \delta s) \geq \gamma(s) - 4\alpha \delta s$. This would tighten the lower bound (and therefore would only serve to improve the efficiency of the state preparation procedure). We henceforth assume without loss of generality that any lower bound $\gamma(s)$ on the gap is consistent with Weyl's inequality.
\item The condition Eq.~\eqref{segment_length} can be relaxed to \begin{equation} \label{segment_length2} c_0\frac{\gamma_i}{\alpha} \leq s_{i + 1} - s_i \leq c_1\frac{\gamma_i}{\alpha} \end{equation} (where $\gamma_i \coloneqq \min_{s \in [s_i, s_{i+1}]}\gamma(s)$) for some constants $c_0, c_1 \in (0,1/4)$.
\item Theorem~\ref{thm:gap} holds for \textit{any} choice of the points $\{s_i\}_i$ that satisfy Eq.~\eqref{segment_length} (or Eq.~\eqref{segment_length2}). This implies that the points can be chosen using any procedure (e.g., a simple greedy algorithm) without affecting the scaling in Eqs.~\eqref{gap1} and~\eqref{gap2}.
\end{itemize}

The quantities $L$ and $R$ in the assumptions of Theorem~\ref{thm:gap} characterise two relevant properties of the gap (or a lower bound on the gap). Intuitively, we would not expect to be able to appreciably improve the scaling in $\gamma \coloneqq \min_{s\in[0,1]}\gamma(s)$ if $\gamma(s)$ were close to $\gamma$ for many values of $s$ (in which case $L$ would be large), or if $\gamma(s)$ oscillated wildly as a function of $s$ (in which case $R$ would be large). In either case, the minimum of $\gamma(s)$ over any large interval would typically be very small, so most segments would be short [cf.~Eq.~\eqref{segment_length}], and both the total number of segments $q$ and the sum of $1/\gamma_i$ over all of the segments would be large.

As a concrete example, we apply Theorem~\ref{thm:gap} to the adiabatic algorithm for Grover search discussed in~\cite{Roland2002}. For a search space of $N$ elements,~\cite{Roland2002} uses the linear interpolation $H(s) = (1-s)H_0 + sH_1$ with $H_0 = I - \ket{\phi}\bra{\phi}$ and $H_0 = I - \ket{m}\bra{m}$, where $\ket{\phi} \coloneqq \sum_{x=1}^N \ket{x}/\sqrt{N}$ and $\ket{m}$ corresponds to the single marked element. Thus, the marked element can be found by preparing the ground state of $H_1$. Since $\|H_0\|= \|H_1\| = 1$, we take $\alpha = 1$. The gap $\gamma(s)$ between the ground state and the first excited state is given exactly by~\cite{Roland2002}
\[ \gamma(s) = \sqrt{1 - 4\frac{N-1}{N}s(1-s)}, \]
so $\gamma = 1/\sqrt{N}$. For any $\gamma \leq y \leq 1$, $\gamma(s) \leq y$ for $|s - 1/2| \leq \sqrt{(y^2N-1)/(N-1)}$, from which it follows that the assumptions of Theorem~\ref{thm:gap} are satisfied with $L = 1$ and $R = 2$. 

The following theorem is an immediate consequence of Theorems~\ref{thm:main1} and~\ref{thm:gap}. Though it is stated for $L, R \in \mathcal{O}(1)$, for simplicity, the dependence on $L$ and $R$ can be added back in if desired by using Theorem~\ref{thm:gap}. 

\begin{theorem} 
\label{thm:improved_gap}
For self-adjoint operators $H_0, H_1 \in \mathbb{C}^{2^{n_s} \times 2^{n_s}}$, let $\ket{\psi_k(0)}$ and $\ket{\psi_k(1)}$ denote the $k$th eigenstates of $H_0$ and $H_1$, respectively. Let it be promised that $\alpha \geq \max\{\|H_0\|,\|H_1\|\}$ and $\beta \geq \|H_1 - H_0\|$, and that for all $s \in [0,1]$, the $k$th eigenstate of $(1-s)H_0 + sH_1$ is non-degenerate and separated from the rest of the spectrum by a gap of at least $\gamma(s) > 0$. Suppose that $\gamma(s)$ satisfies the assumptions of Theorem~\ref{thm:gap} with $L, R \in \mathcal{O}(1)$, and let $\gamma \coloneqq \min_{s\in[0,1]}\gamma(s)$. Then, an operator $\widetilde{U}$ can be implemented such that  
\[ \left\|\ket{\psi_0(1)} - \widetilde{U}\ket{\psi_0(0)}\right\| \leq \epsilon \]
with probability $1-\mathcal{O}(\epsilon)$ using
\begin{itemize}
\item $\mathcal{O}\left(\frac{\beta}{\gamma}\log\left(\frac{\alpha}{\gamma}\right)\polylog\left[\frac{1}{\epsilon}\log\left(\frac{\alpha}{\gamma}\right)\right]\right)$ queries to $O_{H_0}$ and $O_{H_1}$, $\mathcal{O}\left(\log\left(\frac{\alpha}{\gamma}\right)\polylog\left[\frac{1}{\epsilon}\log\left(\frac{\alpha}{\gamma}\right)\right]\right)$ queries to $O_{H'}$,
\item $\mathcal{O}\left(\frac{\beta}{\gamma}\log\left(\frac{\alpha}{\gamma}\right)\left\{n_a + n_b + \polylog\left[\frac{\beta}{\gamma\epsilon}\log\left(\frac{\alpha}{\gamma}\right)\right]\right\}\polylog\left[\frac{1}{\epsilon}\log\left(\frac{\alpha}{\gamma}\right)\right]\right)$ elementary gates, and 
\item $n_s + n_a + \mathcal{O}\left(n_b + \log^2\left[\frac{\beta}{\gamma\epsilon}\log\left(\frac{\alpha}{\gamma}\right)\right]\right)$ qubits,
\end{itemize}
where $O_{H_0}, O_{H_1} \in \mathbb{C}^{2^{n_a + n_s} \times 2^{n_a + n_s}}$ and $O_{H'} \in \mathbb{C}^{2^{n_b + n_s} \times 2^{n_b + n_s}}$ are defined as in subsection~\ref{sec:inputmodel}. 
\end{theorem}

For the example of Grover search, Theorem~\ref{thm:improved_gap} implies that a state $\ket{\widetilde{m}}$ such that $\|\ket{m} - \ket{\widetilde{m}}\| \leq \epsilon$ can be prepared using $\mathcal{O}(\sqrt{N}\log N\polylog(\log N/\epsilon))$ queries to $O_{H_0}$ and $O_{H_1}$ (and a similar number of additional gates). This is an optimal scaling in $N$ up to logarithmic factors, and exponentially improves the scaling in $\epsilon$ compared to a digitised version of Roland and Cerf's algorithm~\cite{Roland2002}. 

Of course, Grover search is a very special case, in which the gap can be solved for exactly. It is typically hard to even estimate the spectra of $H(s) = (1-s)H_0 + H_1$ given arbitrary $H_0$ and $H_1$, which may call into question the validity of the assumption in Theorem~\ref{thm:improved_gap} that a gap lower bound $\gamma(s)$ with $L,R \in \mathcal{O}(1)$ is known. For a broad class of optimisation problems (i.e., problems where $H_1$ is diagonal in the computational basis), this assumption can be justified by the existence of an explicit algorithm that returns tight estimates on the gap between the ground state and the first excited state~\cite{Jarret2018}. Under mild assumptions on $H_1$, these estimates satisfy the conditions of Theorem~\ref{thm:gap} with $L, R \in \mathcal{O}(1)$.\footnote{Technically, the gap estimates returned by the algorithm in~\cite{Jarret2018} are such that $L, R \in \mathcal{O}(1)$ for most of $s \in [0,1]$, but it is shown that the region in which this condition is violated is small enough that the conclusions of Theorem~\ref{thm:gap} [Eqs.~\eqref{gap1} and~\eqref{gap2}] still hold with $L, R \in \mathcal{O}(1)$.} 

\section{Faster preparation of ground states}\label{sec:adiabatic_ground}

In this section, we describe a method for preparing ground states that has slightly better scaling than Algorithm~\ref{alg1} (which applies to arbitrary eigenstates). Algorithm~\ref{alg2} essentially combines the approach of~\cite{Lin2020}, which assumes access to a oracle that prepares an initial state with bounded overlap with the ground state, with an explicit procedure for this initial state preparation. Specifically,
we prepare the initial state by digitally simulating the evolution generated by the time-dependent Hamiltonian $H(t) \coloneqq (1-\frac{t}{T})H_0 + \frac{t}{T} H_1$, choosing $T$ such that the diabatic error is constant. This simulation can be performed by applying the truncated Dyson series algorithm of~\cite{Low2018} to $H(t)$.\footnote{Unlike in Algorithm~\ref{alg1}, where we apply the truncated Dyson series algorithm to the quasi-adiabatic continuation operator $D_{\Delta,T,N}(s)$ to achieve polylogarithmic scaling in the target error, here we directly approximate the unitary generated by $H(t)$ since we are only interested in simulating the adiabatic evolution with constant error.} Then, the eigenstate filtering technique~\cite{Lin2019,Lin2020} can be used to prepare a state close to the ground state of $H_1$ with constant success probability. Eigenstate filtering requires estimating the ground state energy of $H_1$ to a precision proportional to the spectral gap $\gamma_1$ of $H_1$, i.e., $c\gamma_1$ for a constant $c < 1$.~\cite{Lin2020} proposes a hybrid quantum-classical algorithm based on binary search that can estimate the ground state energy to arbitrary precision. However, we cannot use this algorithm directly because it assumes that the initial state preparation procedure is unitary, whereas the operator applied by the truncated Dyson series algorithm is not unitary in general. This issue can be remedied with a minor modification, which we explain below.

Central to this discussion is a unitary that~\cite{Lin2020} refers to as $\text{PROJ}(\mu, \delta, \epsilon)$. We review the relevant facts. $\text{PROJ}(\mu, \delta, \epsilon)$ is constructed from the oracle block-encoding of the Hamiltonian whose ground state we seek to prepare. In our setup, the Hamiltonian of interest is $H_1$, and we use the block-encoding $O_{H_1}$ defined in subsection~\ref{sec:inputmodel}.
Let $|\widetilde{\psi}_0(1)\rangle$ be a state such that $\eta := |\langle \psi_0(1)| \widetilde{\psi}_0(1)\rangle|  \in\Omega(1)$. Then,~\cite{Lin2020} shows that
\begin{equation*}
    \left\|(\langle 0^{n_a+3}|\otimes I) \text{PROJ}(x,h/(2\alpha), \epsilon')(|0^{n_a+3}\rangle |\widetilde{\psi}_0(1)\rangle) \right\|
    \begin{cases}
    \geq \eta - {\epsilon'}/{2} \quad &E_0(1) \leq x- h\\
    \leq {\epsilon'}/{2} \quad &E_0(1) \geq x+h,
    \end{cases}
\end{equation*}
where $E_0(1)$ denotes the ground state energy of $H_1$.
Choosing $\epsilon'={\eta}/{2}$, the difference between the two cases is lower-bounded by ${\eta}/{2} \in \Omega(1)$. 

If we are given a unitary ${u}$ such that ${u}|0\rangle = |\widetilde{\psi}_0(1)\rangle$, by applying binary amplitude amplification~\cite{Lin2020}, we can correctly distinguish between the two cases with probability $1-\delta$ using $\mathcal{O}({\log (1/\delta)}/{\eta})=\mathcal{O}(\log (1/\delta))$ applications of $\text{PROJ}(x,h/2\alpha, \eta/2)$ and ${u}$. More generally, suppose that we only have access to a unitary $\widetilde{u}$ such that 
\begin{equation} \label{utilde} \widetilde{u}\ket{0}_f\ket{0} = \sqrt{\kappa}|0\rangle_f|\widetilde{\psi}_0(1)\rangle  + \sqrt{1-\kappa} |\varphi^{\perp}\rangle
\end{equation} for some $\kappa \in (0,1)$, where $(\langle 0|_f \otimes I)|\varphi^{\perp}\rangle = 0$. It is easy to see that
\begin{equation*}
    \left\|(\langle 0^{n_a+3}|\langle 0|_f\otimes I) \text{PROJ}(x,h/(2\alpha), \eta/2)(\ket{0^{n_a+3}}\otimes \widetilde{u}\ket{0}_f\ket{0}) \right\|
    \begin{cases}
        \geq 3\eta\sqrt{\kappa}/4 \quad &E_0(1) \leq x- h\\
        \leq \eta\sqrt{\kappa}/4 \quad &E_0(1) \geq x+h
    \end{cases}
\end{equation*}
Therefore, for $\kappa \in \Omega(1)$, the difference between the two cases is again $\Omega(1)$, and they can be distinguished via binary amplitude amplification with probability $1 - \delta$ using  $\mathcal{O}(\log(1/\delta))$ applications of $\text{PROJ}(x,h/2\alpha, \eta/2)$ and $\widetilde{u}$. Then, using binary search, the value of $E_0(1)$ can be located within an interval of length $h$ in at most $\mathcal{O}(\log(\alpha/h))$ steps. To ensure that the overall procedure succeeds with probability $1-\epsilon$, it suffices to choose $\delta \in \mathcal{O}({\epsilon}/{\log (\alpha/ h)})$. 

Now, we tally the cost of estimating the ground state energy to precision $c\gamma_1$ with success probability $1 - \epsilon$. With $\delta \in \mathcal{O}({\epsilon}/{\log (\alpha/\gamma_1)})$, the total number of applications of $\text{PROJ}(x,c\gamma_1/(2\alpha), \eta/2)$ and $\widetilde{u}$ is  $\mathcal{O}(\log(\alpha/\gamma_1) \log ({\log(\alpha/\gamma_1)}/{\epsilon}))$. Each $\text{PROJ}(x,\gamma_1/(2\alpha), \eta/2)$ uses one query to $O_{H_1}$~\cite{Lin2020}, so the total number of queries to $O_{H_1}$ is $\mathcal{O}(\log(\alpha/\gamma_1) \log ({\log(\alpha/\gamma_1)}/{\epsilon}))$. $\widetilde{u}$ can be further decomposed into a single query to $G_0$, which prepares the ground state of $H_0$ (from $\ket{0}$) [cf.~subsection~\ref{sec:inputmodel}], and the simulation of the interpolating Hamiltonian $H(t)$. The complexity of implementing $\widetilde{u}$ in this way is given by the following lemma.

\begin{lemma}
\label{lemma:adiabatic_constant_error}
For self-adjoint operators $H_0, H_1 \in \mathbb{C}^{2^{n_s} \times 2^{n_s}}$, let it be promised that $\alpha \geq \max \{\|H_0\|, $ $\|H_1\| \}$ and $\beta \geq \|H_1-H_0\|$, and that for all $s\in[0, 1]$, the ground state of $(1-s)H_0+sH_1$ is non-degenerate and separated from the rest of the spectrum by a gap of at least $\gamma>0$. Let $\ket{\psi_0(1)}$ denote the ground state of $H_1$. Then, a unitary $\widetilde{u}$ satisfying Eq.~\eqref{utilde} with $\kappa \in \Omega(1)$ and
\begin{equation*} 
\left|\langle {\psi}_0(1) |\widetilde{\psi}_0(1)\rangle\right| \in \Omega(1)
\end{equation*}
using
\begin{itemize}
    \item  $\mathcal{O}\left(\frac{\alpha\beta}{\gamma^2}\frac{\log({\alpha\beta}/{\gamma^2})}{\log\log({\alpha\beta}/{\gamma^2})}\right)$ queries to $O_{H_0}$ and $O_{H_1}$, one query to $G_0$,
    \item $\mathcal{O}\left(\frac{\alpha\beta}{\gamma^2}\left[n_a + \log\left(\frac{\alpha\beta}{\gamma^2}\right)\mathcal{M}\left(\log\left(\frac{\alpha\beta}{\gamma^2}\right)\right) \right]\frac{\log({\alpha\beta}/{\gamma^2})}{\log\log({\alpha\beta}/{\gamma^2})}\right)$ elementary gates, and
    \item $n_s + \mathcal{O}\left(n_a + \log^2\left(\frac{\alpha\beta}{\gamma^2}\right)\right)$ qubits,
\end{itemize}
where $O_{H_0}, O_{H_1} \in \mathbb{C}^{2^{n_a + n_s} \times 2^{n_a + n_s}}$ and $G_0 \in \mathbb{C}^{2^{n_s} \times 2^{n_s}}$ are defined as in subsection~\ref{sec:inputmodel}.
\end{lemma} 
\begin{proof}
Let $U(T) \coloneqq \mathcal{T}[e^{-i\int_0^T dt\,H(t)}]$ be the time evolution generated by $H(t) \coloneqq (1-\frac{t}{T})H_0 + \frac{t}{T}H_1$. We can digitially simulate $U(T)$ by applying the truncated Dyson series algorithm of~\cite{Low2018} to the unitaries $\sel U_{\widetilde{H}}$ (constructed in Lemma~\ref{lemma:sel_UtildeH}), which block-encode $\sum_{m=0}^{M-1}\ket{m}\bra{m}\otimes H[m]$ [cf.~Eq.~\eqref{H[m]}].
Suppose that $\sel[]U_{\widetilde{H}}$ is constructed such that $\|H[m] - \widetilde{H}[m]\|\leq \epsilon_0$ for all $m \in \{0,\dots, M-1\}$. By  Corollary~4 of~\cite{Low2018} (with the same modification as in Theorem~\ref{thm:TDS}) and Proposition~\ref{prop:Dyson}, a unitary $V$ such that \[ \left\|{U}(T) - (\bra{0} \otimes I)V(\ket{0}\otimes I) \right\| \in \mathcal{O}(1) \]
can be implemented with success probability $\Omega(1)$
using $\mathcal{O}\left(\alpha T \frac{\log(\alpha T)}{\log\log(\alpha T)}\right)$ calls to $\sel U_{\widetilde{H}}$, along with $\mathcal{O}\left( n_a\alpha T \frac{\log(\alpha T)}{\log\log(\alpha T)}\right)$ elementary gates and $\mathcal{O}(n_a)$ ancilla qubits, and by choosing $\epsilon_0 \in \Theta(1/T)$. With this choice of $\epsilon_0$, it follows from Lemma~\ref{lemma:sel_UtildeH} that each application of $\sel U_{\widetilde{H}}$ uses one query to each of $O_{H_0}$ and $O_{H_1}$, $\mathcal{O}(\log(\alpha T)\mathcal{M}(\log(\alpha T)))$ elementary gates, and $\mathcal{O}(\log^2(\alpha T))$ ancilla qubits (all of which are returned to their initial states and can be reused). Therefore, $V$ uses a total of $\mathcal{O}\left(\alpha T \frac{\log(\alpha T)}{\log\log(\alpha T)}\right)$ queries to $O_{H_0}$ and $O_{H_1}$, $\mathcal{O}\left(\alpha T[n_a + \log(\alpha T)\mathcal{M}(\log(\alpha T))]\frac{\log(\alpha T)}{\log\log(\alpha T)}\right)$ elementary gates, and $n_s + \mathcal{O}(n_a + \log^2(\alpha T))$ qubits.

According to Theorem 3 of~\cite{Jansen2007}, 
\begin{equation*}
    \left\| U(T)|\psi_0(0)\rangle - |\psi_0(1)\rangle \right\| \leq  \frac{\beta}{T\gamma^2},
\end{equation*}
where $\ket{\psi_0(0)}$ is the ground state of $H_0$.
To ensure that this error is $\mathcal{O}(1)$, we take $T\in\Theta(\beta/\gamma^2)$. Then, since $G_0\ket{0} = \ket{\psi_0(0)}$ by definition, $\widetilde{u} \coloneqq V(I \otimes G_0)$ satisfies Eq.~\eqref{utilde} with $\kappa \in \Omega(1)$ and $\ket{\widetilde{\psi}_0(1)}$ having constant overlap with $\ket{\psi_0(1)}$. Thus, $\widetilde{u}$ uses a single query to $G_0$. Substituting $T \in \Theta(\beta/\gamma^2)$ into the complexity of implementing $V$ gives the total number of queries to $O_{H_0}$ and $O_{H_1}$, gates, and qubits.
\end{proof}

The ground state preparation procedure comprises three steps [cf.~Algorithm~\ref{alg2}]. First, using the binary search algorithm of~\cite{Lin2020}, the ground state energy is estimated to precision $c\gamma_1$ with probability $1-\epsilon$. As shown above, this uses the unitary $\widetilde{u}$ constructed in Lemma~\ref{lemma:adiabatic_constant_error} $\mathcal{O}(\log (\alpha/\gamma_1) \log (\log(\alpha/\gamma_1)/\epsilon))$ times. Next, $\widetilde{u}$ is applied to $\ket{0}_f\ket{0}$ to prepare a state $\ket{\widetilde{\psi}_0}$ with constant overlap with $\ket{0}_f\ket{\psi_0(1)}$ [cf.~Eq.~\eqref{utilde}]. Lastly, we use eigenstate filtering~\cite{Lin2019,Lin2020} to $\ket{\widetilde{\psi}_0}$ to filter out $|\psi_0(1)\rangle$. The third step succeeds with constant probability; the success probability can be boosted to $1- \mathcal{O}(\epsilon)$ by repeating the second and third steps, or through amplitude amplification. In either case, the overall complexity is dominated by the first step. 
\adiabaticg*
If the assumptions in Theorem~\ref{thm:improved_gap} hold, the dependence of the query and gate complexities on $\gamma$ in Theorem~\ref{thm:adiabatic_ground} can be improved to $\widetilde{\mathcal{O}}(\beta/\gamma)$ by performing the adiabatic simulation in Lemma~\ref{lemma:adiabatic_constant_error} in segments of varying lengths, in a procedure analogous to that described in Section~\ref{sec:improved_gap_dependence}. 

\section{Discussion}
\label{sec:conclusion}
We proposed two algorithms for state preparation. Our main algorithm (Algorithm~\ref{alg1}) is based on digitally simulating quasi-adiabatic continuation, and can be applied to any eigenstate that is separated from the rest of the spectrum along the adiabatic path. This algorithm can be easily extended to degenerate eigenspaces via a straightforward generalisation of the results in Section~\ref{sec:qac}. Our second algorithm (Algorithm~\ref{alg2}) involves applying eigenstate filtering to a state prepared by simulating an adiabatic evolution with constant error. Both approaches have query complexity $\widetilde{\mathcal{O}}({\alpha \beta}/{\gamma^2})$, where $\alpha$ is an upper bound on the norms of the initial and final Hamiltonians $H_0$ and $H_1$, $\beta$ is an upper bound on $\|H_1 - H_0\|$, and $\gamma$ is a lower bound on the gap of $H(s) \coloneqq (1-s)H_0 + H_1$. This scaling is essentially optimal in $\alpha$, $\beta$, and $\gamma$ (in the setting where only a uniform lower bound $\gamma$ on the gap of $H(s)$ is known). 

Moreover, we saw that the costs of our algorithms can be considerably reduced given extra information about the gap of $H(s)$. In particular, if the assumptions of Theorem~\ref{thm:gap} hold, a factor of $\alpha/\gamma$ can be shaved off from both the query and gate complexity (see Theorem~\ref{thm:improved_gap}), 
which would be a significant reduction when $\alpha$ is large and $\gamma$ is small. Such a scenario may arise in simulations of many-body quantum systems near the quantum critical point~\cite{Sachdev2011}. In a $d$-dimensional system of size $L^d$, the gap generically reaches its minimum near the phase transition point, where it becomes $\Theta(1/L)$. The scaling in $L$ would therefore be $\widetilde{\mathcal{O}}(L^{2d+2})$ according to Theorems~\ref{thm:main1} and~\ref{thm:adiabatic_ground}. However, if the gap profile satisfies the assumptions of Theorem~\ref{thm:gap}, this could be reduced to $\widetilde{\mathcal{O}}(L^{d+1})$ (by using the procedure described in Section~\ref{sec:improved_gap_dependence}), a potential quadratic improvement.

There are various adjustments that can be made to further improve our algorithms. For one, the gate and space complexities stated in our theorems are specific to the na\"ive methods for approximating special functions that we considered in Appendix~\ref{specialfunctions}, and can be straightforwardly improved by using more sophisticated approaches. For instance,
the $\mathcal{O}(\log^2(\alpha\beta/\gamma^2))$ term in the ancilla cost can be reduced by implementing Gidney's space-efficient algorithm for performing multiplication on a quantum computer~\cite{Gidney2019}. Similarly, more optimised quantum circuits for division~\cite{Thapliyal2018} and for computing elementary functions~\cite{Haner2018} may be helpful. Alternatively, it may make sense in practice to use approximation schemes for special functions that work well empirically; see~\cite{Abramowitz1983} for relevant examples. 

It should also be noted that although our choice of the function $W_\Delta$ in Section~\ref{sec:qac} suffices for our purpose of obtaining polylogarithmic error dependence, it may not be optimal. There may well exist other efficiently integrable functions that yield parametrically tighter bounds on the diabatic error, and using such a function to construct the discretised quasi-adiabatic continuation operator (in subsection~\ref{sec:3.2}) would improve the error scaling of Algorithm~\ref{alg1}. We leave this as an open problem.  

Both algorithms involve simulating the dynamics generated by a time-dependent self-adjoint operator. To this end, we used the truncated Dyson series algorithm of~\cite{Low2018}, which, to the best of our knowledge, achieves the best scaling for the query model that we consider. Since the costs of our algorithms are determined in large part by the complexity of the truncated Dyson series algorithm, any future advances in time-dependent Hamiltonian simulation (using the same query model) would automatically imply that these costs could be reduced. 

\section*{Acknowledgments}
The authors acknowledge PsiQuantum's support for this project.

\appendix

\section{Oracles for linear combinations of $H_0$ and $H_1$} \label{appendixA} 
Given access to the the oracles $O_{H_0}$ and $O_{H_1}$ that block-encode $H_0$ and $H_1$, it is straightforward to construct block-encodings of linear combinations of $H_0$ and $H_1$. Of particular relevance to our algorithms are block-encodings of $H' \coloneqq H_1 - H_0$ and of $H(s) \coloneqq (1-s)H_0 + sH_1$. 

\subsection{Constructing $O_{H'}$ using $O_{H_0}$ and $O_{H_1}$} \label{sec:oracle_relation}

The block-encoding $O_{H'}$ of $H' \coloneqq H_1 - H_0$ can be constructed using controlled versions of $O_{H_0}$ and $O_{H_1}$. We define the operation 
\begin{equation} \label{selOH}
\sel O_{H} \coloneqq \ket{0}\bra{0} \otimes O_{H_0} + \ket{1}\bra{1}\otimes O_{H_1}, 
\end{equation}
which applies $O_{H_i}$ to the target state if the control qubit is in the state $\ket{i}$, $i \in \{0,1\}$. This can be implemented using a single query to each of controlled-$O_{H_0}$ and controlled-$O_{H_1}$ (along with a constant number of elementary gates). Since
\[ (\bra{+}_{b'}\bra{0}_a\otimes I_s)(\sel[b'] O_{H})(\ket{-}_{b'}\ket{0}_a\otimes I_s) = \frac{1}{2}\left(\frac{H_1}{\alpha} - \frac{H_0}{\alpha}\right) = \frac{H'}{2\alpha}, \]
with $\ket{\pm} \coloneqq (\ket{0}\pm \ket{1})/\sqrt{2}$, 
$O_{H'}$ can be constructed as
\begin{equation} \label{OH'explicit} O_{H'} = \left(\textsc{Had}_{b'}\otimes I_{as}\right)(\sel[b'] O_{H})\left((X\cdot Z\cdot \textsc{Had})_{b'}\otimes I_{as}\right), \end{equation}
where $\textsc{Had}$ denotes the Hadamard gate. This satisfies Eq.~\eqref{O_H'} with $\ket{0}_b = \ket{0}_{b'}\ket{0}_a$ and $\beta = 2\alpha \geq \|H'\|$. 

\subsection{Block-encodings of $H(s)$}
\label{sec:oracle_linear_combination}

A block-encoding $O_{H(s)}$ of $H(s)\coloneqq (1-s)H_0 + sH_1$ for any $s \in [0,1]$ can similarly be implemented using one query to each of $O_{H_0}$ and $O_{H_1}$. For any $\theta \in \mathbb{R}$, let $R(\theta) \coloneqq e^{-i2\pi\theta Y}$ so that $R(\theta)\ket{0} = \cos(2\pi\theta)\ket{0} + \sin(2\pi\theta)\ket{1}$. $O_{H(s)}$ can be constructed as
\[ O_{H(s)} \coloneqq (R(\theta(s))_{b'}^\dagger \otimes I_{as})(\sel[b'] O_H)(R(\theta(s))_{b'} \otimes I_{as}), \]
where $\theta(s) \coloneqq \arcsin(\sqrt{s})/(2\pi)$ and $\sel O_H$ is defined by Eq.~\eqref{selOH}. Then,
\[ \left(\bra{0}_{b'}\bra{0}_a \otimes I_s\right) O_{H(s)}\left(\ket{0}_{b'}\ket{0}_a \otimes I_s\right) = \frac{H(s)}{\alpha}, \]
so the size of the ancillary register of this block-encoding is $n_a + 1$. 

\section{Special functions} \label{specialfunctions}
We bound the complexity of approximating certain special functions to arbitrary absolute error on a quantum computer. As in the main text, we use $\mathcal{M}(b)$ to denote the complexity of $b$-bit multiplication in terms of the number of elementary logic gates. While $\mathcal{M}(b)$ depends on the specific multiplication algorithm that is implemented, some of the proofs use the fact that $\mathcal{M}(b)$ scales no worse than as $b^2$ (corresponding to the standard long multiplication algorithm). 

The following results are based on classical circuits for approximating the functions of interest. The general strategy is to upper-bound the number of arithmetic operations required to achieve a desired precision in the approximation. Each arithmetic operation can be implemented using at most $\mathcal{O}(\mathcal{M}(b))$ elementary logic gates, where $b$ denotes the working precision. These (possibly irreversible) classical gates can be replaced with $\mathcal{O}(\mathcal{M}(b))$ quantum gates at the cost of $\mathcal{O}(\mathcal{M}(b))$ ancilla qubits. By applying $b$ \textsc{cnot} gates to copy the result of the arithmetic operation to a $b$-qubit register initialised to $\ket{0}$, then applying the circuit in reverse, the state of the ancillae (as well as any bits that are unnecessary for the subsequent computation, due to rounding) can be uncomputed. Hence, the same $\mathcal{O}(\mathcal{M}(b))$ ancillae can be used for each arithmetic operation if the operations are performed sequentially. The $\textsc{cnot}$ gates and uncomputation incur no more than a constant multiplicative factor in the gate complexity. Thus, the overall gate complexity is $\mathcal{O}(\ell\mathcal{M}(b))$ and the space overhead is $\mathcal{O}(\ell b + \mathcal{M}(b))$ if at most $\ell$ arithmetic operations are performed in total. 

We note that the bounds provided below are by no means tight, but suffice for our purposes. For practical implementations, more sophisticated algorithms may be used.

\begin{proposition} \label{prop:arcsin}
For any $\delta > 0$, there exists a quantum circuit that maps
\[ \ket{y}\ket{0} \mapsto \ket{y}\ket{a_y} \]
for all fixed-point numbers $y \in [0,1]$, 
where for each $y$, $a_y$ is a $\mathcal{O}(\log(1/\delta))$-bit number such that
\[ \left|\arcsin(\sqrt{y}) - a_y\right| \leq \delta. \]
Such a circuit can be constructed using at most $\mathcal{O}(\log(1/\delta)\mathcal{M}(\log(1/\delta)))$ elementary gates and $\mathcal{O}(\log^2(1/\delta))$ ancilla qubits (initialisd in and reset to $\ket{0}$). 

\begin{proof}
Let $z \coloneqq \sqrt{y}$, and let $z_f \leq 1$ be a fixed-point approximation of $\sqrt{y}$ with $f$ fractional bits, which can be computed to absolute precision $2^{-f}$ from the first $\mathcal{O}(f)$ bits of $y$ using $\mathcal{O}(\mathcal{M}(f))$ elementary operations~\cite{Alt1979,Muller2005}. 

For $|z| \leq 1$,
\[ \arcsin(z) = \sum_{k=0}^{\infty}\frac{(2k-1)!!}{(2k)!!}\frac{z^{2k+1}}{2k+1}. \]
Define the truncated series
\[ \arcsin_\ell(z) \coloneqq \sum_{k=0}^\ell \frac{(2k-1)!!}{(2k)!!}\frac{z^{2k+1}}{2k+1}. \]
Then, for $z \in [0,1/\sqrt{2}]$,
\begin{align*}
\left|\arcsin({z}) - \arcsin_\ell({z})\right| &= \sum_{k=\ell+1}^{\infty}\frac{(2k-1)!!}{(2k)!!}\frac{z^{2k+1}}{2k+1} \\
&\leq \sum_{k=\ell +1}^\infty \frac{1}{2}\left(\frac{1}{\sqrt{2}}\right)^{2k+1} \\
&\leq \frac{1}{2\sqrt{2}}2^{-\ell}.
\end{align*}
Since $|z - z_f| \leq 2^{-f}$,
\begin{align*}
\left|\arcsin_\ell(z) - \arcsin_\ell(z_f)\right| &= \left|\sum_{k=0}^{\ell}\frac{(2k-1)!!}{(2k)!!}\frac{z^{2k+1}}{2k+1} - \sum_{k=0}^{\ell}\frac{(2k-1)!!}{(2k)!!}\frac{{z_f}^{2k+1}}{2k+1}\right| \\
&\leq \sum_{k=0}^{\ell}\frac{|z^{2k+1} - {z_f}^{2k+1}|}{2k+1} \\
&\leq \sum_{k=0}^\ell \frac{(2k+1)|z-z_f|}{2k+1} \\
&\leq \ell 2^{-f},
\end{align*}
where the second inequality follows from the fact that $z, z_f \leq 1$.
Thus, $|\mathrm{arcsin}(z) - \arcsin_\ell(z_f)| \leq 2^{-\ell}/(2\sqrt{2}) + \ell 2^{-f}$. In addition to the error in approximating the square root of $y$ and in truncating the Taylor series expansion of $\mathrm{arcsine}$, there is round-off error in computing the truncated series $\arcsin_\ell(z_f)$. For each $k \in \{1,\dots, \ell\}$, the $k$th term can be computed by multiplying the $(k-1)$th term by $(2k-1)^2{z_f}^2$ then dividing by $2k(2k+1)$. Since $z_f \leq 1$ and $(2k-1)^2 < 2k(2k+1)$, this introduces an additive error of at most $c2^{-f}$ for some constant $c$ if we round to $f$ fractional bits. 
Hence, the value computed for the $k$th term differs from $(2k-1)!! {z_f}^{2k+1}/[(2k)!!(2k+1)]$ by at most $ck 2^{-f}$. Adding all $\ell$ terms together to obtain $a_y$ (for $y \leq 1/2$), the total round-off error is smaller than $c\ell^2 2^{-f}$. The total error is therefore
\begin{align*} 
\left|\arcsin(\sqrt{y}) - a_y\right| \leq \frac{1}{2\sqrt{2}}2^{-\ell} + \ell 2^{-f} + c\ell^2 2^{-f},
\end{align*}
which can be upper-bounded by $\delta$ by taking $\ell \in \mathcal{O}(\log(1/\delta))$ and $f \in \mathcal{O}(\log(\ell/\delta)) = \mathcal{O}(\log(1/\delta))$. 

Iteratively calculating the terms of the truncated series $\arcsin(z_f)$ requires $\mathcal{O}(\ell)$ arithmetic operations (additions, multiplications, and divisions).  The largest number involved in this calculation is $2\ell + 1$, so we allocate $\ceil{\log (2\ell+1)}$ magnitude bits. Adding this to the number $f$ of fractional bits, the number of qubits needed to store the result of each intermediate computation, as well as the final approximation $a_y$, is at most $b \propto \ceil{\log(2\ell + 1)} + f \in \mathcal{O}(\log(1/\delta))$. The number of ancillae needed to store all $\mathcal{O}(\ell)$ intermediate results is then $\mathcal{O}(\ell b) = \mathcal{O}(\log^2(1/\delta))$.  Each arithmetic operation has gate complexity at most $\mathcal{O}(\mathcal{M}(b))$, and can be performed reversibly using $\mathcal{O}(\mathcal{M}(b))$ ancillae. The state of these ancillae can be uncomputed after each arithmetic operation (resulting in only a constant multiplicative overhead in the overall gate complexity), so that the total number of ancillae required is $\mathcal{O}(\ell b + \mathcal{M}(b)) = \mathcal{O}(\log^2(1/\delta))$. The complexity of computing $z_f$ reversibly has the same scaling. Thus, the total number of elementary gates used in computing $a_y$ for $y \in [0,1/2]$ (and uncomputing the state of all of the ancillae) is 
\[ \mathcal{O}(\ell \mathcal{M}(b)) = \mathcal{O}(\log(1/\delta)\mathcal{M}(\log(1/\delta))). \]

For $y \in (1/2,1]$, it follows from the identity
\[ \arcsin(\sqrt{y}) = \frac{\pi}{2} - \arcsin(\sqrt{1-y}) \]
that $a_y$ can be computed by approximating $\arcsin(\sqrt{1-y})$ using the above procedure and subtracting the result from (a fixed-point approximation of) $\pi/2$.

\end{proof}
\end{proposition}

\begin{proposition} \label{prop:gaussian}
For any $\delta > 0$, there exists a quantum circuit that maps
\[ \ket{x}\ket{0} \mapsto \ket{x}\ket{g_x} \]
for all fixed-point numbers $x \in [-r,r]$, where for each $x$, $g_x$ is a $\mathcal{O}(\log(1/\delta))$-bit number such that
\[ \left|\exp(-x^2) - g_x\right| \leq \delta. \]
Such a circuit can be constructed using at most $\mathcal{O}(\log(r/\delta)\mathcal{M}(\log[(\log r )/\delta]))$ elementary gates and $\mathcal{O}(\log^2(r/\delta))$ ancilla qubits (initialised in and reset to $\ket{0}$. 
\end{proposition}

\begin{proof}
We can assume without loss of generality that each input $x$ has at most $\mathcal{O}(\log(1/\delta))$ fractional bits. If this were not the case, we can apply a circuit that acts trivially on all of the fractional bits after the first $\ceil{\log(2/\delta)}$. Effectively, this circuit takes as input an approximation $\widetilde{x}$ to $x$ such that $|x - \widetilde{x}| \leq \delta/2$. Since $|\mathrm{exp}(-x^2) - \exp(-\widetilde{x}^2)| \leq |x - \widetilde{x}| \leq \delta/2$, by approximating $\exp(-\widetilde{x}^2)$ to within absolute error $\delta/2$ using the procedures described below, the circuit computes an approximation that differs from $\exp(-x^2)$ by at most $\delta$. Therefore, it suffices to consider inputs $x$ with at most $\mathcal{O}(\log(1/\delta))$ fractional bits.  

For $|x| \leq 1$, we approximate $\exp(-x^2)$ using a truncated Taylor series
\[ \exp_\ell(-x^2)\coloneqq \sum_{k=0}^\ell \frac{(-x^2)^k}{k!}. \] The error in this approximation depends on the truncation order $\ell$ and the precision to which the terms are computed. To determine an appropriate value of $\ell$, note that for $x$ in this range, the remainder is bounded by
\begin{align*} \left|\exp(-x^2) - \exp_\ell(-x^2)\right| &\leq \left|\frac{(-x^2)^{\ell+1}}{(\ell+1)!}\right| \leq \frac{1}{(\ell+1)!} \leq \left(\frac{e}{\ell+1}\right)^{\ell+1}. \end{align*}
The terms of the series can be calculated iteratively: multiplying the $(k-1)$th term by $-x^2$ and dividing by $k$ yields the $k$th term, for each $k \in \{1,\dots, \ell\}$. Since $|x| \leq 1$ and $k \geq 1$, each multiplication and division incurs a round-off error of at most $c 2^{-f}$ for some constant $c$, where $f$ is the number of fractional bits that are kept. Hence, the round-off error in the approximation to the $k$th term is at most $ck 2^{-f}$. Adding all $\ell$ terms together to obtain $g_x$, the total round-off error is bounded by $c\ell^2 2^{-f}$. The total error is therefore
\[ \left|\exp(-x^2) - g_x\right| \leq \left(\frac{e}{\ell + 1}\right)^{\ell + 1} + c \ell^2 2^{-f}. \] 
The first term is at most $\delta/2$ if we choose
\begin{equation*} \label{ellg}  \ell = \left\lceil \frac{\ln(2/\delta)}{\mathcal{W}(\ln(2/\delta)/e)}\right\rceil - 1 \in \Theta\left(\frac{\log(1/\delta)}{\log\log(1/\delta)} \right). \end{equation*} 
Then, the second can be upper-bounded by $\delta/2$ by taking $f \in \mathcal{O}(\log(\ell/\delta)) = \mathcal{O}(\log(1/\delta))$. 

By iteratively computing the terms of the truncated series $\exp_\ell(-x^2)$ iteratively as discussed above, the total number of arithmetic operations (additions, multiplications, and divisions) performed is $\mathcal{O}(\ell)$. The largest number involved in this calculation is $\ell$, so we use $\ceil{\log \ell}$ magnitude bits. Adding this to the number $f$ of fractional bits, the number of bits needed to store the result of each intermediate computation, as well as the final approximation $g_x$, is at most $b \propto \ceil{\log \ell} + f \in \mathcal{O}(\log(1/\delta))$. Thus, the number of ancillae needed to store all $\mathcal{O}(\ell)$ intermediate results is $\mathcal{O}(\ell b) = \mathcal{O}(\log^2(1/\delta))$. 
The number of gates can be estimated as follows. Each arithmetic operation has gate complexity at most $\mathcal{O}(\mathcal{M}(b))$, and can be performed reversibly using $\mathcal{O}(\mathcal{M}(b))$ ancillae. By uncomputing the state of these ancillae after each arithmetic operation (which results in only a constant multiplicative overhead in the gate complexity), the total number of ancillae required is $\mathcal{O}(\ell b + \mathcal{M}(b)) = \mathcal{O}(\log^2(1/\delta))$. The total number of elementary gates used in computing $g_x$ for $|x| \leq 1$ (and uncomputing the state of all of the ancillae) is clearly
\[ \mathcal{O}(\ell\mathcal{M}(b)) = \mathcal{O}\left(\frac{\log(1/\delta)}{\log\log(1/\delta)}\mathcal{M}(\log(1/\delta))\right). \]

For the $|x| > 1$ case, we apply argument reduction. We can use the above procedure for $|x| \leq 1$ to compute a $\mathcal{O}(\log(1/\delta_0))$-bit approximation $\widetilde{g}_{x,0} \leq 1$ to $\exp[-(x^2-\floor{x^2})]$ such that 
\[ \left|\exp[-(x^2-\floor{x^2})] - \widetilde{g}_{x,0} \right| \leq \delta_0 \]
with $\mathcal{O}(\log(1/\delta_0)\mathcal{M}(\log(1/\delta_0))/\log\log(1/\delta_0))$ gates and $\mathcal{O}(\log^2(1/\delta_0))$ ancillae. Using the same procedure, a $\mathcal{O}(\log(1/\delta_0))$-bit approximation $\widetilde{e} \leq 1/2$ to $\exp(-1)$ with $|\mathrm{exp}(-1) - \widetilde{e}| \leq \delta_0$ can also be computed, with the same gate and space complexity. Then, $|\mathrm{exp}(-\floor{x^2}) - \widetilde{e}^{\floor{x^2}}| \leq |\mathrm{exp}(-1) - \widetilde{e}| \leq \delta_0$ since $\exp(-1), \widetilde{e} \leq 1/2$. By repeatedly squaring to obtain $\widetilde{e}^{2^k}$ for $k \in \{1,\dots, \ceil{\log\floor{x^2}}\}$, $\widetilde{e}^{\floor{x^2}}$ can be computed using at most $2\ceil{\log \floor{x^2}}$ multiplications. Since $\widetilde{e} < 1$, each of these multiplications introduces a round-off error of at most $c 2^{-f}$ for some constant $c$, if we round each product to $f$ (fractional) bits. Thus, we obtain an approximation $\widetilde{g}_{x,1}$ to $\exp(-\floor{x^2})$ such that 
\[ \left|\exp(-\floor{x^2}) - \widetilde{g}_{x,1}\right| \leq \delta_0 + \ceil{\log\floor{x^2}} c 2^{1-f}. \]
Multiplying $\widetilde{g}_{x,1}$ and $\widetilde{g}_{x,1}$ yields a $\mathcal{O}(\log(1/\delta_0) + f)$-bit approximation $\widetilde{g}_x$ to $\exp(-x^2)$ that satisfies 
\begin{align*}
\left|\exp(-x^2) - \widetilde{g}_x \right| &= \left|\exp[-(x^2-\floor{x^2})]\exp(-\floor{x^2}) - \widetilde{g}_{x,0}\widetilde{g}_{x,1}\right| \\
&\leq \left|\exp[-(x^2 - \floor{x^2})] - \widetilde{g}_{x,0} \right| + \left|\exp(-\floor{x^2}) - \widetilde{g}_{x,1}\right| \\
&\leq \delta_0 + (\delta_0 + \ceil{\log \floor{x^2}} c 2^{1-f}),
\end{align*}
where the first equality uses the fact that $\exp(-\floor{x^2}), \widetilde{g}_{x,0} \leq 1$. Hence, the total error can be upper-bounded by $\delta/2$ by choosing $\delta_0 \in \Theta(\delta)$ and $f \in \mathcal{O}(\log[(\log x^2)/\delta]) = \mathcal{O}(\log[(\log r)/\delta])$, since $|x| \leq r$. 

It follows that computing $\widetilde{g}_{x,0}$ and $\widetilde{e}$ uses $\mathcal{O}(\log(1/\delta)\mathcal{M}(\log(1/\delta))/\log\log(1/\delta))$ elementary gates and $\mathcal{O}(\log^2(1/\delta))$ ancillae. Computing $\widetilde{g}_{x,1}$ from $\widetilde{e}$ as described above (and subsequently multiplying $\widetilde{g}_{x,1}$ by $\widetilde{g}_{x,0}$ to obtain $\widetilde{g}_x$) requires $\mathcal{O}(\log r)$ multiplications of $\mathcal{O}(\log[(\log r)/\delta])$-bit numbers. These multiplications can be performed reversibly using $\mathcal{O}((\log r)\mathcal{M}(\log[(\log r)/\delta]))$ elementary gates and $\mathcal{O}(\mathcal{M}(\log[(\log r)/\delta]))$ ancillae in total, by uncomputing the state of the ancillae after each multiplication. In addition, $\mathcal{O}((\log r)\log[(\log r)/\delta])$ ancillae are needed for storing the $\mathcal{O}(\log r)$ intermediate products, so a total of $\mathcal{O}(\log^2(r/\delta))$ ancillae suffice for the computation of $\widetilde{g}_{x,1}$. Combining the gate and space requirements for computing $\widetilde{g}_{x,0}$ and $\widetilde{g}_{x,1}$, we see that  at most $\mathcal{O}(\log(r/\delta)\mathcal{M}(\log[(\log r)/\delta))$ elementary gates and $\mathcal{O}(\log^2(r/\delta))$ ancillae are required to ensure that $|\mathrm{exp}(-x^2) - \widetilde{g}_x| \leq \delta$ for any $x \in [-r,r]$. Then, the most significant $\ceil{\log(2/\delta)} = \mathcal{O}(\log(1/\delta))$ (fractional) bits of $\widetilde{g}_x$ represent a $\delta$-approximation $g_x$ to $\exp(-x^2)$. The remaining fractional bits, as well as all of the ancillae, can be uncomputed with constant multiplicative overhead.
\end{proof}

\begin{proposition}[\cite{Chevillard2012}] \label{prop:erfc}
For any $\delta > 0$, there exists a quantum circuit that maps 
\[ \ket{x}\ket{0} \mapsto \ket{x}\ket{c_x} \]
for all fixed-point numbers $x \in [-r,r]$, where for each $x$, $c_x$ is a $\mathcal{O}(\log(1/\delta))$-bit number such that 
\[ \left|\erfc(x) - c_x\right| \leq \delta. \]
Such a circuit can be constructed using at most $\mathcal{O}(\log(r/\delta)\mathcal{M}(\log[(\log r )/\delta]))$ elementary gates and $\mathcal{O}(\log^2(r/\delta))$ ancilla qubits (initialised in and reset to $\ket{0}$). 
\end{proposition}

\begin{proof}
We consider Algorithm 4 of~\cite{Chevillard2012}, which, given $\delta > 0$ and a floating-point number $x$, computes a $\mathcal{O}(\log(1/\delta))$-bit approximation to $\erf(x)$ with relative error $\delta$. Since $|\mathrm{erf}(x)| \leq 1$ for all $x$, such an approximation is trivially within absolute error $\delta$, and subtracting it from $1$ yields a $\delta$-approximation of $\erfc(x)$. 

The total cost of the algorithm is dominated by that of $\ell \in \mathcal{O}(\log(|x|/\delta))$ additions and multiplications involving numbers with at most $b \in \mathcal{O}(\log[(\log |x|)/\delta])$ bits. Each of these arithmetic operations can be implemented reversibly using $\mathcal{O}(\mathcal{M}(b)) = \mathcal{O}(\mathcal{M}(\log[(\log |x|)/\delta]))$ elementary gates and ancillae. The ancillae can be uncomputed and reused in subsequent operations. Thus, the space overhead of making the operations reversible is independent of $\ell$. The results of all $\mathcal{O}(\ell)$ operations can be stored in $\mathcal{O}(\ell b) = \mathcal{O}(\log(|x|/\delta)\log[(\log |x|)/\delta])$ additional ancillae. Hence, the total number of ancillae required is $\mathcal{O}(\log^2(|x|/\delta))$. The total gate cost for the arithmetic operations (and uncomputation) is $\mathcal{O}(\ell \mathcal{M}(b)) = \mathcal{O}(\log(|x|/\delta) \mathcal{M}(\log[(\log |x|)/\delta]))$. This dominates the complexity of converting between floating-point and fixed-point representations. The claim follows from the assumption that $|x| \leq r$.
\end{proof}

\section{Approximate preparation of $\ket{{W}_{\Delta,T,N}}$} \label{sec:prepareW}
We describe an efficient implementation of a unitary that approximately prepares the $\ceil{\log N}$-qubit state $\ket{W_{\Delta,T,N}}$, which encodes the coefficients of the quasi-adiabatic continuation operator defined in Eq.~\eqref{D_DTN}.  We restate the relevant definitions: 
\begin{equation} \label{eq:Wstate} \ket{W_{\Delta,T,N}} = \sum_{n=1}^N \sqrt{W_n}\ket{n}, \end{equation}
where
\[ W_n \coloneqq \frac{1}{\mathcal{N}_{\Delta,T}}\int_{(n-1)T/N}^{nT/N}dt\, W_\Delta(t) \]
and $\mathcal{N}_{\Delta,T} \coloneqq \int_{0}^T dt\, W_{\Delta}(t)$ is the normalisation factor. For convenience, $\ket{n}$ is used to denote the computational basis state corresponding to the binary representation of $n-1$, for all $n \in \{1,\dots, N\}$.  
By Eqs.~\eqref{eq:WDelta} and~\eqref{eq:wDelta},
\begin{equation} \label{eq:Werfc}
W_{\Delta}(t) = \frac{1}{2}\erfc\left(\frac{\Delta t}{\sqrt{2}}\right)
\end{equation}
for $t \geq 0$. 

By interpreting $W_\Delta(t)/\mathcal{N}_{\Delta,T}$ as a probability density function on $t \in [0,T]$, we can use the method of~\cite{Grover2002} for generating states whose coefficients (in the computational basis) correspond to efficiently integrable probability distributions. In this approach, a rotation operator is applied to each qubit, with the rotation angle specified by the state of the more significant qubits. These rotation angles are determined as follows. For each $i \in \{0, \dots, \ceil{\log N} - 1\}$, the interval $[0, 2^{\ceil{\log N}} T/N]$ is divided into $2^i$ regions of equal length. Labelling these regions by $0, \dots, 2^i - 1$, we define
\begin{equation} \label{eq:eta}
\eta_{i,k} \coloneqq \frac{1}{2\pi}\arcsin\left(\sqrt{\dfrac{\int_{(t^{(R)}_{k} - t^{(L)}_k)/2}^{t^{(R)}_{k}}dt\, W_\Delta(t)}{\int_{t^{(R)}_k}^{t^{(L)}_{k}}dt\, W_\Delta(t)}} \right),
\end{equation}
for $k \in \{0,\dots, 2^i - 1\}$, where $t_k^{(L)}$ and $t_k^{(R)}$ denote the left and right endpoints of region $k$, respectively. To account for the fact that $2^{\ceil{\log N}}$ may be greater than $N$, we simply take $W_{\Delta}(t) = 0$ for all $t > T$ for the purpose of calculating the integrals in Eq.~\eqref{eq:eta} (and set $\eta_{i,k} = 0$ if the denominator vanishes). Conditioned on the state of the more significant qubits being $\ket{k}$ for $k \in \{0, \dots, 2^i-1\}$, the rotation operator $R(\eta_{i,k}) = e^{-i2\pi \theta Y}$ is applied to qubit $i$. Explicitly, a circuit of the form of Fig.~\ref{fig:Wperf}, which depicts the $\ceil{\log N} = 4$ case, would prepare $\ket{W_{\Delta,T,N}}$. The qubits are ordered from most significant to least significant. First, $R(\eta_{0,0})$ is applied unconditionally to most significant qubit. Then, for $i = 1, \dots, \ceil{\log N} - 1$ (in ascending order), the (multiply-)controlled rotation $\sum_{k=0}^{2^i-1}\ket{k}\bra{k}\otimes R(\eta_{i,k})$ is applied to the first $i+1$ qubits. 

\begin{center}
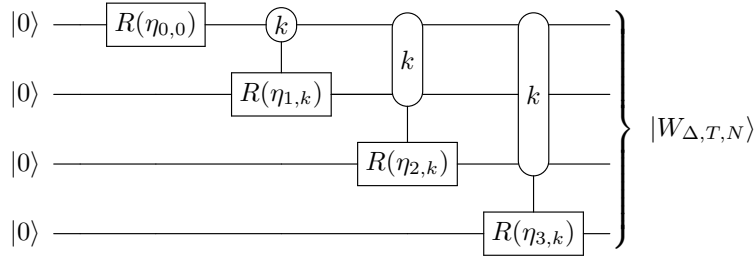

\captionsetup{type=figure}
\[
{\small\Qcircuit @C=1em @R=1em {
\lstick{\ket{0}} &\qw &\gate{R({\eta}_{0,0})} &\cbox{k} \qwx[1] &\multimeasure{1}{\hspace{-0.2em} k\hspace{-0.2em}} &\multimeasure{2}{\hspace{-0.2em} k\hspace{-0.2em}} &\qw \\
\lstick{\ket{0}} &\qw &\qw &\gate{R({\eta}_{1,k})} &\ghost{\hspace{-0.2em} k\hspace{-0.2em}} \qw \qwx[1] &\ghost{\hspace{-0.2em} k\hspace{-0.2em}} &\qw &\rstick{\raisebox{-3.4em}{$\ket{W_{\Delta,T,N}}$}}\\
\lstick{\ket{0}} &\qw &\qw &\qw &\gate{R({\eta}_{2,k})} &\ghost{\hspace{-0.2em} k\hspace{-0.2em}} \qwx[1] &\qw \\
\lstick{\ket{0}} &\qw &\qw &\qw &\qw &\gate{R({\eta}_{3,k})} &\qw
\gategroup{1}{7}{4}{7}{1em}{\}}
}}\]
\captionof{figure}{Circuit that prepares $\ket{W_{\Delta,T,N}}$. The multiply-controlled rotation operators are idealised in the sense that the angles $\eta_{i,k}$ are assumed to be computed with infinite precision. This assumption is relaxed in Lemma~\ref{lem:W}.} \label{fig:Wperf}
\end{center}

Of course, this construction assumes infinite-precision computation of the angles $\eta_{i,k}$ (by the quantum circuit underlying each multiply-controlled rotation). In the following proposition, we upper-bound the cost of coherently computing fixed-point approximations to $\eta_{i,k}$. This determines the complexity of approximately preparing $\ket{W_{\Delta,T,N}}$, as shown in Lemma~\ref{lem:W}. 

Recall that the real numbers $\Delta$ and $T$ and the integer $N$ are fixed parameters chosen according to Corollary~\ref{corollary:parameters}. As can be seen from Corollary~\ref{corollary:parameters}, $\Delta$ and $T$ must be bounded in terms of the desired precision but are not required to be a particular exact number, so we can assume that $\Delta$ and $T$ are both fixed-point numbers (computed beforehand by a classical computer). 

\begin{proposition} \label{prop:eta}
Let $\eta_{i,k}$ be defined as in Eq.~\eqref{eq:eta}. For any $i \in \{0, \dots, \ceil{\log N} - 1\}$ and $\delta > 0$, there exists a quantum circuit that maps
\[ \ket{k}\ket{0} \mapsto \ket{k}\ket{\widetilde{\eta}_{i,k}} \]
for all $k \in \{0, \dots, 2^i-1\}$, where for each $k$, $\widetilde{\eta}_{i,k}$ is a $b$-bit integer, with $b \in \mathcal{O}(\log(1/\delta))$, such that
\[ \left|\eta_{i,k} - \widetilde{\eta}_{i,k}/2^b\right| \leq \delta. \]
Such a circuit can be constructed using at most $\mathcal{O}([\Delta^2 T^2 + \log(N/\delta)]\mathcal{M}(\Delta^2 T^2 +\log(N/\delta)))$ elementary gates and $\mathcal{O}([\Delta^2 T^2 + \log(N/\delta)]^2)$ ancilla qubits.
\end{proposition}

\begin{proof} It follows from Eq.~\eqref{eq:Werfc} that the integrals of $W_\Delta(t)$ in Eq.~\eqref{eq:eta} can be evaluated as
\begin{equation} \label{eq:antiderivative} \int_{t_1}^{t_2} dt\, W_\Delta(t) = \frac{1}{2}\left.\left(t\erfc\left(\frac{\Delta t}{\sqrt{2}}\right) + \sqrt{\frac{2}{\pi}}\frac{1}{\Delta}e^{-\Delta^2t^2/2} \right)\right|_{t_1}^{t_2}. \end{equation}
We use Propositions~\ref{prop:gaussian} and~\ref{prop:erfc}, which construct quantum circuits for approximating the Gaussian function and the complementary error function, respectively. Note that the limits of integration that are relevant to the computation of $\eta_{i,k}$ [cf.~Eq.~\eqref{eq:eta}] are of the form $t = (n/N)T$ for integers $n \in \{1, \dots, N\}$. Hence, for any such $t$, a fixed-point approximation $x_f$ to $x \coloneqq \Delta t/\sqrt{2}$ with $f$ fractional bits can be computed such that $|x -x_f| \leq 2^{-f}$. Then, by Propositions~\ref{prop:gaussian} and~\ref{prop:erfc}, there exist quantum circuits that compute $\mathcal{O}(\log(1/\delta_0))$-bit approximations $g_x$ and $c_x$ such that $|\mathrm{exp}(-{x_f}^2) - g_x| \leq \delta_0$ and $|\mathrm{erfc}(x_f) - c_x| \leq \delta_0$ using $\mathcal{O}(\log(\Delta T/\delta_0)\mathcal{M}(\log[\log(\Delta T)/\delta_0]))$ elementary gates and $\mathcal{O}(\log^2(\Delta T/\delta_0))$ ancillae. Using these approximations in Eq.~\eqref{eq:antiderivative} yields a $\widetilde{y}$ such that
\begin{align} \label{eq:antiderivativebound}
\left|\int_{t_1}^{t_2}dt\, W_\Delta(t) - \widetilde{y} \right| \leq \left(T + \frac{1}{\Delta}\right)\left(\delta_0 + \mathcal{O}(2^{-f})\right), 
\end{align}
which can be bounded by $\delta_1$ by taking $\delta_0 = \delta_1/[2(T + 1/\Delta)]$ and $f \in \mathcal{O}(\log[(T + 1/\Delta)/\delta_1])$.  

Next, suppose that each integral in Eq.~\eqref{eq:eta} is computed to within absolute error $\delta_1$ this way. More precisely, for a fixed $i$ and $k$, let $y_0 \coloneqq \int_{(t^{(R)}_{k} - t^{(L)}_k)/2}^{t^{(R)}_{k}}dt\, W_\Delta(t)$ and $y_1 \coloneqq \int_{t^{(R)}_k}^{t^{(L)}_{k}}dt\, W_\Delta(t)$, so that $\eta_{i,k} = \arcsin(\sqrt{y_0/y_1})/(2\pi)$, and let $\widetilde{y}_0$ and $\widetilde{y}_1$ denote fixed-point numbers such that $|y_0 - \widetilde{y}_0| \leq \delta_1$ and $|y_1 - \widetilde{y}_1| \leq \delta_1$. The round-off error in computing $\widetilde{y}_0/\widetilde{y}_1$ can also be suppressed to $\mathcal{O}(\delta_1)$ by using $\mathcal{O}(\log(1/\delta_1))$ fractional bits. The error $|y_0/y_1 - \widetilde{y}_0/\widetilde{y}_1|\in \mathcal{O}(\delta_1/y_1)$ depends on both $\delta_1$ and the size of the denominator $y_1$. Observing that the smallest possible value of $y_1$ corresponds to an integral over a region of length at least $2^{\ceil{\log N}}T/N \geq T/N$, it follows from the lower bound 
\[ \erfc(x) \geq \frac{2}{\sqrt{\pi}}\frac{e^{-x^2}}{x + \sqrt{x^2 + 2}} \]
for $x > 0$ and the fact that $W_\Delta(t) \propto \erfc(\Delta t/\sqrt{2})$ is monotonically decreasing that 
\[ y_1 \geq \left. \frac{T}{N} \frac{2}{\sqrt{\pi}} \frac{e^{-x^2}}{x + \sqrt{x^2 + 2}}\right|_{x =\Delta T /\sqrt{2}} \in \Omega\left(\frac{e^{-\Delta^2 T^2/2}}{N \Delta}\right).\]
Thus, $|y_0/y_1 - \widetilde{y}_0/\widetilde{y}_1| \in \mathcal{O}(N\Delta e^{\Delta^2 T^2/2}\delta_1)$. By Proposition~\ref{prop:arcsin}, we can construct a quantum circuit that computes a $b$-bit approximation $\widetilde{\eta}_{i,k}/2^b$ to $\arcsin(\sqrt{\widetilde{y}_0/\widetilde{y}_1})/(2\pi)$ with absolute error $\delta/2$, with $b \in \mathcal{O}(\log(1/\delta))$. Since $|\mathrm{arcsin}(\sqrt{y}) - \arcsin(\sqrt{\widetilde{y}})| = \mathcal{O}(\sqrt{|y - \widetilde{y}|})$ for any $y, \widetilde{y} \geq 0$, we have
\begin{align*}
\left|\eta_{i,k} - \widetilde{\eta}_{i,k}/2^b\right| &= \left|\frac{1}{2\pi}\arcsin(\sqrt{y_0/y_1}) - \frac{1}{2\pi}\arcsin(\sqrt{\widetilde{y}_0/\widetilde{y}_1})\right| + \left|\frac{1}{2\pi}\arcsin(\sqrt{\widetilde{y}_0/\widetilde{y}_1}) - \widetilde{\eta}_{i,k}/2^b \right| \\
&= \mathcal{O}(N\Delta e^{\Delta^2 T^2/2}\delta_1) + \delta/2.
\end{align*}
This error is bounded by $\delta$ for some $\delta_1 \in \Theta(\delta e^{-\Delta^2 T^2/2}/(N\Delta))$. Eq.~\eqref{eq:antiderivativebound} then implies that it suffices to compute both the Gaussian function and the complementary error function to within absolute error $\delta_0 \in \Theta(\delta e^{-\Delta^2 T^2/2}/(\Delta TN)).$

The gate and space complexity of the overall circuit is clearly dominated by that of the subroutines that compute the special functions, which are given by Propositions~\ref{prop:arcsin}, \ref{prop:gaussian}, and~\ref{prop:erfc}. The circuits for the Gaussian function and the complementary error function use \[ \mathcal{O}(\log(\Delta T/\delta_0)\mathcal{M}(\log[\log(\Delta T)/\delta_0])) = \mathcal{O}([\Delta^2 T^2 + \log(N/\delta)]\mathcal{M}(\Delta^2 T^2 + \log(N/\delta))) \] elementary gates and $\mathcal{O}([\Delta^2 T^2 + \log(N/\delta)]^2)$ ancillae by Propositions~\ref{prop:gaussian} and~\ref{prop:erfc}. The circuit for approximating arcsine uses $\mathcal{O}(\log(1/\delta)\mathcal{M}(\log(1/\delta)))$ gates and $\mathcal{O}(\log^2(1/\delta))$ ancillae by Proposition~\ref{prop:arcsin}. The total number of elementary gates is therefore $\mathcal{O}([\Delta^2 T^2 + \log(N/\delta)]\mathcal{M}(\Delta^2 T^2 +\log(N/\delta)))$ and the total number of ancillae required is $\mathcal{O}([\Delta^2 T^2 + \log(N/\delta)]^2)$. 
\end{proof}

By using the circuit of Proposition~\ref{prop:eta} in conjunction with the multiply-controlled rotation $\sel {R}$ of Proposition~\ref{prop:selR}, we can construct a finite-precision unitary that prepares a state arbitrarily close to $\ket{W_{\Delta,T,N}}$.

\begin{lemma}
\label{lem:W}
Let $\ket{W_{\Delta,T,N}}$ be the $\ceil{\log N}$-qubit state defined by Eq.~\eqref{eq:Wstate}. For any $\epsilon_2 > 0$, a $\ceil{\log N}$-qubit unitary $\widetilde{W}_{\Delta,T,N}$ for which
\[ \left\|\ket{W_{\Delta,T,N}} - \widetilde{W}_{\Delta,T,N}\ket{0}\right\| \leq \epsilon_2 \]
can be implemented using at most $\mathcal{O}([\Delta^2 T^2 + \log(N/\epsilon_2)]\mathcal{M}(\Delta^2 T^2 +\log(N/\epsilon_2))\log N)$ elementary gates and $\mathcal{O}([\Delta^2 T^2 + \log(N/\epsilon_2)]^2)$ ancilla qubits.
\end{lemma}

\begin{proof}
The unitary $\widetilde{W}_{\Delta,T,N}$ can be implemented by a circuit of the same form as that in Fig~\ref{fig:Wtilde}, which is drawn for the example of $\ceil{\log N} = 4$ (the generalisation to arbitrary $N$ is straightforward). First, ${R}({\eta}_{0,0})$ is applied to the most significant qubit. Then, for each $i \in \{1,\dots, \ceil{\log N} - 1\}$, a $b$-bit integer approximation $\widetilde{\eta}_{i,k}$ to $2^b\eta_{i,k}$ is computed and stored in the $b$-qubit ancilla register (corresponding to the bottom wire in Fig.~\ref{fig:Wtilde}), conditioned the first $i$ qubits being in the state $\ket{k}$, for $k \in \{0, \dots, 2^i - 1\}$. The multiply-controlled rotation $\sel {R}$ constructed in Proposition~\ref{prop:selR} is subsequently applied, with the ancilla register as the control register and the $(i+1)$-th qubit as the target. Finally, the ancilla register is uncomputed. This sequence of operations effectively implements $\sum_{k=0}^{2^i-1}\ket{k}\bra{k} \otimes {R}(\widetilde{\eta}_{i,k})$ on the first $i+1$ qubits, for each $i \in \{1,\dots, \ceil{\log N}-1\}$, while the ancilla register is returned to its initial state $\ket{0}$. 

\begin{center}
\captionsetup{type=figure}
\hspace{-3.6em} {\small\Qcircuit @C=0.4em @R=1em {
&\lstick{\ket{0}} &\qw &\gate{{R}({\eta}_{0,0})} &\cbox{k} \qwx[4] &\qw &\cbox{k} \qwx[4] &\multimeasure{1}{\hspace{-0.2em} k\hspace{-0.2em}} &\qw &\multimeasure{1}{\hspace{-0.2em} k\hspace{-0.2em}} &\multimeasure{2}{\hspace{-0.2em} k\hspace{-0.2em}} &\qw &\multimeasure{2}{\hspace{-0.2em} k\hspace{-0.2em}} &\qw \\
&\lstick{\ket{0}} &\qw &\qw &\qw &\gate{{R}(\widetilde{\eta}_{1,k}/2^b)} &\qw &\ghost{\hspace{-0.2em} k\hspace{-0.2em}} \qwx[3] &\qw &\ghost{\hspace{-0.2em} k\hspace{-0.2em}} \qwx[3] &\ghost{\hspace{-0.2em} k\hspace{-0.2em}} &\qw &\ghost{\hspace{-0.2em} k\hspace{-0.2em}} &\qw &\rstick{\raisebox{-3.4em}{$\widetilde{W}_{\Delta,T,N}\ket{0}$}}\\
&\lstick{\ket{0}} &\qw &\qw &\qw &\qw &\qw &\qw &\gate{{R}(\widetilde{\eta}_{2,k}/2^b)} &\qw &\ghost{\hspace{-0.2em} k\hspace{-0.2em}} \qwx[2] &\qw &\ghost{\hspace{-0.2em} k\hspace{-0.2em}} \qwx[2] &\qw \\
&\lstick{\ket{0}} &\qw &\qw &\qw &\qw &\qw &\qw &\qw &\qw &\qw &\gate{{R}(\widetilde{\eta}_{3,k}/2^b)} &\qw &\qw \\
&\lstick{\ket{0}} &/\qw &\qw &\gate{\widetilde{\eta}_{1,k}} &\cbox{\widetilde{\eta}_{1,k}} \qwx[-3] &\gate{\text{$\widetilde{\eta}_{1,k}$}^\dagger} &\gate{\widetilde{\eta}_{2,k}} &\cbox{\widetilde{\eta}_{2,k}} \qwx[-2] &\gate{\text{$\widetilde{\eta}_{2,k}$}^\dagger} &\gate{\widetilde{\eta}_{3,k}} &\cbox{\widetilde{\eta}_{3,k}} \qwx[-1] &\gate{\text{$\widetilde{\eta}_{3,k}$}^\dagger} &\qw &\rstick{\ket{0}}
\gategroup{1}{14}{4}{14}{1em}{\}}
}}

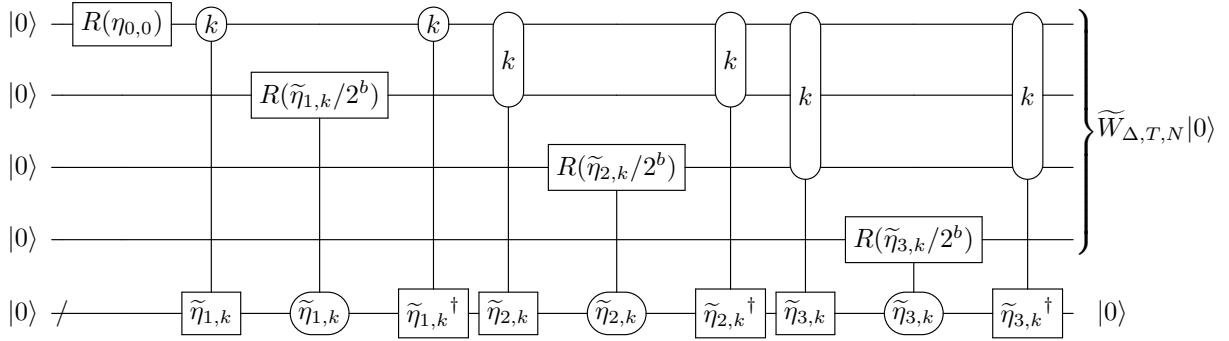
\captionof{figure}{Circuit that implements the unitary $\widetilde{W}_{\Delta,T,N}$ considered in Lemma~\ref{lem:W}.} \label{fig:Wtilde}
\end{center}

Comparing to the idealised circuit of Fig.~\ref{fig:Wperf}, which prepares $\ket{W_{\Delta,T,N}}$, it is clear that
\begin{align*}
\left\|\ket{W_{\Delta,T,N}} - \widetilde{W}_{\Delta,T,N}\ket{0}\right\| &\leq \sum_{i=1}^{\ceil{\log N} - 1} \max_{k}\left\|R(\eta_{i,k}) - {R}(\widetilde{\eta}_{i,k}/2^b)\right\| \\
&\leq \ceil{\log N} \max_{i,k}\left(2\pi\left|\eta_{i,k} - \widetilde{\eta}_{i,k}/2^b\right|\right),
\end{align*}
which is at most $\epsilon_2$ if each $\eta_{i,k}$ is approximated to within absolute error $\epsilon_2/(2\pi\ceil{\log N})$. This determines the complexity of the operations that coherently compute $\widetilde{\eta}_{i,k}$ as well as the number of qubits $b$ required for the ancilla register.
By taking $\delta = \epsilon_2/(2\pi\ceil{\log N})$ in Proposition~\ref{prop:eta}, it follows that $b = \mathcal{O}(\log[(\log N)/\epsilon_2])$ and that each of the operations mapping $\ket{k}\ket{0} \mapsto \ket{k}\ket{\widetilde{\eta}_{i,k}}$ (and their inverses) can be constructed such that
\[ \left|\eta_{i,k} - \widetilde{\eta}_{i,k}\right| \leq \frac{\epsilon_2}{2\pi\ceil{\log N}}\] using at most $\mathcal{O}([\Delta^2 T^2 + \log(N/\epsilon_2)]\mathcal{M}(\Delta^2 T^2 +\log(N/\epsilon_2)))$ elementary gates and $\mathcal{O}([\Delta^2 T^2 + \log(N/\epsilon_2)]^2)$ additional ancilla qubits. There are $\mathcal{O}(\log N)$ of these operations in the circuit for $\widetilde{W}_{\Delta,T,N}$, for a total of $\mathcal{O}([\Delta^2 T^2 + \log(N/\epsilon_2)]\mathcal{M}(\Delta^2 T^2 +\log(N/\epsilon_2))\log N)$ elementary gates. This dominates the gate complexity $\mathcal{O}(b) = \mathcal{O}(\log[(\log N)/\epsilon_2])$ of the multiply-controlled rotations [cf.~Proposition~\ref{prop:selR}]. The ancillae are always reset to their initial state and can be reused in all of the operations. Thus, the total number of ancillae required is $\mathcal{O}([\Delta^2 T^2 + \log(N/\epsilon_2)]^2)$. 

\end{proof}

\bibliographystyle{myhamsplain2}
\bibliography{bib}

\end{document}